\documentclass[11pt,letterpaper]{article}

\usepackage[utf8]{inputenc}
\usepackage[T1]{fontenc}
\usepackage{bold-extra}
\usepackage{lmodern}

\usepackage[american]{babel}

\usepackage{xcolor}

\usepackage[margin=1in]{geometry}

\usepackage{amsmath}
\usepackage{amssymb}
\usepackage{amsthm}
\usepackage{amsfonts}
\usepackage{mathtools}
\usepackage{enumerate}
\usepackage{complexity}
\usepackage{thmtools,thm-restate}

\usepackage{algorithm}
\usepackage{algpseudocode}

\usepackage{csquotes}
\usepackage{microtype}

\usepackage{xspace}

\usepackage[backend=biber,minbibnames=3,maxbibnames=10,sortcites,maxcitenames=2]{biblatex}
\usepackage{hyperref}
\usepackage[capitalize]{cleveref}

\usepackage{tikz}
\usepackage{subcaption}
\usetikzlibrary{shapes.geometric}
\usetikzlibrary{decorations,positioning}
\usetikzlibrary{backgrounds}

\hypersetup{
    colorlinks=true,
    linkcolor=black,
    citecolor=black,
    filecolor=black,
    urlcolor=black,
}

\newtheorem{theorem}{Theorem}[section]
\newtheorem*{theorem*}{Theorem}
\newtheorem{lemma}[theorem]{Lemma}

\newtheorem{corollary}[theorem]{Corollary}

\allowdisplaybreaks

\theoremstyle{definition}
\newtheorem{definition}[theorem]{Definition}

\theoremstyle{remark}
\newtheorem{remark}{Remark} \DeclarePairedDelimiter\abs{\lvert}{\rvert} \DeclarePairedDelimiter\ceil{\lceil}{\rceil} \DeclarePairedDelimiter\floor{\lfloor}{\rfloor} \DeclarePairedDelimiter\norm{\lVert}{\rVert}

\renewcommand{\AA}{\mathcal{A}}

\let\EE\relax  \newcommand{\EE}{\mathcal{E}}
\newcommand{\FF}{\mathcal{F}}

\newcommand{\II}{\mathcal{I}}
\newcommand{\LL}{\mathcal{L}}

\newcommand{\NN}{\mathcal{N}}
\newcommand{\OO}{\mathcal{O}}
\let\PP\relax  \newcommand{\PP}{\mathcal{P}}

\renewcommand{\SS}{\mathcal{S}}

\newcommand{\VV}{\mathcal{V}}
\newcommand{\XX}{\mathcal{X}}

\newcommand{\xx}{\mathbf{x}}
\newcommand{\yy}{\mathbf{y}}

\newcommand{\eps}{\varepsilon}

\newcommand{\real}{\mathbb{R}}
\newcommand{\nat}{\mathbb{N}}
\newcommand{\natPos}{\nat_+}
\newcommand{\st}{\middle|} \DeclareMathOperator{\dist}{dist}
\DeclareMathOperator{\del}{del}
\mathchardef\mhyph="2D

\newcommand{\pr}[1]{\Pr\left[#1\right]}

\newcommand{\myO}[1]{\OO\!\left(#1\right)}
\newcommand{\myOmega}[1]{\Omega\!\left(#1\right)}
\newcommand{\myTheta}[1]{\Theta\!\left(#1\right)}

\newcommand{\myOTilde}[1]{\tilde{\OO}\!\left(#1\right)}

\newcommand{\modelnamemath}[1]{\operatorname{\mathsf{#1}}}
\newcommand{\modelname}[1]{$\modelnamemath{#1}$}
\newcommand{\congest}{\modelname{CONGEST}\xspace}

\newcommand{\local}{\modelname{LOCAL}\xspace}
\newcommand{\detlcl}{\modelname{det-LOCAL}\xspace}
\newcommand{\detcgt}{\modelname{det-CONGEST}\xspace}
\newcommand{\randlcl}{\modelname{rand-LOCAL}\xspace}

\newcommand{\qlcl}{\modelname{quantum-LOCAL}\xspace}
\newcommand{\philcl}{$\varphi$\modelname{-LOCAL}\xspace}
\newcommand{\nslcl}{\modelname{NS-LOCAL}\xspace}
\newcommand{\nslclFull}{non-signaling \local}

\renewcommand{\det}{{\modelnamemath{det}}}
\newcommand{\rand}{{\modelnamemath{rand}}}

\newcommand{\cluster}{{\normalfont\texttt{cluster}}\xspace}

\newcommand{\leader}{\mathtt{leader}}

\newcommand{\problem}{\PP} \newcommand{\outcome}{\mathrm{O}} \newcommand{\lclCHR}[1]{\mathcal{L}\mathcal{X}_{#1}} \newcommand{\rjoin}[1]{\star_{#1}}

\DeclareMathOperator{\diamop}{diam}
\newcommand{\diam}[1]{\diamop\!\left(#1\right)}

\DeclareMathOperator{\girthop}{girth}
\newcommand{\girth}[1]{\girthop\!\left(#1\right)}
\DeclareMathOperator{\projop}{pr}
\newcommand{\proj}[1]{\projop_{#1}}   \newcommand{\inptData}{x} \newcommand{\id}{\text{id}} \newcommand{\lbl}{\lambda} \newcommand{\inLabels}{\Sigma_{\text{in}}} \newcommand{\inLabel}{\lbl_{\text{in}}} \newcommand{\inLblRes}{\bar{\lbl}_{\text{in}}}  \newcommand{\outLabels}{\Sigma_{\text{out}}} \newcommand{\outLabel}{\lbl_{\text{out}}} \newcommand{\indOutLabel}[1]{\lbl_{\text{(out,#1)}}} \newcommand{\resIndOutLabel}[1]{\bar{\lbl}_{\text{(out,#1)}}} \newcommand{\outLblRes}{\bar{\lbl}_{\text{out}}} \newcommand{\outPr}{p} \newcommand{\lvlCR}{t} \newcommand{\eulerT}[1]{\varphi(#1)} \newcommand{\nc}[1]{NC\!\left(#1\right)}  \newcommand{\gridW}{n_1} \newcommand{\gridH}{n_2}  
\AtEveryBibitem{
  \ifentrytype{article}{
    \clearlist{language}
    \clearfield{eprint}
    \clearfield{eprinttype}
    \clearfield{isbn}
    \clearfield{issn}
    \clearfield{month}
    \clearfield{url}
  }{}
  \ifentrytype{book}{
    \clearlist{address}
    \clearlist{location}
    \clearlist{language}
    \clearfield{url}
  }{}
  \ifentrytype{inbook}{
    \clearfield{isbn}
    \clearfield{url}
  }{}
  \ifentrytype{inproceedings}{
    \clearfield{isbn}
    \clearfield{issn}
    \clearfield{url}
  }{}
}

\begin{document}

\newcommand{\myemail}[1]{\,$\cdot$\, {\small #1}}
\newcommand{\myaff}[1]{\,$\cdot$\, {\small #1}\par\medskip}

\newenvironment{myabstract}
{\list{}{\listparindent 1.5em
        \itemindent    \listparindent
        \leftmargin    0cm
        \rightmargin   0cm
        \parsep        0pt}\item\relax}
{\endlist}

\newenvironment{mycover}
{\list{}{\listparindent 0pt
        \itemindent    \listparindent
        \leftmargin    0cm
        \rightmargin   1.5cm
        \parsep        0pt}\raggedright
    \item\relax}
{\endlist}

\begin{mycover}
{\huge\bfseries\boldmath No distributed quantum advantage for approximate graph coloring \par}
\bigskip
\bigskip

\textbf{Xavier Coiteux-Roy}
\myaff{School of Computation, Information and Technology, Technical University of Munich, Germany \,$\cdot$\, Munich Center for Quantum Science and Technology, Germany}

\textbf{Francesco d'Amore}
\myaff{Aalto University, Finland}

\textbf{Rishikesh Gajjala}
\myaff{Aalto University, Finland \,$\cdot$\, Indian Institute of Science, India}

\textbf{Fabian Kuhn}
\myaff{University of Freiburg, Germany}

\textbf{Fran{\c c}ois Le Gall}
\myaff{Nagoya University, Japan}

\textbf{Henrik Lievonen}
\myaff{Aalto University, Finland}

\textbf{Augusto Modanese}
\myaff{Aalto University, Finland}

\textbf{Marc-Olivier Renou}
\myaff{Inria, Université Paris-Saclay, Palaiseau, France \,$\cdot$\, CPHT, Ecole Polytechnique, Institut Polytechnique de Paris, Palaiseau, France}

\textbf{Gustav Schmid}
\myaff{University of Freiburg, Germany}

\textbf{Jukka Suomela}
\myaff{Aalto University, Finland}

\bigskip
\end{mycover}

\begin{myabstract}
\noindent\textbf{Abstract.}
We give an almost complete characterization of the hardness of $c$-coloring $\chi$-chromatic graphs with distributed algorithms, for a wide range of models of distributed computing.
In particular, we show that these problems do not admit any distributed quantum advantage. To do that:
\begin{enumerate}
    \item We give a new distributed algorithm that finds a $c$-coloring in $\chi$-chromatic graphs in $\tilde{\mathcal{O}}(n^{\frac{1}{\alpha}})$ rounds, with $\alpha = \bigl\lfloor\frac{c-1}{\chi - 1}\bigr\rfloor$.
    \item We prove that any distributed algorithm for this problem requires $\Omega(n^{\frac{1}{\alpha}})$ rounds.
\end{enumerate}
Our upper bound holds in the classical, deterministic $\mathsf{LOCAL}$ model, while the near-matching lower bound holds in the \emph{non-signaling} model.
This model, introduced by Arfaoui and Fraigniaud in 2014, captures all models of distributed graph algorithms that obey physical causality; this includes not only classical deterministic $\mathsf{LOCAL}$ and randomized $\mathsf{LOCAL}$ but also $\mathsf{quantum}$-$\mathsf{LOCAL}$, even with a pre-shared quantum state.

We also show that similar arguments can be used to prove that, e.g., 3-coloring 2-dimensional grids or $c$-coloring trees remain hard problems even for the non-signaling model, and in particular do not admit any quantum advantage.
Our lower-bound arguments are purely graph-theoretic at heart; no background on quantum information theory is needed to establish the proofs.
\end{myabstract}
 \thispagestyle{empty}
\newpage

\pagenumbering{arabic}

\section{Introduction}

In this work, we settle the distributed computational complexity of approximate graph coloring, for deterministic, randomized, and quantum versions of the \local model of distributed computing.

In brief, the setting is this: We have an input graph $G$ with $n$ nodes. Each node is a computer and each edge represents a communication link. Computation proceeds in synchronous rounds: each node sends a message to each of its neighbors, receives a message from each of its neighbors, and updates its own state. After $T$ rounds, each node has to stop and announce its own output, and the outputs have to form a proper $c$-coloring of the input graph $G$. If the chromatic number of $G$ is $\chi$, in this setting it is trivial to find a $\chi$-coloring in $T = \OO(n)$ rounds, as in $\OO(n)$ rounds all nodes can learn the full topology of their own connected component and they can locally find an optimal coloring by brute force without any further communication. But the key questions are: How well can we color graphs in $T \ll n$ rounds? And how much does it help if we use quantum computers that can exchange quantum information, possibly with a pre-shared entangled state?

\subsection{Main result}

We show that for all constants $c$, $\chi$, and $\alpha$, it is possible to find a $c$-coloring of a $\chi$-colorable graph in $T = \tilde{\OO}(n^{1/\alpha})$ communication rounds if and only if
\[
  \alpha \le \floor*{\frac{c-1}{\chi - 1}}.
\]
For example, if the graph is bipartite ($\chi = 2$), this means that the complexity of $2$-coloring is $\tilde{\Theta}(n)$ rounds, $3$-coloring is $\tilde{\Theta}(\sqrt{n})$ rounds, and $4$-coloring is $\tilde{\Theta}(n^{1/3})$ rounds. Here we use $\tilde{\OO}$ and $\tilde{\Theta}$ to hide polylogarithmic factors, that is, our results are tight up to polylogarithmic factors.

Perhaps the biggest surprise is that this result holds for a wide range of models of distributed computing: the answer is the same for deterministic, randomized, and quantum versions of the \local model, and it holds even if the algorithm has access to shared randomness or pre-shared quantum state (as long as the quantum state is prepared before we reveal the structure of graph $G$).

In particular, we show that there is \emph{no distributed quantum advantage for approximate graph coloring} in the context of the \local model, at least up to polylogarithmic factors.

\subsection{Significance and motivation}

Our work is directly linked to two lines of research: understanding the quantum advantage in distributed settings, and the complexity of distributed graph coloring in classical settings.

\paragraph{Distributed quantum advantage.}

There is a long line of work \cite{elkin2014,legall2018,magniez2022,wu2022,wang2022,izumi2019,censorhillel2022,izumi2020,apeldoorn2022} on quantum advantage in the \congest model---this is a bandwidth-limited version of the \local model. However, much less is known about quantum advantage in the \local model.

Earlier work by \textcite{gavoille2009} and \textcite{arfaoui2014} on \qlcl brought primarily bad news: they showed that many classical \local model lower bounds still hold in the \qlcl model. The quantum advantage demonstrated by \cite{gavoille2009} was limited to constant factors or required pre-shared quantum resources. The major breakthrough was the recent work by \textcite{legall2019} that demonstrated that there is a problem that can be solved in only $2$ rounds using quantum communication, whereas solving it in the classical setting requires \(\myOmega{n}\) rounds.

However, the problem from \textcite{legall2019} is very different from the classical problems commonly studied in the field of distributed graph algorithms, and most importantly, it is not a \emph{locally checkable} problem. Locally checkable problems are graph problems in which the task is to find a feasible solution subject to local constraints---perhaps the best-known example of such a problem is graph coloring. A lot of recent work on the classical \local model has focused on locally checkable problems, and there is nowadays a solid understanding of the landscape of the distributed computational complexity of such problems for the classical models---see, e.g., \cite{balliu18lcl-complexity,balliu18almost-global,Chang2019,Ghaffari2018,balliu20lcl-randomness,fischer17sublogarithmic,Rozhon2019,brandt16lll,chang16exponential,ghaffari17distributed,balliu19mm,Brandt2019automatic}. However, what is wide open is how \qlcl changes the picture.

A major open problem is whether there is \emph{any} locally checkable graph problem that can be solved asymptotically faster in \qlcl in comparison with the classical randomized \local model, and it has been conjectured that no such problem exists \cite{suomela2023open}. In this work we provide more
evidence in support of this conjecture: we show that various problems related to
graph coloring do not admit any significant quantum advantage.

\paragraph{Hardness of distributed coloring.}

In a very recent work, \textcite{akbari23online-local} studied the notion of locality in three different settings: distributed, dynamic, and online graph algorithms. They showed that for locally checkable problems in rooted regular trees the three notions of locality coincide, but more generally the notions are distinct. The prime example of a problem that separates the models is $3$-coloring bipartite graphs: the distributed locality (i.e., round complexity) of the problem is $\Omega(\sqrt{n})$ \cite{brandt2017}, but the online locality is $\OO(\log n)$ \cite{akbari23online-local}. While this demonstrates that there is large gap between distributed and online settings, this also highlights a blind spot in our understanding of seemingly elementary questions in the classical \local model: What, exactly, is the distributed complexity of $3$-coloring bipartite graphs? Can we solve it in $\tilde{\OO}(\sqrt{n})$ rounds? And, more generally, what is the distributed complexity of $c$-coloring $\chi$-colorable graphs?

Given the prominent role graph coloring plays in distributed graph algorithms, the state of the art is highly unsatisfactory---the upper and lower bounds are far from each other, even if we consider the seemingly elementary question of coloring bipartite graphs:
\begin{itemize}
  \item As mentioned above, the complexity of $3$-coloring bipartite graphs is known to be somewhere between $\Omega(\sqrt{n})$ and $\OO(n)$. \textcite{brandt2017} show that $3$-coloring $2$-dimensional grids requires $\Omega(\sqrt{n})$ rounds, and even though they study toroidal grids (which are not necessarily bipartite), the same result can be adapted to also show that $3$-coloring bipartite graphs requires $\Omega(\sqrt{n})$ rounds. It is not known if this is tight; to the best of our knowledge, there is no upper bound other than the trivial $\OO(n)$-round algorithm.
  \item The complexity of $4$-coloring bipartite graphs is only known to be somewhere between $\Omega(\log{n})$ and $\OO(n)$. \Citeauthor{linial1992}'s \cite{linial1992} lower bound for coloring trees applies, so we know that the complexity has to be at least $\Omega(\log n)$, but beyond that very little is known. The lower bound construction from \cite{brandt2017} cannot be used here since it is easy to $4$-color grids. To come up with a nontrivial upper bound, it would be tempting to use network decompositions in the spirit of \textcite{barenboim2012}, but we are lacking network decomposition algorithms with suitable parameters, and in any case this approach cannot produce $4$-colorings or $5$-colorings in $o(\sqrt{n})$ rounds.
\end{itemize}
In this work we solve all these open questions, up to polylogarithmic factors, for the general task of $c$-coloring $\chi$-colorable graphs. We show that there is plenty of room for improvement in both upper and lower bounds. For example, in the case of $4$-coloring bipartite graphs, the right bound turns out to be $\tilde{\Theta}(n^{1/3})$, which is far from what can be achieved with the state of the art outlined above.

\subsection{Contributions in more detail}

We will now describe all of our results and contributions in more detail; we refer to \cref{sec:ideas} for an overview of the proof ideas and to \cref{sec:coloring:ub,sec:lb-technique} for the proofs.

\subsubsection{Classical upper bound (Section \ref{sec:coloring:ub})}

Let us start with the upper bound. We design new distributed algorithms with the following properties:
\begin{restatable}{theorem}{thmMainResultUB}\label{thm:upperbound_full}There exists a \detlcl algorithm $\AA_\det$ and a \randlcl algorithm $\AA_\rand$ that, given a parameter $\alpha \in \nat$,
  find a proper vertex coloring with $\alpha(\chi-1) + 1$ colors in any graph with chromatic number $\chi$, as follows:
	\begin{itemize}
		\item $\AA_\det$ runs in
		$
			\OO\bigl(n^{1/\alpha} \log^{3-1/\alpha} n\bigr) \cdot (\log\log n)^{\OO(1)}
		$
    rounds.
    \item $\AA_\rand$ runs in
		$
			\OO\bigl(n^{1/\alpha} \log^{2-1/\alpha} n\bigr)
		$
		rounds and succeeds with probability $1 - 1/\poly(n)$.
	\end{itemize}
\end{restatable}
We note that the algorithms do not need to know $\chi$; it is sufficient to know $\alpha$ and $n$. As a corollary, we can, e.g., $3$-color bipartite graphs in $\tilde{\OO}(\sqrt{n})$ rounds by setting $\alpha = 2$.

The number of colors $c = \alpha(\chi-1) + 1$ may look like a rather unnatural expression, and there does not seem to be a priori any reason to expect that this would be tight---however, as we will see, this is indeed exactly the right number.

\subsubsection{Non-signaling model}

Our main goal is to show that the algorithms in \cref{thm:upperbound_full} are optimal (up to polylogarithmic factors), not only in the classical models but also in, e.g., all reasonable variants of the \qlcl model. To this end, we work in the \emph{non-signaling model}, as defined by \textcite{arfaoui2014}; this is essentially equivalent to the \philcl model defined earlier by \textcite{gavoille2009}.

The non-signaling model is a characterization of output distributions that do not violate the \emph{no-signaling from the future} principle or, equivalently, the \emph{causality} principle \cite{dariano2017}. 
To better understand this, suppose we have some classical \randlcl algorithm
$\AA$ that runs in $T$ rounds and outputs a vertex coloring.
Let $p(G)$ be the output distribution of $\AA$ when run on a graph $G$.
The key observation is that this distribution is not arbitrary---in particular,
it must satisfy the following property:

\begin{definition}[Non-signaling distribution, informal version]\label{def:causality-intro-informal}
  The output distribution $p(G)$ of $\AA$ is non-signaling beyond distance $T$ if the following holds:
  Let \(V\) be a set of nodes of a graph \(G\) with \(\abs{V} = n\).
  Fix a subset of nodes $U \subseteq V$ and consider $p(G) {\restriction_U}$, the restriction of $p(G)$ to $U$. 
  Let $G[U,T]$ be the graph induced by the radius-$T$ neighborhood of $U$ in $G$. 
  Now modify $G$ outside $G[U,T]$ to obtain a different $n$-node graph $G'$
  while preserving $G[U,T] = G'[U,T]$. 
  Then $p(G){\restriction_U} = p(G'){\restriction_U}$.
\end{definition}

Put otherwise, changes more than distance $T$ away from $U$ cannot influence the output distribution of $U$. It is not hard to see that this holds for \detlcl and \randlcl, even if the algorithm has access to shared randomness. 
But what makes this notion particularly useful is that it is satisfied also by the \qlcl model, even with a pre-shared quantum state \cite{arfaoui2014,gavoille2009}. 
Informally, a system that violates the non-signaling property would violate causality and enable faster-than-light communication, which is something that quantum physics does not allow.

We use \nslcl to refer to the non-signaling model. We say that $\AA$ is an \nslcl algorithm that runs in $T$ rounds if it produces an output distribution that is non-signaling beyond distance~\(T\). 
We will then prove statements of the form \enquote{any \nslcl algorithm for this
problem requires at least $T$ rounds.} As a corollary, this gives a $T$-round
lower bound for \detlcl, \randlcl, and \qlcl, even if we have access to shared
randomness and pre-shared quantum states. 
This also puts limits on the existence of so-called finitely dependent colorings \cite{holroyd2016,holroyd2018}.

\subsubsection{Non-signaling lower bounds (Section \ref{sec:lb-technique})}

The precise version of our lower bound result states that, for every large
enough number of nodes $n$, there exists a \(\chi\)-chromatic graph on \(n\)
nodes that is hard to color in \nslcl.

\begin{restatable}{theorem}{thmMainResultLB}\label{thm:lb-coloring}
  Let \(\chi \ge 2\), \(c \ge \chi\) be integers, and let \(\alpha = \floor{\frac{c-1}{\chi-1}}\). Let \(\varepsilon \in (0, \frac{\alpha-1}{\alpha})\) be a real value, and let \(n \in \nat\) with 
  \[
      n \ge \ceil*{\frac{\log \varepsilon^{-1}}{\log (1 + \frac{1}{\alpha}) }}\cdot \frac{(6\chi + 1)^ {\alpha
  +1} - 1}{6}.
  \] 
  Suppose $\AA$ is an \nslcl algorithm for \(c\)-coloring graphs in the family
  \(\FF\) of \(\chi\)-chromatic graphs of \(n\) nodes with success probability
  \(q > \varepsilon\).
  Then the running time of $\AA$ is at least
  \[
      T = \Omega\Biggl(\frac{1}{\chi^{1 + \frac{1}{\alpha}}} \cdot \biggl(\frac{n}{\log \varepsilon^{-1}}\biggr)^{\frac{1}{\alpha}}\Biggr).
  \]
\end{restatable}

A key observation is that, if the parameters $\chi$, $c$, and $\varepsilon$ in
\cref{thm:lb-coloring} are constants, then \(T = \Omega(n^{\frac{1}{\alpha}})\),
which matches the upper bound in \cref{thm:upperbound_full} up to
polylogarithmic factors.
In particular, there is at best polylogarithmic room for any distributed quantum
advantage.

While \cref{thm:lb-coloring} implies bounds for coloring bipartite graphs in general, we will also use our techniques to prove bounds for specific bipartite graphs.
By prior work, it is known that $3$-coloring $2$-dimensional grids is hard in the \detlcl model \cite{brandt2017,aboulker2019}. 
We show that this also holds for the non-signaling model:

\begin{restatable}{theorem}{thmCslLclLbGrids}
\label{thm:csl-lcl-lb-grids}
    Let \(\varepsilon \in (0,\frac{3}{4})\) and \(N = \ceil{\log(\varepsilon^{-1}) / \log ( \frac{4}{3})}\). Let \(\gridW, \gridH \in \nat\) with 
    \(
        \floor{\frac{\gridW}{N}} \ge 5 \text{ and } \floor{\frac{\gridH}{N}} \ge 5.  
    \)
    Suppose $\AA$ is an \nslcl algorithm that 3-colors \(\gridW \times \gridH\) grids with probability \(q > \varepsilon\).
    Then, the running time of $\AA$ is at least
    \[
        T = \Omega\biggl(\frac{\min(\gridW, \gridH)}{\log \varepsilon^{-1}}\biggr).
    \]
\end{restatable}

This result is easiest to interpret in the case of a square grid, i.e., $\gridW = \gridH$. Then the lower bound (for constant $\varepsilon$) is simply $\Omega(\sqrt{n})$, where \(n = \gridW \cdot \gridH\), and this is trivially tight since the diameter of the grid is $\OO(\sqrt{n})$; hence the problem can be solved in $\OO(\sqrt{n})$ rounds with a \detlcl algorithm. In particular, there is no room for distributed quantum advantage (beyond possibly constant factors).

Finally, we revisit the classical result by \textcite{linial1992} about the
hardness of coloring trees.
We show that essentially the same lower bound holds in the non-signaling model:

\begin{restatable}{theorem}{thmTreesLb}
\label{thm:trees:lb}
    Let \(c \ge 2\) be an integer, and \(\varepsilon \in (0,1)\).
    Suppose $\AA$ is an \nslcl algorithm that \(c\)-colors trees of size \(n \in \nat\) with probability \(q > \varepsilon\).
    Then, for infinitely many \(n\), as long as \(\varepsilon > e^{-n}\), the running time of $\AA$ is at least
    \[
        T = \Omega\bigl(\log _ c n - \log_c \log \varepsilon^{-1}\bigr).
    \]
\end{restatable}  

\section{Key new ideas and techniques}\label{sec:ideas}

In this section, we will give an informal overview of the key new ideas and techniques that we use to prove \cref{thm:upperbound_full,thm:lb-coloring,thm:csl-lcl-lb-grids,,thm:trees:lb}. We refer to \cref{sec:coloring:ub,sec:lb-technique} for the formal proofs.

\subsection{Classical upper bound (Section \ref{sec:coloring:ub})}

\paragraph{Background and prior work.}

The only existing distributed algorithm for solving the approximate coloring
problem in general graphs that we are aware of is a folklore algorithm based on
\emph{network decompositions} \cite{awerbuch89,linialsaks93}.
For parameters $\alpha$ and $d$, an \emph{$(\alpha, d)$-network decomposition}
is a partition of the nodes $V$ of a graph $G=(V,E)$ into clusters of (weak)
diameter at most $d$ together with a proper $\alpha$-coloring of the cluster
graph;
recall that the \emph{weak diameter} of a cluster $C\subseteq V$ is the maximum
distance in $G$ between any two nodes in $C$.

Given such a decomposition, it is not hard to see how to color graphs of
chromatic number $\chi$ with $\alpha \chi$ colors in $d$ rounds by using
disjoint color palettes for clusters of different colors:
For every $i \in [\alpha]$, the nodes in a cluster of color $i$ use colors from
the palette $\{i, \alpha+i, 2\alpha+i, \dotsc\}$.
Each such cluster $C$ is colored by having a leader node collect the entire
topology of the cluster and then brute forcing an optimal coloring $\varphi_C$
of the cluster.
Since the graph has chromatic number $\chi$, coloring $\varphi_C$ uses at most $\chi$
colors.
Finally, the leader broadcasts the coloring and assigns each node $v$ in $C$ the
color $\alpha(\varphi_C(v)-1) + i$.
This results in a proper coloring with at most $\alpha\chi$ colors.
In addition, the nodes do not require knowledge of \(\chi\) in advance.
This algorithm has for example been used by \textcite{barenboim2012} to compute
a non-trivial approximate coloring in a constant number of rounds.
Our algorithm is based on two new ideas that are outlined below.

\paragraph{New ingredient 1: New network decomposition algorithms.}

Network decomposition algorithms mostly focus on optimizing the product of
$\alpha$ and $d$ (as most applications of network decompositions require time
proportional to $\alpha d$) or on minimizing the number of cluster colors for a
given maximum cluster diameter (e.g.,
\cite{awerbuch89,linialsaks93,barenboim2012,Rozhon2019,GGHIR23}).
We are interested in network decompositions with a \emph{fixed} number of
cluster colors $\alpha$ that is beyond our control and where we wish to optimize
the value of $d$.
By using existing clustering techniques
\cite{miller2013parallel,chang2023,Rozhon2019,GGHIR23} with some minor
adaptations, we give randomized and deterministic algorithms that, for any
parameter $\eps>0$, compute a set of non-adjacent clusters such that the cluster
diameter of each cluster is $\polylog(n)/\eps$ and the total number of
unclustered nodes is at most $\eps n$.
For every integer $\alpha$, this can in turn be used to compute an
$(\alpha,d)$-network decomposition with $d=\tilde{\OO}(n^{1/\alpha})$.

Let us illustrate the idea for the case $\alpha=2$. 
Setting $\eps=1/\sqrt{n}$, we compute a set of non-adjacent clusters of diameter
$\tilde{\OO}(\sqrt{n})$ such that at most $\OO(\sqrt{n})$ nodes remain
unclustered.
The constructed clusters can all be colored with color $1$ and the connected
components of the unclustered nodes form the clusters of color $2$.
We thus obtain a $(2,\tilde{\OO}(\sqrt{n}))$-network decomposition in time
$\tilde{\OO}(\sqrt{n})$ and one can therefore color graphs of chromatic number
$\chi$ with $2\chi$ colors in $\tilde{\OO}(\sqrt{n})$ rounds.

\paragraph{New ingredient 2: The hiding trick.}

In \cref{sec:coloring:ub} we show how to reduce the number of colors used while
keeping the round complexity the same.
The main new idea is what we call the \emph{hiding trick}:
First we make sure that clusters of the same color are at least four hops apart;
this can be achieved by running a network decomposition protocol on $G^3$.
For simplicity, assume that we have a $(2,\tilde{\OO}(\sqrt{n}))$-network
decomposition.
Let $C$ be a cluster of color $1$. 
We first extend $C$ to an extended cluster $C_1$ that includes all unclustered
nodes that are adjacent to $C$.
Next we find a proper $\chi$-coloring of $C_1$ using colors $\{1,\dotsc,\chi\}$
by brute force.
Finally, any boundary node of $C_1$ that has color $\chi$ is removed, thus
yielding a cluster $C_0$ with $C \subseteq C_0 \subseteq C_1$.
Note that $C_0$ is colored using at most $\chi$ colors and that and all boundary
nodes of $C_0$ have a color different from $\chi$.
We have effectively \emph{hidden the color $\chi$ inside the cluster}.
Now we continue with the rest of the process and can safely use colors
$\{\chi,\dotsc,2\chi-1\}$ to color the uncolored nodes in clusters of color $2$.
The end result is a proper $(2\chi-1)$-coloring, and the running time is simply
a constant times the cluster diameter $d = \tilde{\OO}(\sqrt{n})$.
With this strategy, for instance, we can compute a $3$-coloring of bipartite
graphs in $\tilde{\OO}(\sqrt{n})$ rounds.

\cref{thm:upperbound_full} is in essence a formalization and generalization of
the hiding trick.
We can hide one color in each cluster and therefore reuse one of the colors for
all of the $\alpha$ cluster colors.
This results in a coloring with $\alpha(\chi-1) + 1$ colors. 
In addition, if one chooses the color palettes carefully, it is not necessary
for the nodes to know $\chi$ beforehand.
At first this may seem like an ad-hoc trick that cannot possibly be
optimal---after all, we are saving only one color in each step.
However, our matching lower bound shows that the hiding trick is essentially the
best that we can do in distributed coloring, even if we have access to quantum
resources.

\subsection{Non-signaling lower bounds (Section \ref{sec:lb-technique})}\label{sec:intro:lb}

\paragraph{Prior work on classical lower bounds.}

As a warm-up, let us first see how one could prove a lower bound similar to \cref{thm:lb-coloring} for classical models. For concreteness, let us focus on the case $\chi = 2$ and $c = 3$ in the \detlcl model. Then $\alpha = 2$, and we would like to show that $3$-coloring bipartite graphs requires $\Omega(\sqrt{n})$ rounds.

Here we can make use of existential graph-theoretic arguments similar to the one
already used by \textcite{linial1992}.
Let $\AA$ be a \detlcl algorithm that purportedly finds a $3$-coloring in
bipartite graphs in $T(n) = o(\sqrt{n})$ rounds.
Now suppose that we are able to construct a graph \(G\) on \(n\) nodes with the
following properties:
\begin{enumerate}
  \item $G$ is locally bipartite: for any node $v$ of $G$, the subgraph of $G$ induced by the radius-$\Theta(\sqrt{n})$ neighborhood of $v$ is bipartite.
  \item $G$ is not $3$-colorable: the chromatic number of $G$ is at least~$4$.
\end{enumerate}
The graph \(G\) is not bipartite, but (since $\AA$ operates in the \local model)
we can apply $\AA$ to $G$ anyway and observe what happens. 
As the chromatic number of $G$ is at least $4$, certainly $\AA$ cannot $3$-color
$G$; hence there has to be at least one node $v$ such that in the local
neighborhood $X$ of $v$ the output of $\AA$ is not a valid $3$-coloring. 
However, by assumption, the local neighborhood of $v$ up to distance
$\Theta(\sqrt{n})$ is bipartite, so with some cutting and pasting we can
construct another graph $G' = (V, E')$ on the same set of nodes such that $G'$
is bipartite and the graph induced by the radius-$T(n)$ neighborhood of $X$ is the same in $G$ and
$G'$.
Hence the output of $\AA$ in $X$ is the same in both graphs, which means $\AA$
produces an invalid coloring in $G'$ (which is bipartite) in the neighborhood
$X$.
The key point in this argument is the existence of the \emph{cheating graph}
\(G\) that \enquote{fools} $\AA$ as $\AA$ cannot locally tell the difference
between $G$ and the valid input graph $G'$.

\paragraph{New ingredient 1: Bogdanov's construction.}

To apply the proof strategy outlined above, we need a suitable construction of a cheating graph. It turns out we can make direct use of \Citeauthor{bogdanov2013}'s \cite{bogdanov2013} graph-theoretic work---this is a 10-year-old result, but so far this seems to have been a blind spot in the research community, and we are not aware of any lower bounds in the context of distributed graph algorithms that make use of it.

Together with our new algorithm from \cref{thm:upperbound_full}, this then gives
a near-complete characterization of the complexity of $c$-coloring
$\chi$-colorable graphs in \detlcl. 
A similar argument applies (with some adjustments) to the \randlcl model as
well; in particular, we can exploit the independence of the output distribution
between well-separated subgraphs of the input graph to boost the failing
probability of an algorithm (see \cref{sec:lb:indistinguishability} for the
details).

\paragraph{New ingredient 2: Defining cheating graphs for \boldmath\nslcl.}

While proof techniques based on cheating graphs are commonly used in the context of \detlcl and
\randlcl, we stress the same line of reasoning \emph{does not hold} in \qlcl or
\nslcl.
In fact, in the context of \nslcl new challenges emerge, which we discuss later in this section.
Our new proof strategy  overcomes these issues: it builds on the idea of cheating graphs, but it allows us to directly derive a lower bound for \nslcl. To the best of our knowledge, this is the first work establishing that \Citeauthor{linial1992}'s argument \cite{linial1992} can be adapted to more general non-signaling models. This is one of our main technical contributions.

In \cref{sec:lb-technique} we present a new definition of a cheating graph that is applicable in \nslcl.
Suppose we are interested in a locally checkable problem \(\problem\) in graph family \(\FF\).
\begin{definition}[Cheating graph, informal version]\label{def:intro:cheating-graph}
  Consider a sufficiently large \( N > 0\).
  A graph $G$ is a cheating graph for \((\problem,\FF)\) if
  \begin{enumerate}[(1)]
    \item problem \(\problem\) is not solvable on \(G\);
    \item for a suitable $k$, we can cover $G$ with $k$ subgraphs $G^{(1)}, \dotsc, G^{(k)}$ such that \(\problem\) is solvable over each of the graphs induced by their radius-\(T(n)\) neighborhoods, where \(n = \abs{V(G)} \cdot N\);
    \item we can take any $N$ subgraphs $G^{(x_1)}, \dotsc, G^{(x_N)}$ together with their radius-\(T(n)\) neighborhoods, possibly with replacement, form their disjoint union $\tilde{G}$, and find a graph $H \in \mathcal{F}$ of size \(n\) that contains an induced subgraph isomorphic to \(\tilde{G}\).
  \end{enumerate}
\end{definition}

See \cref{def:cheating-graph} for the formal version.
We show that the existence of graphs of arbitrarily large size with the above properties directly implies
a lower bound equal to \(T\)
to the problem \(\problem\) over the graph family \(\FF\) that holds also in \nslcl---this is formalized in \cref{thm:lb-technique}.

We note that the precise requirements for \(k\) and \(N\) depend on \(T\): in \cref{sec:coloring:trees} we will exploit the fact that when $T$ is small we can afford a large $k$, while in \cref{sec:coloring:lb,sec:coloring:grids} we deal with a large $T$, and then it will be important to construct a cheating graph with a constant $k$.

Our definition of a cheating graph is somewhat technical, but through three examples we demonstrate that this is indeed an effective way of proving lower bounds.

\paragraph{New ingredient 3: New analysis of Bogdanov's construction.}

While in the classical models we could directly apply \textcite{bogdanov2013} as a black box, this is no longer the case in \nslcl. Nevertheless, in \cref{sec:coloring:lb} we show that the construction of \cite{bogdanov2013} indeed gives a cheating graph for \(c\)-coloring \(\chi\)-chromatic graph.

It is known that the graph constructed by \cite{bogdanov2013} is locally \(\chi\)-chromatic, but globally has chromatic number greater than \(c\), implying property (1) in \cref{def:intro:cheating-graph}.
We further go through the details of the construction and prove that the graph also satisfies properties (2) and (3) from \cref{def:intro:cheating-graph}---these are properties outside the scope of \cite{bogdanov2013}. 
Then from \cref{thm:lb-technique} we obtain \cref{thm:lb-coloring}.

\paragraph{New ingredient 4: New analysis of quadrangulations of the Klein bottle.}

In \cref{sec:coloring:grids} we make use of properties of quadrangulations of the Klein bottle
\cite{archdeacon2001,mohar2002,mohar2013} to construct a family of graphs that
is locally grid-like but is not $3$-colorable. 
We show that such quadrangulations give cheating graphs
for \(3\)-coloring grids, and then \cref{thm:lb-technique} implies \cref{thm:csl-lcl-lb-grids}.

\paragraph{New ingredient 5: New analysis of Ramanujan graphs.}

In \cref{sec:coloring:trees} we revisit the construction of Ramanujan graphs \cite{lubotszky1988}, that is, high-girth and high-chromatic regular graphs, which \citeauthor{linial1992} used in his lower-bound proof.
Again, we show that it provides us with a cheating graph (\cref{def:intro:cheating-graph}) for \(c\)-coloring trees, and \cref{thm:trees:lb} follows.

\paragraph{Discussion.}

While quadrangulations of the Klein bottle and Ramanujan graphs have been used in prior work to prove lower bounds for the classical models, by e.g.\ \textcite{aboulker2019,linial1992}, we remark that to our knowledge, this is the first time that the applicability of \Citeauthor{bogdanov2013}'s graph-theoretic work \cite{bogdanov2013} in the context of distributed computing and quantum computing lower bounds is recognized (in spite of it being a 10-year-old result).

We also note that the classical version of \cref{thm:csl-lcl-lb-grids} by \cite{brandt2017} uses fundamentally different proof techniques: the argument in \cite{brandt2017} is primarily \emph{algorithmic},
while our proof is primarily \emph{graph-theoretic}. 
The algorithmic proof from the prior work seems to be fundamentally incompatible with the \nslcl model, while the graph-theoretic proof also yields a lower bound for \nslcl. 
This suggests a general blueprint for lifting prior lower bounds from \detlcl or \randlcl to \nslcl: (1)~re-prove the previous result using existential graph-theoretic arguments, and (2)~apply the cheating graph idea to lift it to \nslcl. 

While the proof technique developed in this work is applicable in many graph problems, there are also problems for which cheating graphs do not exist (e.g., sinkless orientation on $2$-regular graphs). An open research direction is developing proof strategies that can be used to derive \nslcl lower bounds for those cases.

\subsection{Lower bound technique in more details}

Fix a sufficiently large \(N > 0\). Consider any locally checkable problem \(\problem\) over some graph family \(\FF\).
We want to show that, whenever a cheating graph
(\cref{def:intro:cheating-graph}) for the pair \((\problem, \FF)\) exists, 
any \(T\)-round algorithm solving the problem has failing probability at least
\(1 - (1-1/k)^N\), where \(k\) is the size of the subgraph cover of the cheating
graph.

Let \(G\) be the cheating graph.
For simplicity, we can think of \(\problem\) as the \(3\)-coloring problem, and \(\FF\) to be the family of bipartite graphs.
Provided that \(\FF\) respects some natural properties, properties (1) and (2) from \cref{def:intro:cheating-graph} ensure that we can get a lower bound in \randlcl.
Indeed, assume there is a $T$-round randomized algorithm \( \AA \) that 3-colors bipartite graphs.
Clearly, \(\AA\) fails to 3-color \(G\) with probability 1.
Hence, there is an $i^\star \in [k]$ such that the failing probability of \(\AA\) over \(G^{(i^\star)}\) is at least \(1/k\).
Hence, \(\AA\) will fail on any bipartite graph of at most \(n = \abs{V(G)} \cdot N\) nodes containing an induced subgraph isomorphic to the radius-\(T(n)\) neighborhood of \(G^{(i^\star)}\) with probability \(1/k\). 
\begin{figure}[p]
  \centering
  \begin{subcaptionblock}{\textwidth}
    \centering
    \includegraphics[scale=0.66625]{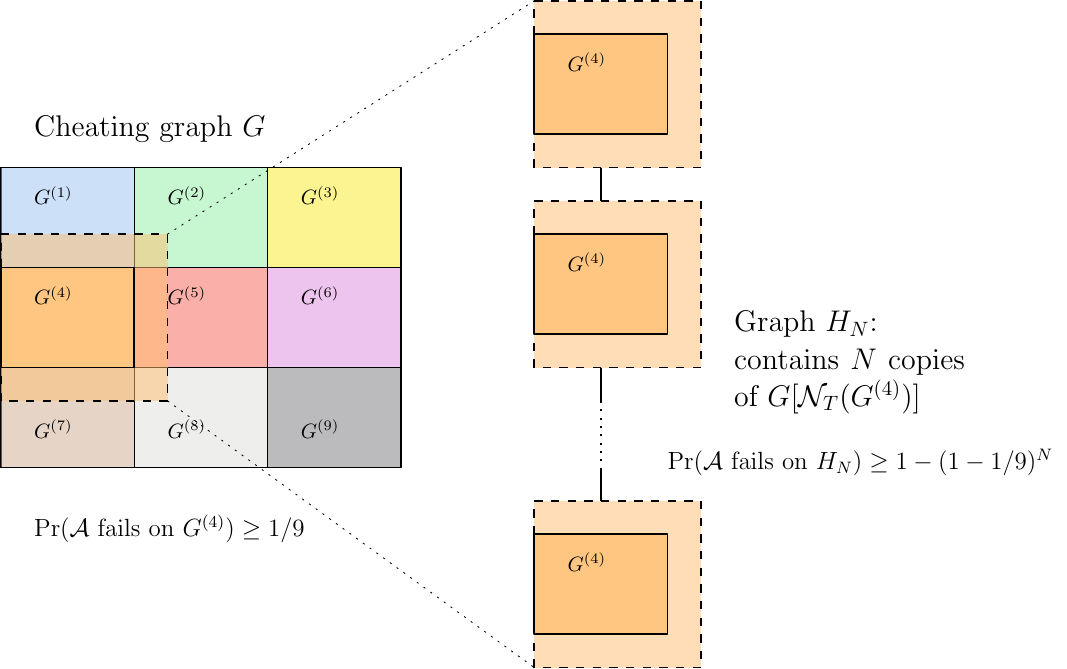}
    \caption{Construction for the \randlcl model. 
    For any \(T(n)\)-round algorithm \(\AA\) solving the problem, there is an \(i^{\star} \in [9]\) (in the figure, $i^\star = 4$) such that $\Pr[\AA \text{ fails in } G^{(i^\star)}] \ge 1/9$.  
    Then, \(\Pr[\AA \text{ fails on } H_N ] \ge 1 - (1-1/9)^N\) by independence, where \(H_N\) is an admissible instance.
    As long as \(\abs{V(H_N)} \le n\), this gives the lower bound.}
    \label{fig:lb-technique-randlcl}
  \end{subcaptionblock}
  \\[5mm]
  \begin{subcaptionblock}{\textwidth}
    \centering
    \includegraphics[scale=0.66625]{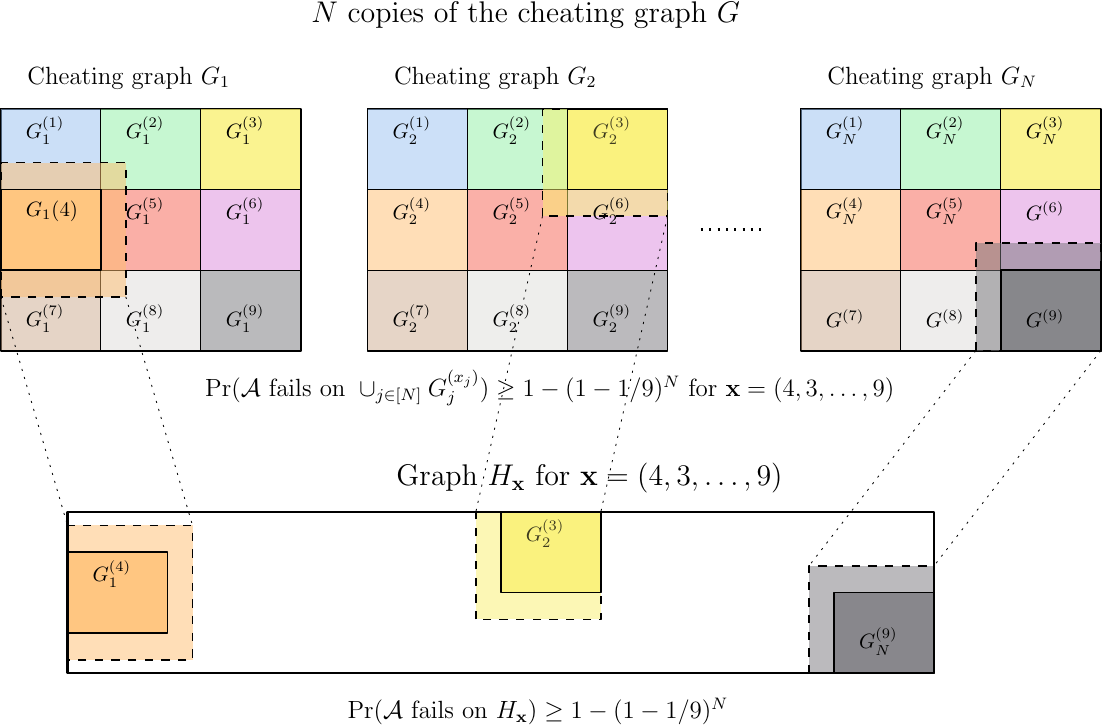}
    \caption{Construction for the \nslcl model. 
    We start with \(N\) copies \(G_1, \dots, G_N\) of \(G\) and consider their disjoint union.
    We prove that, in this specific graph, there is already a combination of indices \(\xx = (x_1, \dots, x_N) \in [9]^N\) (in the figure, \(\xx = (4,3,\dots,9)\)) for which \(\Pr[\AA \text{ fails on } \bigcup_{j \in [N]} G_j^{(x_j)} ] \ge 1 - (1-1/9)^N\).
    Then, property (3) of \cref{def:intro:cheating-graph} ensures that we can construct an admissible instance \(H_\xx\) as shown in the figure, with \(\abs{V(H_\xx)} = n\).
    By the properties of the \nslcl model, since \(H_\xx\) and \(\bigsqcup_{i \in [N]} G_i\) share the same local view around \(\bigcup_{j \in [N]} G_j^{(x_j)}\), \(\AA\) fails on \(H_\xx\) too with at least the same probability.}
    \label{fig:lb-technique-nslcl}
  \end{subcaptionblock}
\caption{Illustration of the lower-bound argument based on the cheating graph $G$.
  For \(n = \abs{V(G)} \cdot N\), the problem is solvable in each \(T(n)\)-radius neighborhood of \(G^{(i)}\), \(i \in [9]\), but not on \(G\).}
  \label{fig:lb-technique}
\end{figure}
If \(k\) is not small enough (e.g., \(k = w(1)\)), the failing probability tends to $0$.
In \randlcl we can amplify the failure probability as follows:
Suppose \(\FF\) contains a graph \(H_N\) of at most \(n = \abs{V(G)} \cdot N\) nodes that contains, as subgraphs, \(N\) disjoint copies of the radius-\(T(n)\) neighborhood of \(G^{(i^\star)}\) in \(G\). 
This is always possible if \(\FF\) is the family of bipartite graphs.
By independence, the failing probability of \(\AA\) over \(H_N\) is at least \(1 - (1 - 1/k)^N\) (see \cref{fig:lb-technique-randlcl}). 
Hence, such an algorithm cannot exist.
The property that \(\FF\) contains \(H_N\) is reasonable for many natural problems (e.g., \(c\)-coloring \(\chi\)-chromatic graphs for all combinations of \(c\) and \(\chi\)) where, given a graph for which the problem is solvable, one can connect disjoint copies of the graphs and obtain a solvable instance of the problem.

However, as we anticipated, in the \nslcl model some issues arise:
\begin{enumerate}[(i)]
  \item If two graphs \(G\) and \(H\) have different sizes, then even if they share two identical subgraphs \(G'\) and \(H'\) with isomorphic radius-\(T\) neighborhoods, a non-signaling output distribution is not guaranteed to be identical over \(G'\) and \(H'\); this is due to the no-cloning principle \cite{wootters1982,dariano2017,masanes2006}.
  \item If we look at the output distributions for two subsets of nodes $X$ and $Y$, then even if $X$ and $Y$ are far from each other, we cannot assume that the outputs of these subsets are independent.
\end{enumerate}

Issue (i) puts some limits on the choice of the graph used to ``fool'' the algorithm, while issue (ii) makes it necessary to deal with possible dependencies among different parts of the input graph.
Such issues are the reason why we require the cheating graph to satisfy property (3) in \cref{def:intro:cheating-graph}. 

To solve issue (i), we consider a graph that is the disjoint union of \(N\) copies \(G_1, \dots,G_n\) of the cheating graph \(G\):
such graph has \(\abs{V(G)} \cdot N\) vertices, exactly the same as the graph
of property (3) from \cref{def:intro:cheating-graph}.
Consider now any \nslcl algorithm \(\AA\) that \(3\)-colors bipartite graphs with locality \(T\), and apply it to the graph \(\sqcup_{i \in [N]} G_i\).
Clearly, \(\AA\) will fail to solve the problem in each \(G_j\) with probability 1.
At this point, we cannot continue as before: while it is true that in each \(G_j\) we can find an index \(i^{\star}\) such that the probability of \(\AA\) failing in \(G^{(i^\star)}_j\) is at least \(1/k\), we cannot use independence to increase the failing probability.

Property (3) ensures that, for a sufficiently large \(N\), and for any sequence of indices \( \xx = (x_1, \dots, x_N) \in [k]^N\), there exists a graph \(H_{\xx}\) of size \(\abs{V(G)} \cdot N\) that contains an induced subgraph isomorphic to the disjoint union of the radius-\(T\) neighborhoods of \(G^{(x_1)}, \dots, G^{(x_N)}\).
However, correlations among these subgraphs of \(H_{\xx}\) might hold.
To overcome this issue, we need to consider all possible sequences of subgraphs \(G_1^{(x_1)}, \dots, G_N^{(x_N)}\) at the same time, where \(\xx = (x_1, \dots, x_N) \in [k]^N\) (see \cref{fig:lb-technique-nslcl}).
Fix a total order for the elements in \( [k]^N \), and let its ordered elements be \(\xx_1, \dots, \xx_{k^N}\).
Let \(\FF_{\xx_j}\) be the event that \(\AA\) fails in \(G^{(z_i)}_i\), where \(z_i\) is the \(i\)-th element of \(\xx_j\), for each \(i = 1, \dots, N\). 
Furthermore, for each index \(\xx \in [k]^N\), let \(\II_\xx\) be the set of all indices \(\yy \in [k]^N\) such that \(\yy\) and \(\xx\) share the same element at the \(i\)-th position, for some \(i\), i.e., \(\xx(i) = \yy(i)\). 
Notice that, for \(\xx = (x_1, \dots, x_N)\), \(\cup_{\yy \in \II_{\xx}} \FF_\yy\) describes the event that there is an \(i \in [N]\) such that \(\AA\) fails on \(G_i^{(x_i)}\).

We claim that there exists a \(\xx^\star \in [k]^N\) such that \(\pr{\cup_{\yy \in \II_{\xx^\star}}\FF_{\yy}} \ge 1 - (1-1/k)^N\), implying that the dependencies behave ``well enough'', hence solving issue (ii).
We proceed by contradiction: we assume that, for all \(\xx \in [k]^N\), \(\pr{\cup_{\yy \in \II_\xx}\FF_{\yy}} < 1 - (1-1/k)^N\).
While \(\pr{\cup_{\xx \in [k]^N} {\FF_{\xx}}} = 1\), the events in \(\{\FF_{\xx}\}_{\xx \in [k]^N}\) are not disjoint and the sum of their probability is not 1. 
To better deal with the math, we define \(\EE_{\xx_1} = \FF_{\xx_1}\) and, recursively, we define \(\EE_{\xx_j} = \FF_{\xx_j} \setminus (\cup_{i = 1}^{j-1}) \EE_{\xx_i}\).
Clearly the events in \(\{\EE_\xx\}_{\xx \in [k]^N}\) are pairwise disjoint: furthermore, it holds that \(\sum_{\xx \in [k]^N} \pr{\EE_{\xx}} = 1\) as \(\cup_{\xx \in [k]^N} {\EE_{\xx}} = \cup_{\xx \in [k]^N} {\FF_{\xx}}\).
For each \(\xx \in [k]^N\), it trivially holds that
\[
  \sum_{\yy \in \II_\xx}\pr{\EE_{\yy}} + \sum_{\yy \in [k]^N \setminus \II_\xx } \pr{\EE_{\yy}} = 1.
\]
Moreover, for each \(\xx \in [k]^N\), we have \(\pr{\EE_{\xx}} \le \pr{\FF_\xx}\) as \(\EE_{\xx} \subseteq \FF_{\xx}\), hence \(\sum_{\yy \in \II_\xx}\pr{\EE_{\yy}} < 1 - (1 - 1/k)^N\).
Thus,
\[
  \sum_{\yy \in [k]^N \setminus \II_\xx } \pr{\EE_{\yy}} > (1-1/k)^N.
\]
It follows that
\[
  \sum_{\xx \in [k]^N}\sum_{\yy \in [k]^N \setminus \II_\xx } \pr{\EE_{\yy}} > k^N (1-1/k)^N = (k-1)^N.
\]
Also, notice that for each \(\yy \in [k]^N\) the cardinality of the set \(\left\{\xx \in [k]^N \ \st \ \yy \in [k]^N \setminus \II_\xx \right\}\) is \((k-1)^N\).
Hence,
\begin{align*}
  \sum_{\xx \in [k]^N}\sum_{y \in [k]^N \setminus \II_\xx} \pr{\EE_{\yy}} & = \sum_{\yy \in [k]^N}\sum_{\substack{\xx \in [k]^N : \\ \yy \in [k]^N\setminus \II_\xx}} \pr{\EE_{\yy}}
  = (k-1)^N \sum_{\yy \in [k]^N} \pr{\EE_{\yy}}
  = (k-1)^N,
\end{align*}
reaching a contradiction.
Thus, there exists an \(\xx^{\star} \in [k]^N\)
such that
$
  \pr{\cup_{\yy \in \II_{\xx^\star}} \FF_\yy} \ge 1 - (1-1/k)^N.
$
Property (3) of \cref{def:intro:cheating-graph} ensures that there is a graph \(H_{\xx^\star} \in \FF\) of \(n\) nodes such that \(H_{\xx^\star}\) contains, as induced subgraph, \(\cup_{i \in [N]} G_{i}^{(\xx^\star_i)}\), and \(H_{\xx^\star}\) and \(\sqcup_{i \in [N]} G_i\) share the same radius-\(T(n)\) neighborhood around \(\cup_{i \in [N]} G_{i}^{(\xx^\star_i)}\). 
By the property of the \nslcl model, we get that the failing probability of \(\AA\) on \(H_{\xx^\star}\) is at least \(1 - (1-1/k)^N\).
 \section{Preliminaries}\label{sec:preliminaries}

We consider the set \(\nat\) of natural numbers to start with 0.
We also define \(\nat_+ = \nat \setminus \{0\}\). 
For any positive integer \(n \in \natPos\), we denote the set \(\{1, \dots, n\}\) by \([n]\).

\paragraph{Graphs.}

All graphs in this paper are simple graphs without
self-loops.
Let \(G = (V,E)\) be a simple graph with \(n \in \nat\) nodes.
If the set of nodes and the set of edges are not specified, we refer to them by \(V(G)\) and \(E(G)\), respectively.
For any edge \(e = \{u,v\} \in E(G)\), we also write \(e = uv = vu\).

If \(G\) is a subgraph of \(H\), we write \(G \subseteq H\).
For any subset of nodes \(A \subseteq V\), we denote by \(G[A]\) the subgraph
induced by the nodes in \(A\).
For any nodes $u,v \in V$, $\dist_G(u,v)$ denotes the distance between $u$ and
$v$ in $G$ (i.e., the number of edges of any shortest path between $u$ and $v$
in $G$); if $u$ and $v$ are disconnected, then $\dist_G(u,v) = +\infty$.
If $G$ is clear from the context, we may also simply write $\dist(u,v) =
\dist_G(u,v)$.
For \(T \in [n]\), the \(T\)-neighborhood of a node \(u \in V\) is the set
\(\NN_T(u) = \left\{v \in V \ \st \ \dist(u,v) \le T \right\}\).
The \(T\)-neighborhood of a subset \(A \subseteq V\) is the set \(\NN_T(A) =
\left\{v \in V \ \st \ \exists u \in A : \dist(u,v) \le T \right\}\).
Similarly, the \(T\)-neighborhood of a subgraph \(H \subseteq G\) is the set
\(\NN_T(H) = \left\{v \in V(G) \ \st \ \exists u \in V(H) : \dist_G(u,v) \le T
\right\}\).

For $c \in \nat$, a \emph{$c$-coloring} of a graph $G$ is a map $\varphi\colon
V(G) \to [c]$.
The coloring $\varphi$ is said to be \emph{proper} if we have $\varphi(u) \neq
\varphi(v)$ for every $uv \in E$.
If, for some \(\chi \in \nat\), there exists a proper $\chi$-coloring for $G$ and $\chi$ is minimal with this
property, then $G$ is said to be \emph{$\chi$-chromatic}; we also say that the chromatic
number of \(G\), denoted by \(\XX(G)\), is \(\chi\). In the $c$-coloring problem, the input is a graph $G$, and the task is to find a proper $c$-coloring of $G$.

\paragraph{\boldmath The \local model.}

The \local model is a distributed system consisting of a set of \( \abs{V} =
n\) \emph{processors} (or \emph{nodes}) that operates in a sequence of
synchronous rounds. 
In each round the processors may perform unbounded computations on their
respective local state variables and subsequently exchange of messages of
arbitrary size along the links given by the underlying input graph.
Nodes identify their neighbors by using integer labels assigned successively
to communication ports.
(This assignment may be done adversarially.)
Barring their degree, all nodes are identical and operate according to the
same local computation procedures.
Initially all local state variables have the same value for all processors;
the sole exception is a distinguished local variable \(\inptData(v)\) of
each processor \(v\) that encodes input data. 

Let \(c \ge 1\) be a constant, and let \(\inLabels\) be a set of input labels.
The input of a problem is defined in the form of a labeled graph \((G,
\inptData)\) where \(G = (V, E)\) is the system graph, \(V\) is the set of
processors (hence it is specified as part of the input), and \(\inptData\colon V
\to [n^c] \times \inLabels\) is an assignment of a \emph{unique} identifier
\(\id(v) \in [n^c]\) and of an input label \(\inLabel(v) \in \inLabels\) to
each processor \(v\). 
The output of the algorithm is given in the form of a vector of local output
labels \(\outLabel\colon V \to \outLabels\), and the algorithm is assumed to
terminate once all labels \(\outLabel(v)\) are definitely fixed. 
We assume that nodes and their links are fault-free.
The local computation procedures may be randomized by giving each processor
access to its own set of random variables; in this case, we are in the
\emph{randomized} \local (\randlcl) model as opposed to \emph{deterministic}
\local (\detlcl).

The running time of an algorithm is the number of synchronous rounds required by all nodes to produce output labels. 
If an algorithm running time is \(T\), we also say that the algorithm has locality \(T\).
Notice that \(T\) can be a function of the size of the input graph.

We say that the \(c\)-coloring problem over some graph family \(\FF\) has
complexity \(T\) in the \detlcl model if there exists a \detlcl algorithm
solving the problem in time \(T\) for all input graphs \(G\in \FF\), but no
\detlcl algorithm solves the problem in time \(T-1\) (where \(T\) can be a
function of the input graph size) for all input graphs \(G\in \FF\).
The complexity in the \randlcl model is defined similarly.
 \section{New classical graph coloring algorithms} \label{sec:coloring:ub}

In this section we prove the existence of algorithms in \detlcl and \randlcl
that almost match the \nslcl lower bound.
We recall here the precise statement for the reader's convenience.

\thmMainResultUB*

Notice that, if we plug in the same parameters (with $c = \alpha(\chi-1)+1$) in 
\Cref{thm:lb-coloring} with an appropriate choice of constant $\eps$, we get
that $\alpha = \floor*{\frac{c-1}{\chi-1}} $, implying a lower bound of
\(
	\myOmega{n^{1/\alpha}/\chi^{1 + 1/\alpha}}
\)
for the problem.
Therefore, for constant $\chi$ and $\alpha$, our algorithms give a perfect
trade-off between quality of approximation and time complexity up to logarithmic
factors.

\paragraph{Fast coloring from fast network decomposition.}
Our algorithm follows an approach similar to that of \textcite{barenboim2018},
which in turn is based on \emph{network decompositions}.
\begin{definition}[Network decomposition]
	An $(\alpha,d)$-\emph{network decomposition} of a graph $G$ is a partition
	$V(G) = C_1 \cup \dots \cup C_k$ along with a map $\mu\colon \{ C_1,\dots,C_k
	\} \to [\alpha]$ meeting the following conditions:
	\begin{itemize}
		\item The \emph{clusters} $C_i$ are pairwise disjoint (i.e., $C_i \cap C_j =
		\varnothing$ unless $i = j$).
		\item For every $i$, the (weak) diameter $\max_{u,v \in C_i} \dist_G(u,v)$
		of $C_i$ is at most $d$.
		\item The \emph{supergraph} $S = S(G)$ with node set $V(S) = \{
		C_1,\dots,C_k \}$ and edge set $E(S) = \{ \{C_i,C_j\} \mid \exists u \in
		C_i, v \in C_j: \{u,v\} \in E(G) \}$ that is obtained by contracting each
		$C_i$ is $\alpha$-colorable.
		\item $\mu$ is an $\alpha$-coloring of $S$.
	\end{itemize}
	In addition, the coloring $\mu$ is presented explicitly to the nodes of $G$;
	that is, every node $v \in C_i$ knows the \emph{cluster color} $\mu(C_i)$ of
	its respective cluster $C_i$.
\end{definition}

We recall the algorithm of \textcite{barenboim2018}.
Suppose we are given an $(\alpha, d(n))$-network decomposition $\mathcal{D}$.
We iterate sequentially through the $\alpha$ cluster colors of $\mathcal{D}$.
For each cluster color $a$ and each cluster $C$ having the cluster color $\mu(C)
= a$, collect the entire topology of $C$ in some arbitrary node $v \in C$
(chosen by, e.g., leader election).
The node $v$ then computes a perfect coloring $\pi$ for the nodes of $C$ that
uses at most $\chi$ colors and then broadcasts to each node $u \in C$ its color
$\pi(u)$.
Clearly in the \local model each iteration takes at most diameter of $C_i$ many
rounds, so at most $\myO{d(n)}$ rounds.
With a clever implementation, this process can also be sped up by doing all the
iterations in parallel. We will see this in the later sections.

In order for this strategy to work, we must ensure that \emph{neighboring
clusters use distinct sets of colors}; otherwise, the colorings of the nodes of
neighboring clusters may not match.
To deal with this, we have clusters of different cluster colors use completely
different sets of colors for the nodes, which we will refer to as \emph{color
palettes}.
More precisely, each cluster $C$ of cluster color $\mu(C) = a$ may only use
colors from the palette $p_a = \{ (a,b) \mid b \in \nat_+ \}$.
Since there are $\alpha$ many colors for the clusters, (assuming each cluster
is colored using at most $\chi$ colors) this then gives a coloring of $G$ with
$\alpha \chi$ colors.
As we color each cluster by gathering its entire structure in a single node,
knowledge of $\chi$ is not needed in order to use the optimal number of colors
for each cluster.

\paragraph{The \enquote{hiding trick}.}
We show how to optimize the above strategy using what we call a \enquote{hiding
trick}.
We add one special \emph{hidden color} $-1$ in common to all color palettes
$p_a$ and that is guaranteed to appear only in the \enquote{inside} of the
clusters; that is, if a node $v$ in some cluster $C$ has a neighbor that is not
in $C$, then $v$ is guaranteed to not be colored $-1$.
This ensures the colorings produced by two neighboring clusters are still
compatible since the only color shared by their palettes is the hidden color
$-1$, which is only present in the \enquote{insides} of the clusters.
Since the palettes now share exactly one color, this allows us to save $\alpha -
1$ colors in total.
Surprisingly enough, this small modification is enough to attain the minimum
number of colors possible, that is, $\alpha(\chi-1) + 1$ colors (as per our
lower bound from \cref{thm:lb-coloring}).

\paragraph{Fast decomposition from fast clustering.}
The specific complexity of the resulting algorithm depends on value $d(n)$ of
the underlying decomposition $\mathcal{D}$.
We show how to obtain a decomposition with $d(n) = \myOTilde{n^{1/\alpha}}$ in
$\myO{d(n)}$ time.
To do so, we show how to efficiently turn any existing \emph{clustering} algorithm
into a network decomposition algorithm.
The difference between the two is that the former only needs to group a
\emph{subset} of nodes in the graph, whereas the latter must partition the
\emph{entire} graph.

\begin{definition}[$(\lambda, d)$-clustering]
	Given a graph $G$, a \emph{$(\lambda, d)$-clustering} is a partition $V(G) = D
	\cup S_1 \cup \dots \cup S_k$ meeting the following conditions:
	\begin{itemize}
		\item $S_1, \dots, S_k$ are mutually non-adjacent; that is, the
		distance between any two nodes $u \in S_i$ and $v \in S_j$ where $i \neq j$
		is at least 2.
		\item For every $i$, the (weak) diameter $\max_{u,v \in S_i} \dist_G(u,v)$
		of $S_i$ is at most $d$.
		\item $D$ contains at most $\lambda \abs{V(G)}$ vertices.
	\end{itemize}
\end{definition}

Our conversion from a clustering algorithm into a network decomposition one is
by a \emph{bootstrapping} procedure:
if the clustering algorithm is guaranteed to cluster at least half of the nodes
in $G$, then we can apply it again and again until only an $\eps$ fraction of
nodes remains unclustered, where $\eps$ is a parameter of our choosing.
For an appropriate choice of $\eps$, the fraction of nodes that remains is
sufficiently small that we can directly gather the remaining nodes into their
own cluster (simply by grouping every remaining connected component of the
graph).

\paragraph{Organization.}
In \cref{sec:hiding-trick} we show how to obtain a fast coloring algorithm given
the underlying network decomposition $\mathcal{D}$.
In \cref{sec:alg-decomposition} we then show how to obtain such a decomposition
following the two-step approach described above: first we show our bootstrapping
result for clustering algorithms in \cref{sec:fast-clustering} and then how this
implies a network decomposition algorithm in \cref{sec:fast-decomposition}.
Plugging in the state-of-the-art for clustering algorithms, we then obtain
\cref{thm:upperbound_full}.

\subsection{The hiding trick}
\label{sec:hiding-trick}

The main result of this section is the following, which is a coloring algorithm
based on our so-called hiding trick (see also \cref{lem:hiding} below).
The algorithm presupposes the existence of a network decomposition
algorithm, which we show how to obtain in \cref{sec:alg-decomposition}.

\begin{theorem}
	\label{thm:alg-coloring}Fix some parameter $\alpha \in \nat$ and suppose there is an
	$(\alpha,d(n))$-network decomposition algorithm $\mathcal{B}$ for \detlcl or
	\randlcl that runs in time $d(n)$.
	Then there is an algorithm $\mathcal{A}$ that $(\alpha(\chi-1)+1)$-colors any
	$\chi$-chromatic graph $G$ with $n$ nodes in $O(d(n))$ time.
	Moreover, $\mathcal{A}$ works in \detlcl or \randlcl, depending on which model
	$\mathcal{B}$ runs in.
	In addition, if $\mathcal{B}$ is randomized (i.e., it is a \randlcl algorithm)
	and succeeds with probability $1-1/\poly(n)$, then so does $\mathcal{A}$
	succeed with probability $1-1/\poly(n)$.
\end{theorem}

The core idea of our algorithm is the following constructive lemma.
It shows that, in any graph $G$, it is always possible to color a subset $A$ of
nodes in a way that \enquote{hides} one designated hidden color $-1$.
To color $A$ in this manner, it may be necessary to fix the color of some nodes
in $N(A)$ as well.

\begin{lemma}[Hiding Lemma]\label{lem:hiding}Let \(G = (V, E)\) be a graph, and let \(\chi\) be the chromatic number of
	$G$.
	For any subset \(A \subseteq V\), there exists $A \subseteq A' \subseteq (A
	\cup N(A))$ and a proper coloring $\varphi\colon A' \rightarrow [\chi - 1]
	\cup \{-1\}$ of \(A'\) such that $A$ is completely colored and, for any node
	\(v \in V\setminus A'\), \(v\) is not adjacent to a node with color $-1$.
\end{lemma}

\begin{proof}
	Since \(G\) is \(\chi\)-colorable, there exists a $\chi$-coloring $\psi$ of
	$A \cup N(A)$.
	We uncolor some of the nodes of $N(A)$ such that none of the uncolored nodes
	has a neighbor with color $-1$.
	Formally, let
	\[
		A' = (A \cup N(A)) \setminus \{ u \in N(A) \mid \psi(u) = -1 \}
	\]
	and $\varphi = \psi{\restriction_{A'}}$ (i.e., the restriction of $\psi$ to
	$A'$).
	Since we only uncolor nodes in $N(A)$ and $\psi$ is a proper $\chi$-coloring
	of $A \cup N(A)$, $\varphi$ is certainly a proper coloring of $A' \subseteq (A
	\cup N(A))$.
	We now argue that $v \in V \setminus A'$ has no neighbor that is colored $-1$
	by $\varphi$.
	Since $A \subseteq A'$, we must consider only the following two cases:
	\begin{description}
		\item[$v \in N(A)$.] Since $v$ is not in $A'$ but still in $N(A)$, $\psi(v)
		= -1$.
		Hence, since $\psi$ is a proper coloring, every node in $N(v)$ colored by
		$\psi$ has a color that is different from $-1$.
		Since $\varphi$ is a restriction of $\psi$, the same holds for any node in
		$N(v)$ colored by $\varphi$.
		\item[$v \notin A \cup N(A)$.] Then $v$ is only adjacent to nodes in $N(A)$.
		Since no node in $N(A)$ is colored $-1$ by $\varphi$ by definition, clearly
		$v$ has no neighbor colored $-1$ by $\varphi$.
		\qedhere
	\end{description}
\end{proof}

Note that \cref{lem:hiding} (in its current form) cannot be immediately applied
to a decomposition of $G$ to produce colorings for the clusters of $G$.
The reason for that is the following:
Consider two clusters $C$ and $C'$ that are assigned the same cluster color $a$
by the network decomposition.
(In particular, this means $C$ and $C'$ are not adjacent.)
Recall that we will color the nodes of $C$ and $C'$ using the same color
palette.
However, if we color them such as in \cref{lem:hiding}, then we are potentially
also coloring nodes in $N(C)$ and $N(C')$ and, for all we know, the intersection
$N(C) \cap N(C')$ may not be empty.
Hence we cannot simply use the same color palette in both clusters, as this
could potentially lead to an invalid coloring.

Therefore, we would like that not only $C$ and $C'$ but also $N(C)$ and $N(C')$
are not adjacent.
Equivalently, we wish for the distance between $C$ and $C'$ to be at least $4$.
We can indirectly guarantee this by using an $(\alpha, d(n))$-network
decomposition $\mathcal{D}$ of $G^3$ instead of $G$.
This is because being not adjacent in $G^3$ immediately implies a distance of at
least $4$ in $G$.
Asymptotically speaking, this does not incur any additional cost (compared to
computing a decomposition of $G$) since, for any $k$, we can simulate each round
of communication in $G^k$ by using $k$ rounds of communication in $G$.

We now present \cref{alg:coloring}, which, given a $\chi$-chromatic graph $G$
and an $(\alpha, d(n))$-network decomposition $\mathcal{D}$ of $G^3$, colors $G$
using $\alpha(\chi-1)+1$ colors.
By proving the correctness of \cref{alg:coloring} we then obtain
\cref{thm:alg-coloring}.

\begin{algorithm}
\caption{Coloring a decomposition}\label{alg:coloring}
\begin{algorithmic}[1]
\Require $G=(V,E)$, $(\alpha,d(n))$-network decomposition $\mathcal{D}$ of $G^3$
\For{each cluster $C_i$ in parallel}
	\State Elect a leader node $\leader_i$ of $C_i$ \label{line:leader-election}
	\State Collect the entire topology of $C_i \cup N(C_i)$ in $\leader_i$
	\label{line:collect}
	\State $a \gets \text{the cluster color assigned to $C_i$ by $\mathcal{D}$}$
	\State $p_a \gets \{ (a,b) \mid b \in \nat \} \cup \{ -1 \}$
	\If{$a = \alpha$}
		\State $\leader_i$ computes a minimal coloring $\varphi_i \colon C_i \to
		p_a$ of $C_i$
		\State $\leader_i$ broadcasts $\varphi_i$ to all nodes in $C_i$
	\Else
		\State $\leader_i$ computes a minimal node coloring $\varphi_i \colon C_i' \to
		p_a$ of some $C_i' \subseteq C_i \cup N(C_i)$ according to \Cref{lem:hiding}
		\State $\leader_i$ broadcasts $\varphi_i$ to all nodes in $C_i \cup N(C_i)$
		\label{line:broadcast}
	\EndIf
\EndFor
\For{each node $u$ in parallel}
	\State $\Phi_u \gets \text{the set of all colors assigned to $u$ by the
	$\varphi_i$}$
	\If{$\Phi_u = \{ -1 \}$}
		\State Color $u$ with the color $-1$
	\Else
		\State Color $u$ with an arbitrary color from $\Phi_u \setminus \{-1\}$
	\EndIf
\EndFor
\State Each node outputs its own color
\end{algorithmic}
\end{algorithm}

Let us give a brief high-level description of \cref{alg:coloring}.
Each cluster $C_i$ acts independently and based on the cluster color $a$ that is
assigned to it by the decomposition $\mathcal{D}$.
First, the entire topology of $C_i \cup N(C_i)$ is gathered in some leader node
$\leader_i$.
Next $\leader_i$ brute-forces an optimal coloring $\varphi_i$ of its respective
cluster $C_i$ using a color palette $p_a$ that depends on $a$.
(Without restriction we use only the smallest elements of $p_a$ (according to
the natural ordering of $p_a$) in $\varphi_i$.
This enables the nodes to choose correct colors even without knowledge of
$\chi$.)
If $a = \alpha$, $\varphi_i$ is simply a $\chi$-coloring of $C_i$ whose
existence is guaranteed by the $\chi$-chromaticness of $G$.
Otherwise (i.e., if $a < \alpha$), $\leader_i$ instead computes a coloring
$\varphi_i$ of $C_i' \subseteq C_i \cup N(C_i)$ according to \cref{lem:hiding}.
Each coloring $\varphi_i$ is broadcasted to all nodes that may have been
assigned a color by $\varphi_i$.
The nodes that were assigned multiple colors (i.e., potentially those at the
border of two distinct clusters) then simply choose one of them arbitrarily.

\begin{lemma}
	\label{lem:alg-coloring-correctness}
	\Cref{alg:coloring} computes a valid coloring of $G$ using no more than
	$\alpha(\chi-1)+1$ colors.
\end{lemma}

\begin{proof}
First we argue that no more than $\alpha(\chi-1)+1$ colors are used in total.
For a cluster color $a \in [\alpha]$, let
\[
	p_a' = \{ \varphi_i(v) \mid \exists i: \mu(C_i) = a \land v \in C_i \}
	\subseteq p_a
\]
be the set of colors actually used by the nodes to color clusters with the color
$a$.
By the $\chi$-chromaticness of $G$, every minimal coloring of any induced
subgraph of $G$ uses at most $\chi$ colors, so we have $\abs{p_a'} \le \chi$.
Since the intersection of two palettes $p_a$ and $p_b$ is exactly $\{ -1 \}$ for
$a \neq b$, we use
\[
	\abs*{\bigcup_{a=1}^\alpha p_a'}
	= \abs*{\{-1\}} + \sum_{a=1}^\alpha \abs*{p_a' \setminus \{-1\}}
	\le \alpha(\chi-1) + 1
\]
colors in total, as claimed.

To show that the color is proper, first observe that, based on
\cref{lem:hiding}, no node in $N(C_i)$ is assigned $-1$ by $\varphi_i$; as a
result, if a node $v$ is assigned the color $-1$ by $\varphi_i$, then
necessarily $v \in C_i$.
Consider any two adjacent nodes $v$ and $v'$ and recall that both of them choose
the largest color among any of the colors they were assigned.
For the sake of contradiction, suppose that both nodes pick the same color $x$.
Consider the following two cases:
\begin{description}
	\item[$x \neq -1$.] Since the intersection of any two color palettes is
	$\{-1\}$, this implies that $v$ and $v'$ were assigned colors from the same
	color palette $p_a$.
	Since the $\varphi_i$ are all valid colorings, the colors of $v$ and $v'$ come
	from different coloring functions $\varphi_i$ and $\varphi_{i'}$,
	respectively.
	However, since $\varphi_i$ and $\varphi_{i'}$ use the same color palette,
	the respective clusters $C_i$ and $C_{i'}$ are assigned the same cluster color by
	$\mathcal{D}$.
	This means that $C_i$ and $C_{i'}$ are not adjacent in $G^3$ and, in turn, the
	distance between nodes in $N(C_i)$ and $N(C_{i'})$ is at least $2$, which
	immediately contradicts $v$ and $v'$ being adjacent.
	\item[$x = -1$.] This means that both of the nodes are assigned their color by
	the coloring of their respective cluster.
	Let $i$ and $i'$ be the numbers of the clusters of $v$ and $v'$, respectively.
	If $i = i'$, then both $v$ and $v'$ pick their colors according to $\varphi_i
	= \varphi_{i'}$, which contradicts $\varphi_i$ being a proper coloring.
	Hence let $i \neq i'$.
	We have then that $C_i$ and $C_{i'}$ are adjacent clusters and, since they are
	distinct, without restriction we have $i \neq \alpha$.
	Since $v' \in N(C_i)$ and the color of both $v$ and $v'$ is $-1$,
	\Cref{lem:hiding} implies that $v'$ is in the domain of $\varphi_i$ (as
	otherwise it would be adjacent to $v$, which has the color $-1$).
	Now since $\varphi_i(v) = -1$, we know that $\varphi_i(v') > -1$, which means
	that the color of $v'$ cannot be $-1$.
	\qedhere
\end{description}
\end{proof}

Note that, in the proof above, we simply showed that no more than $c = \alpha
(\chi-1) + 1$ colors are used, but there is still a minor technicality to be
dealt with since the colors do not come from the set $[c]$ as the definition of
$c$-coloring demands (but rather either the color is $-1$ or a pair $(a,b)$).
A straightforward way of resolving this is, for instance, remapping $-1$ to $1$
and every pair $(a,b) \in p_a$ to $\alpha(b-1) + a + 1$.
(Note this gives a bijection between $\{-1\} \cup \bigcup_{a \in [\alpha]} p_a$
and $\nat_+$.)
Since we use only the smallest elements of the palette $p_a$ in each respective
coloring $\varphi_i$, we know that any $(a,b) \in p_a'$ must be such that $b \le
\chi-1$.
Hence the largest color used is $\alpha(\chi-2) + \alpha + 1 = c$ (corresponding
to $(\alpha,\chi-1) \in p_\alpha$).

\begin{lemma}
	\label{lem:alg-coloring-complexity}
	Given the $(\alpha, d(n))$-network decomposition $\mathcal{D}$,
	\Cref{alg:coloring} can be run distributedly in $\myO{d(n)}$ rounds in the
	\detlcl model.
\end{lemma}
\begin{proof}
The only lines in the algorithm that require communication are 
\cref{line:leader-election,line:collect,line:broadcast}.
Since our clusters have diameter at most $d(n)$, \cref{line:leader-election}
requires only $\myO{d(n)}$ rounds of communication.
The other two trivially take only $d(n)+1$ rounds since the message size in the
\detlcl model is unbounded and also any node in $N(C_i)$ has distance at most
$d(n)+1$ to its respective $\leader_i$.
All other steps in the algorithm are local computations that do not incur any
cost in the \detlcl model.
\end{proof}

Together, \cref{lem:alg-coloring-correctness,lem:alg-coloring-complexity} prove
\cref{alg:coloring} satisfies the requirements of \cref{thm:alg-coloring}, thus
concluding its proof.

\subsection{Fast network decomposition}
\label{sec:alg-decomposition}

Given any clustering algorithm, we can obtain a network decomposition algorithm
as follows.

\begin{theorem}
	\label{thm:NetDec}Let $f,g\colon \nat \to \nat$ be arbitrary functions and suppose there is an
	$\myO{f(n)}$-round distributed $(1/2, g(n))$-clustering algorithm named
	\cluster for \detlcl or \randlcl.
	There is an algorithm $\mathcal{A}$ that, given a graph $G = (V,E)$ and any
	$\alpha \in \nat_+$, computes an \((\alpha, \OO(n^{1/\alpha} g(n)))\)-network
	decomposition of $G$ in
	\[
	  d(n) = \myO{\left(\frac{n}{g(n)}\right)^{1/\alpha}
			\left( f(n) + g(n) \right) \log\frac{n}{g(n)}}
	\]
	rounds.
	The algorithm $\mathcal{A}$ works in \detlcl or \randlcl, depending on which
	model \cluster itself is based on.
	In addition, if \cluster is randomized and succeeds with probability $1 -
	1/\poly(n)$, then $\mathcal{A}$ also succeeds with probability $1 -
	1/\poly(n)$.
\end{theorem}

We mention two state-of-the-art clustering algorithms that can be plugged into
\cref{thm:NetDec}, one for the \detlcl model and one for the \randlcl model. 
For the \detlcl model we use the clustering algorithm from \cite{GGHIR23}.
In fact this algorithm actually works even in the more restricted \congest
model, where each node can only send $O(\log n)$-bit messages each round.

\begin{theorem}[\cite{GGHIR23}]\label{thm:DND}There exists an algorithm that computes a $(1/2, \myO{\log n \cdot \log \log
	\log n})$-clustering in \detcgt in $\myOTilde{\log^2 n}$ rounds.
\end{theorem}

Plugging in this algorithm in \cref{thm:NetDec}, we obtain the first item of
\cref{thm:upperbound_full}.
For the \randlcl model (i.e., the second item of \cref{thm:upperbound_full}), we
use the following.

\begin{theorem}[\cite{chang2023}]\label{thm:RND}There exists an algorithm that computes a $(1/2, \OO(\log n))$-clustering in
	the \randlcl model in $\myO{\log n}$ rounds with probability $1 - 1 /
	\poly(n)$.
\end{theorem}

There are two steps to the proof of \cref{thm:NetDec}.
In \cref{sec:fast-clustering} we show a bootstrapping result where we use the
$(1/2,g(n))$-clustering algorithm \cluster to obtain a
$(\eps,\OO(g(n)/\eps))$-clustering algorithm for any $\eps$ of our choosing.
Plugging in an adequate value for $\eps$, in \cref{sec:fast-decomposition} we
then obtain the network decomposition algorithm of \cref{thm:NetDec}.

\subsubsection{Fast clustering}
\label{sec:fast-clustering}

Next we show our bootstrapping procedure, with which we can reduce a $1/2$
fraction of unclustered nodes to any $\eps$ of our choosing.
The price to pay is only a multiplicative $\myO{1/\eps}$ factor in the diameter
of the clusters and a multiplicative $\myO{\eps^{-1} \log(1/\eps)}$ factor in
the running time.
Formally, what we achieve is the following:

\begin{theorem}
	\label{thm:DSC}Let \cluster be as in \cref{thm:NetDec}.
	For any $0 < \varepsilon \leq 1$, there is an algorithm $\eps$-\cluster that
	computes an $(\eps, \myO{g(n)/\eps})$-clustering in
	\[
		\myO{\frac{(f(n) + g(n)) \log(1/\eps)}{\eps}}
	\]
	rounds.
	The algorithm $\eps$-\cluster can be implemented in the same distributed model
	as \cluster and its runtime is dominated by $2\log(1/\eps)$ invocations of
	\cluster.
	In addition, if \cluster is randomized (i.e., it works in \randlcl) and
	succeeds with probability $1 - 1/\poly(n)$, then $\eps$-\cluster also succeeds
	with probability $1 - 1/\poly(n)$.
\end{theorem}

We first describe the algorithm (\cref{alg:clustering}) and then prove it
satisfies the properties of \cref{thm:DSC}.
Note that, in the description of \cref{alg:clustering} in all three Lines
\ref{line:minJ},\ref{line:delBoundary} and \ref{line:fixC} the neighborhoods are
always taken with respect to $G$.

\begin{algorithm}
\caption{Bootstrapping \cluster}\label{alg:clustering}
\begin{algorithmic}[1]
\Require $G=(V,E), 0 < \varepsilon \leq 1$ 
\Require Clustering algorithm \cluster as in \cref{thm:DSC}
\State $R \gets 4/\eps$
\State $\mathcal{C'} \gets \text{empty clustering}$
\For{$2\log(1/\eps)$ times} \label{line:clustering_outer_for}
	\State Run \cluster on $G^{2R+1}$, producing a clustering $\mathcal{C}$
	\For{each cluster $C \in \mathcal{C}$ in parallel}
	\label{line:clustering_inner_for}
		\State Find $j^* \in [R]$ such that
		$\abs{\NN_{j^*}(C) \setminus \NN_{j^*-1}(C)}$ is minimized
		\label{line:minJ}
		\State Mark all nodes in $\NN_{j^*}(C) \setminus \NN_{j^*-1}(C)$ for
		deletion
		\label{line:delBoundary}
		\State Add $\NN_{j^*-1}(C)$ as a cluster to $\mathcal{C}'$
		\label{line:fixC}
	\EndFor
	\State Let $U$ be the set of nodes not marked for deletion or that did not
	join a cluster of $\mathcal{C'}$
	\State $G \gets G[U]$
\EndFor
\State Each node outputs whether it is part of a cluster of $\mathcal{C}'$ or
not
\end{algorithmic}
\end{algorithm}

At a high level, \cref{alg:clustering} works in multiple iterations, in each of
which we first invoke \cluster.
For each new cluster $C$ that is computed, we delete some boundary around this
cluster and separate it from the rest of the graph.
We then add $C$ to our final clustering.
Since \cluster clusters at least half of the nodes, the size of our graph is
halved in each iteration; hence after $\myO{2 \log 1/\varepsilon}$ iterations we
are left with only an $\varepsilon$ fraction of nodes.

The main challenge lies in the choice of the boundaries and preventing too many
nodes from being deleted.
To solve this we appeal to the computational power of the \local model:
For each cluster $C$, we inspect the $(4/\varepsilon)$-hop neighborhood of $C$
and find $j^*$ such that the number of nodes at distance exactly $j^*$ from $C$
(i.e., $\abs{\NN_{j^*}(C) \setminus \NN_{j^*-1}(C)}$) is minimized.
As we will show, there is always a choice of $j^*$ such that the number of nodes
we delete is not too large.

We now turn to showing that \cref{alg:clustering} satisfies the properties in
\cref{thm:DSC}.
To give an overview, we need to prove the following:
\begin{enumerate}
  \item The clusters created have diameter $\myO{g(n)/\eps}$ and are
  non-adjacent (\cref{lem:clusterSeparation}).
  \item At least a $1-\eps$ fraction of the nodes get clustered
  (\cref{lem:clusterDeletion}).
  \item \cref{alg:clustering} runs in $O((f(n)+g(n)) \cdot \eps^{-1}
  \log(1/\eps))$ rounds (\cref{lem:clusterRuntime}).
	\item If \cluster is randomized and has success probability $1-1/\poly(n)$,
	then \cref{alg:clustering} also succeeds with probability $1-1/\poly(n)$
	(\cref{lem:clusterSuccProb}).
\end{enumerate}
We address these claims now one by one.

\begin{lemma}\label{lem:clusterSeparation}The clusters created by \cref{alg:clustering} have diameter $\myO{g(n)/\eps}$
and are non-adjacent.
\end{lemma}
\begin{proof}
We start by analyzing a single iteration of the for loop on
\cref{line:clustering_outer_for} and show that a cluster with diameter
$\myO{g(n)/\eps}$ is created.
Afterwards we prove that, once a good cluster is created, it is preserved by
later iterations.

\cref{alg:clustering} runs the algorithm \cluster of \cref{thm:DND} on the power
graph $G^{2R +1}$.
Let $C_i, C_j$ be two clusters created by \cluster.
For any nodes $u \in C_i$ and $v \in C_j$ with $i \neq j$, we have
$\dist_{G'}(u,v) \ge 2$, which implies $\dist_G(u,v) \ge 2R+2$.
Let $S_i \in \mathcal{S}$ and $S_j \in \mathcal{S}$ be the clusters that are
fixed by \cref{line:fixC} in the iterations of $C_i$ and $C_j$, respectively.
We observe that $S_i$ (resp., $S_j$) only contains nodes at distance at most
$R-1$ from $C_i$ (resp., $C_j$).
Hence it follows that, for any $u' \in S_i$ and $v' \in S_j$, we have
$\dist_G(u',v') > 2R + 2 - 2(R-1) \ge  4$.

In addition, \cluster guarantees that the diameter of the clusters is at most
$g(n)$ in $G^{2R+1}$.
Therefore, for $u,v \in C_i$,
\[
  \dist_G(u,v) \le (2R+1) \dist_{G'}(u,v) \le (2R+1) g(n).
\]
When creating a fixed cluster in \cref{line:fixC} we increase the diameter by at
most $2j^* \le 2R-2$ hops.
As a result, the diameter of any fixed cluster is bounded above by
\[
  (2R + 1) g(n) + 2R - 2 = \myO{R \cdot g(n)} = \myO{g(n)/\eps}.
\]

Since the entire one-hop-neighborhood of every cluster is deleted, in the next
iterations all of the new clusters will also have at least one deleted node
between themselves and any other previously created cluster.
Also no nodes of previously fixed clusters are ever considered for another
cluster or deleted.
Hence the subsequent iterations do not interfere with the previous ones.
\end{proof}

\begin{lemma}\label{lem:clusterDeletion}\cref{alg:clustering} deletes at most an $\eps$ fraction of nodes during its
  execution.
  All nodes that are not deleted eventually join a cluster.
\end{lemma}
\begin{proof}
Without loss of generality, we assume $\eps \le 1/2$, as otherwise a single
execution of \cluster already gives the result.
We upper-bound the number of deleted nodes using an inductive argument.

Let $\mathcal{S}$ be the set of all clusters following a single execution of the
for loop on \cref{line:clustering_outer_for}.
In addition, for a cluster $C \in \mathcal{S}$, let $\del(C)$ denote the set of
nodes marked for deletion by $C$ (during the execution of the for loop on
\cref{line:clustering_inner_for}).
Observe that, for any two distinct clusters $C_1 \neq C_2$, we have that
$\del(C_1) \neq \del(C_2)$.
This is due to the fact that $\dist_G(C_1,C_2) \ge 4R+2$ and that
$\dist_G(v,C_i) \le j^\ast \le R-1$ for any $v \in \del(C_i)$.
Also, by an averaging argument, for any cluster $C \in \mathcal{S}$,
\[
  \abs*{\del(C)}
  \le \frac{1}{R} \bigcup_{1 \le j \le R} \abs*{\NN_j(C) \setminus \NN_{j-1}(C)}
  = \frac{\abs*{\NN_R(C) \setminus C}}{R}.
\]
This means we can upper-bound the total number of deleted nodes by
\[
	\sum_{C \in \mathcal{S}} \frac{\abs*{\NN_R(C) \setminus C}}{R}
    \leq \frac{n}{R}
    = \frac{\eps n}{4}.
\]

By \cref{thm:DND}, at least half of the remaining nodes in $G$ are clustered in
each iteration of the for loop on \cref{line:clustering_outer_for}.
Arguing as before, we get that the number of remaining nodes is halved with each
execution of the loop.
This means that, during the execution of the loop, the number of deleted nodes
is at most
\[
  \sum_{i=1}^{2\log(1/\eps)} \frac{\eps n}{4 \cdot 2^{i-1}}
  \le \frac{\eps n}{2}.
\]
In turn, due to the aforementioned progress guarantee, the $2\log(1/\eps)$
repetitions of the loop ensure at most a $2^{-2\log(1/\eps)} = \eps^2$ fraction
of nodes are left standing at the end and are then deleted in the final step of
\cref{alg:clustering}.
As a result, using that $\eps \le 1/2$, the total number of deleted nodes is at
most
\[
  \frac{\eps n}{2} + \eps^2 n \le \eps n. \qedhere
\]
\end{proof}

\begin{lemma}\label{lem:clusterRuntime}\cref{alg:clustering} terminates after $\myO{(f(n) + g(n)) \cdot \eps^{-1}
\log(1/\eps)}$ rounds.
\end{lemma}

\begin{proof}
We analyze the runtime of a single iteration of the for loop on
\cref{line:clustering_outer_for}.
When ran on $G^{2R+1}$, \cluster takes $\myO{f(n)}$ time, so for each round of
\cluster we spend $\myO{R} = \myO{1/\eps}$ rounds to simulate its execution on
$G$.
Hence we need $\myO{f(n)/\eps}$ rounds in total to run \cluster on $G$.
\cluster then outputs clusters of diameter $\myO{g(n)}$ on $G^{2R+1}$, which
correspond to clusters of diameter $\myO{R \cdot g(n)} = \myO{g(n) / \eps}$ on
$G$.
Next we collect the entire $R$-neighborhood of a cluster $C$ in some leader
node, compute $j^\ast$, and then broadcast $j^\ast$ to all nodes inside of
$\NN_{j^\ast}(C)$.
This all requires $\myO{g(n) / \eps + R} = \myO{g(n) / \eps}$ rounds.
Hence a single iteration of the for loop on \cref{line:clustering_outer_for}
costs $\myO{(f(n) + g(n))/\eps}$ rounds in total.
The final deletion procedure does not cost any additional rounds since each node
knows at the end whether it has joined a cluster or not.
Since we repeat the loop $\myO{\log(1/\eps)}$ times, the claim follows.
\end{proof}

\begin{lemma}\label{lem:clusterSuccProb}Let \cluster have success probability $1-1/\poly(n)$ in \randlcl.
	Then \cref{alg:clustering} also succeeds with probability $1-1/\poly(n)$.
\end{lemma}

\begin{proof}
	Let $c > 0$ be such that \cluster succeeds with probability at least
	$1-1/n^c$.
	In addition, let us assume $\eps > 1/n$ as otherwise the claim is trivial.
	The observation to make is that, if all executions of \cluster by
	\cref{alg:clustering} are correct, then the result of \cref{alg:clustering} is
	also correct.
	Since there are $2\log(1/\eps)$ executions of \cluster in total, this means
	that, for any $0 < \tilde{c} < c$ (and, in particular, for any fixed choice of
	such a $\tilde{c}$) and large enough $n$, the probability that
	\cref{alg:clustering} succeeds is at least
	\[
		\left(1 - \frac{1}{n^c}\right)^{2 \log(1/\eps)}
		\ge 1 - \frac{2\log(1/\eps)}{n^c}
		> 1 - \frac{2\log n}{n^c}
		> 1 - \frac{1}{n^{\tilde{c}}}.
		\qedhere
	\]
\end{proof}

This concludes the analysis of \cref{alg:clustering}, from which \cref{thm:DSC}
follows.

\subsubsection{Fast network decomposition from fast clustering}
\label{sec:fast-decomposition}

Finally we show how a clustering algorithm as in \cref{thm:DSC} implies a
network decomposition algorithm, thereby giving a proof of \cref{thm:NetDec}.

\begin{lemma}\label{lem:NetDec}
	Let an algorithm $\eps$-\cluster as in \cref{thm:DSC} be given where $\eps =
	(g(n)/n)^{1/\alpha}$.
	Given any $\alpha \in \nat_+$ and a graph $G = (V,E)$, \cref{alg:NetDec}
	computes an $(\alpha, \OO(n^{1/\alpha} g(n)))$-network
	decomposition of $G$ in
	\[
	  d(n) = \myO{\left(\frac{n}{g(n)}\right)^{1/\alpha}
			\left( f(n) + g(n) \right) \log\frac{n}{g(n)}}
	\]
	rounds.
	\Cref{alg:NetDec} works in \detlcl or \randlcl, depending on which model
	\cluster itself is based on.  
	In addition, if $\eps$-\cluster is randomized and succeeds with probability
	$1-1/\poly(n)$, then \cref{alg:NetDec} also succeeds with probability
	$1-1/\poly(n)$.
\end{lemma}

The strategy followed by \cref{alg:NetDec} is very much straightforward:
First apply the clustering algorithm $\alpha-1$ times.
The remaining graph contains then at most $\eps^{\alpha-1} n = \myO{g(n) / \eps}
= \myO{n^{1/\alpha} g(n)}$ many nodes.
At this point we can just put all remaining connected components into their own
clusters, which will trivially have diameter at most $\myO{g(n) / \eps}$.

\begin{algorithm}
\caption{$(\alpha, \OO(n^{1/\alpha}))$-network decomposition from clustering}
\label{alg:NetDec}
\begin{algorithmic}[1]
\Require $G=(V,E)$,
$\alpha \in \nat_+$ 
\Require Clustering algorithm $\eps$-\cluster from \cref{thm:DSC}
\State $U \gets V$
\For{$1 \leq i \leq \alpha-1$}
  \State Run $\eps$-\cluster on $G[U]$, producing a clustering $\mathcal{C}_i$
	\label{line:invokeCluster}
	\State Color every cluster $C \in \mathcal{C}_i$ with the cluster color
	$\mu(C) = i$
  \State $U \gets U \setminus \bigcup_{C \in \mathcal{C}_i} C$
	\label{line:redefineU}
\EndFor
\State Form a clustering $\mathcal{C}_\alpha$ of $G[U]$ by having each connected
component form its own cluster
\label{line:partitionRest}
\State Color every cluster $C \in \mathcal{C}_\alpha$ with the cluster color
$\mu(C) = \alpha$
\State Each node $v \in C$ outputs its cluster color $\mu(C)$
\label{line:colorRest}
\end{algorithmic}
\end{algorithm}

\begin{proof}
	Clearly \cref{alg:NetDec} produces clusters with the correct diameter:
	$\eps$-\cluster produces clusters with diameter $\myO{g(n)/\eps} =
	\myO{n^{1/\alpha} g(n)}$ and there are only $\eps^{\alpha-1}n =
	\myO{g(n)/\eps}$ nodes to be clustered in \cref{line:partitionRest}.
	As for $\mu$ being a proper cluster coloring, notice that $\eps$-\cluster
	already guarantees the clusters formed in clustering $\mathcal{C}_i$ are
	non-adjacent; this is also guaranteed in the clustering created in
	\cref{line:partitionRest}.
	Regarding the round complexity, we have $\alpha-1$ many invocations of
	$\eps$-\cluster and then at most $\eps^{\alpha-1}n$ many nodes to cluster in
	\cref{line:partitionRest}.
	Hence using that the running time of each invocation of $\eps$-cluster is
	$\myO{(f(n)+g(n)) \cdot \eps^{-1} \log(1/\eps)}$ (by \cref{thm:DSC}), we can
	upper-bound the round complexity by
	\[
		\myO{\frac{\alpha (f(n)+g(n)) \log(1/\eps)}{\eps} + \eps^{\alpha-1}n}
		= \myO{\left(\frac{n}{g(n)}\right)^{1/\alpha}
				\left( f(n) + g(n) \right) \log\frac{n}{g(n)}}. \qedhere
	\]
\end{proof}

Together, \cref{thm:DND,thm:DSC} give the $\eps$-\cluster for \cref{lem:NetDec}
and we obtain a network decomposition algorithm as in \cref{thm:NetDec}.
As already discussed above, combining this with the coloring algorithm of
\cref{thm:alg-coloring}, we obtain our main result \cref{thm:upperbound_full}.
 \section{New lower bounds in the non-signaling model}\label{sec:lb-technique}

\subsection{Framework}

In this section we define the framework in which our technique is developed. 
We start with the notion of labeling problem. 

\begin{definition}[Labeling problem]\label{def:labeling-problem}
    Let \(\inLabels\) and \(\outLabels\) two sets of input and output labels, respectively.
    A \emph{labeling problem} \(\problem\) is a mapping \((G,\inLabel) \mapsto \{\indOutLabel{i}\}_{i \in I}\), with \(I\) being a discrete set of indices, that assigns to every graph \(G\) with any input labeling \(\inLabel: V(G) \to \inLabels\) a set of permissible output vectors \(\indOutLabel{i} : V(G) \to \outLabels\) that might depend on \((G, \inLabel)\).
    The mapping must be closed under graph isomorphism, i.e., if \(\varphi: V(G) \to V(G')\) is an isomorphism between \(G\) and \(G'\), and \(\indOutLabel{i} \in \problem((G', \inLabel)) \), then \(\indOutLabel{i} \circ \varphi \in \problem((G, \inLabel \circ \varphi))  \).
\end{definition}

A labeling problem can be thought as defined for \emph{any} input graph of \emph{any} amount of nodes.
If the set of permissible output vectors is empty for some input \((G,\inLabel)\), we say that the problem is not solvable on the input \((G,\inLabel)\):
accordingly, the problem is solvable on the input \((G,\inLabel)\) if \(\problem(G,\inLabel) \neq \emptyset\).

One observation on the generality of definition of labeling problem follows: 
one can actually consider problems that require to output labels on edges.
This variation of \cref{def:labeling-problem} does not affect in any way the applicability and the generality of the result we present in \cref{sec:lb-technique}, namely, our lower bound technique.

We actually focus on labeling problems where, for any input graph, an output vector \(\outLabel\) is permissible if and only if the restrictions of the problem on any local neighborhoods can be solved and there exist compatible local permissible output vectors whose combination provides \(\outLabel\).
This concept is grasped by the notion of locally verifiable labeling (LVL) problems, the generalization of locally checkable labeling (LCL) problems to unbounded degree graphs, first introduced by \textcite{naor1995}. 
For any function \(f : A \to B\) and any subset \(A' \subseteq A\), let us denote the restriction of \(f\) to \(A'\) by \(f \restriction_{A'}\). 
Furthermore, we define a centered graph to be a pair \((H,v_H)\) where \(H\) is a graph and \(v_H \in V(H)\) is a vertex of \(H\) that we name the \emph{center} of \(H\).
The  \emph{radius} of a centered graph is the maximum distance from \(v_H\) to any other node in \(H\).

\begin{definition}[Locally verifiable labeling problem]\label{def:lvl-problem}
    Let \(\lvlCR \in \nat\).
    Let  \(\inLabels\) and \(\outLabels\) two sets of input and output labels, respectively, and \(\problem\) a labeling problem.
    \(\problem\) is \emph{locally verifiable} with checking radius \(\lvlCR\) if there exists a (possibly infinite) family \( \SS = \{((H,v_H), \inLblRes, \outLblRes)_i\}_{i \in I}\) of tuples, where \((H,v_H)\) is a centered graph of radius at most \(\lvlCR\), \(\inLblRes : V(H) \to \inLabels\) is an input labeling for \(H\), \(\outLblRes : V(H) \to \outLabels\) is an output labeling for \(H\) (which can depend on \(\inLblRes\)) with the following property
    \begin{itemize}
        \item for any input \((G, \inLabel)\) to \(\problem\), an output vector \(\outLabel : V(G) \to \outLabels\) is permissible (i.e., \(\outLabel \in \problem((G,\inLabel))\)) if and only if, for each node \(v \in V(G)\), the tuple \(((G[\NN_\lvlCR(v)]), \inLabel \restriction_{\NN_{\lvlCR}(v)}, \outLabel \restriction_{\NN_{\lvlCR}(v)} )\) belongs to \(\SS\).
    \end{itemize}
\end{definition}

We remark that the notion of an (LVL) problem is a graph problem, and does not depend on the specific model of computation we consider (hence, the problem cannot depend on, e.g., node identifiers).
Next definition introduces the concept of outcome of an algorithm.

\begin{definition}[Outcome]
    Let \(\inLabels\) and \(\outLabels\) be two sets of input and output labels, respectively.
    An \emph{outcome} \(\outcome\) is a mapping \((G,\inptData) \mapsto \{(\indOutLabel{i}, \outPr_i)\}_{i \in I}\), with \(I\) being a discrete set of indices, assigning to every input graph \(G\) with any input data \(\inptData = (\id:V(G) \to [\abs{V}^c], \inLabel:V(G) \to \inLabels)\), a discrete probability distribution \(\{\outPr_i\}_{i \in I}\) over (not necessarily permissible) output vectors \(\indOutLabel{i} : V(G) \to \outLabels \) such that:
    \begin{enumerate}
        \item for all \(i \in I\), \(\outPr_i > 0\);
        \item \(\sum_{i \in I} \outPr_i = 1 \);
        \item \(\outPr_i\) represents the probability of obtaining \(\indOutLabel{i}\) as the output vector of the distributed system.
    \end{enumerate}
\end{definition}

Let \(T \ge 0\) be any integer. 
We say that an outcome \(\outcome\) on some graph family \(\FF\) has locality \(T\) if there exists a distributed algorithm in the \local model which, for any input \((G, \inptData)\)  where \(G \in \FF\), produces the same probability distribution over output vectors as \(\outcome\) after \(T\) rounds of computation.
Notice that \(T\) can be a function of the size of the input graph.

An algorithm can be thought of producing an output distribution on every input:
whenever the computations of a node in a given round are defined, the algorithm proceeds normally; 
if at some round some computation is undefined for a node, the node outputs some ``garbage label'', say \(\perp\): we remark that we can assume \(\outLabels\) always contains such a garbage label without loss of generalization.
Hence, an outcome can be always thought of as being defined on the family of all graphs and all valid inputs:
for this reason, we will omit specifying the graph family on which the outcome is defined.

We say that an outcome \(\outcome\) over some graph family \(\FF\) \emph{solves} problem \(\problem\) over \(\FF\) \emph{with probability \(p\)} if, 
for every \(G \in \FF\) and any input data \(\inptData = (\id, \inLabel)\), it holds that
\[
    \sum_{\substack{(\indOutLabel{i}, \outPr_i) \in \outcome((G,\inptData)) \ : \\ \indOutLabel{i} \in \problem((G,\inLabel))}} \outPr_i \ge p.  
\]
When \(p = 1\), we will just say that \(\outcome\) \emph{solves} problem \(\problem\) over the graph family \(\FF\).

We now define the complexity class \(\LL[T]\).
A problem \(\problem\) over some graph family \(\FF\) belongs to the class \(\LL[T]\) (respectively, \(\LL[T, p]\), for some \(p \in [0,1]\)) if there exists a distributed algorithm in the \local model which produces an outcome $\OO$ with locality at most \(T\)  which solves problem \(\problem\) on \(\FF\) (respectively, solves problem \(\problem\) with probability \(p\)).
In such case we write \((\problem, \FF) \in \LL[T]\) (respectively, \((\problem, \FF) \in \LL[T,p]\)).

The next computational model tries to capture the fundamental properties of any
\emph{physical} computational model (in which one can run either deterministic,
random, or quantum algorithms) that respects causality. 
The defining property of such a model is that, for any two (labeled) graphs
$(G_1,\inptData_{1})$ and $(G_2,\inptData_{2})$ that share some identical subgraph $(H,y)$, every
node $u$ in $H$ must exhibit identical behavior in $G_1$ and $G_2$ as long as
its \emph{local view}, that is, the set of nodes up to distance $T$ away from
$u$ together with input data and port numbering, is fully contained in $H$.
As the port numbering can be computed with one round of communication through a fixed procedure (e.g., assigning port numbers \(1,2, \dots, \deg(v)\) based on neighbor identifiers in ascending order) and we care about asymptotic bounds, we will omit port numbering from the definition of local view.

The model we consider has been introduced by \textcite{arfaoui2014} and is equivalent to the \philcl model by \textcite{gavoille2009}; 
however, as in \cite{arfaoui2014}, we explicitly require the outcome to
be defined for every possible graph:
in fact, as argued for distributed algorithms before, every physical procedure producing outcomes for graphs \emph{should} produce some outcome on any input (possibly, by using some garbage label as before).\footnote{We remark that in \cite{gavoille2009} it is ambiguous whether the outcome is defined over any possible input graph. Anyway, such ambiguity does not affect the validity of the proofs.}

In order to proceed, we first define the \emph{non-signaling} property of an outcome.
Let \(T \ge 0\) be an integer, and \(I\) a set of indices.
For any set of nodes \(V\), subset \(S \subseteq V\), and for any input \((G = (V,E), \inptData)\), we define its \emph{\(T\)-local view} as the set 
\[
    \VV_T(G, \inptData, S) = \left\{(u, \inptData(u)) \ \st \ \exists \ u \in V, \ v \in S \text{ such that } \dist_G(u,v) \le T \right\},
\]
where \(\dist_G(u,v)\) is the distance in \(G\).
Furthermore, for any subset of nodes \(S \subseteq V\) and any output distribution \(\{(\indOutLabel{i}, \outPr_i)\}_{i \in I}\), we define the \emph{marginal distribution} of \(\{(\indOutLabel{i}, \outPr_i)\}_{i \in I}\) on set \(S\) as the unique output distribution  \(\{(\resIndOutLabel{i}, \bar{\outPr}_i)\}_{i \in I}\) acting on \(S\) which satisfies the condition 
\[
    \bar{\outPr}_j = \sum_{i \ : \ \resIndOutLabel{j} = \indOutLabel{i}[S]} \outPr_i,
\]
where \(\indOutLabel{i}[S]\) is the restriction of output \(\indOutLabel{i}\) to the processes in \(S\).

\begin{definition}[Non-signaling outcome]\label{def:ns-outcome}
    An outcome \(\outcome: (G, \inptData) \mapsto \{(\indOutLabel{i}, \outPr_i )\}_{i \in I}\) is \emph{non-signaling} beyond distance \(T\) if for all set of nodes \(V\) and all subsets \(S \subseteq V\), for any pair of inputs \((G_1 = (V,E_1), \inptData_{1})\), \((G_2 = (V,E_2), \inptData_{2})\) such that \(\VV_T (G_1, \inptData_{1}, S) = \VV_T (G_2, \inptData_{2}, S)\), the output distributions corresponding to these inputs have identical marginal distributions on the set \(S\).  
    Notice that \(T\) can depend on the input labeled graph.
\end{definition}
\cref{def:ns-outcome} is also the more general definition for the locality of an outcome: an outcome \(\outcome\) has locality \(T\) if it is non-signaling beyond distance \(T\).

\paragraph{The \nslcl model.} 

The \nslclFull (\nslcl) model is a computational model that produces non-signaling outcomes. 
Let \(p \in [0,1]\).
The complexity class \(\NN\SS[T,p]\) is defined by all pairs \((\problem,\FF)\) where \(\problem\) is a problem and \(\FF\) is a graph family such that there exists an outcome \(\outcome\) that is non-signaling beyond distance \(T\) which solves \(\problem\) over \(\FF\) with probability at least \(p\).
If \(p = 1\), we just say that \((\problem,\FF) \in \NN\SS[T]\). 

As every (deterministic or randomized) algorithm running in time at most \(T\) in the \local model produces an outcome which has locality \(T\), we can provide lower bounds for the \local model by proving them in the \nslcl model.
For the sake of readability, we assume that every outcome \(\outcome \) that has locality \(T\) can be produced by a hypothetical non-signaling \local algorithm \(\AA\) with running time \(T\). 
This is just an artifact of the text and does not affect in any way the validity of our proofs.

We now present a lower bound technique that works for LVL problems in \nslcl.

\subsection{Lower bound technique}\label{sec:lb:indistinguishability}

We first introduce the notion of subgraph cover.

\begin{definition}\label{def:subgraph-cover}
    A \emph{subgraph cover} of a graph \(G\) is a family of subgraphs \(\{G_i\}_{i \in I}\) such that \(G_i \subseteq G\) for all \(i \in I\) and \(\cup_{i \in I}G_i = G\).
\end{definition}

\subsubsection{Indistinguishability argument in the classical \texorpdfstring{\local}{LOCAL} model}\label{sec:lb-technique:lcl}

Our technique is an extension of the \emph{indistinguishability argument} already exploited to prove lower bounds in the \local model \cite{kuhn2004,gavoille2009coloring,derbel2008}.
Let \(\problem\) be an LVL problem over some graph family \(\FF\).
The indistinguishability argument basically says that the output of a node \(v\) running a \(T\)-round algorithm \(\AA\) cannot distinguish between inputs that differ only outside its \(T\)-view.
Hence, \(\AA\) cannot solve the problem.
However, in the aforementioned works, all the different inputs considered to ``confuse'' the nodes were solvable instances.
Another approach, the one that we consider in this work, was introduced by \textcite{linial1992} to prove a lower bound for \(c\)-coloring trees: it uses graphs outside the input graph family (on which the problem is \emph{impossible} to be solved) which is locally everywhere a solvable instance.
\citeauthor{linial1992} used a high-girth graph \(G\) (which locally looks like a tree but lies \emph{outside} the input graph family) that has chromatic number bigger than \(c\).
This immediately yields a lower bound that is some constant fraction of \(\girth{G}\).
In this sense, the approach is purely existential graph-theoretic at heart.

We generalize this latter method all the way up to the \nslcl model, and we also present a technique to boost the failing probability of any outcome.
As our argument presents some technicalities, we proceed step by step and present it first for the \detlcl model, then for \randlcl, and finally for \nslcl.

The argument for the \detlcl model goes roughly as follows.

\paragraph{Indistinguishability argument: \detlcl model.}
Suppose we have an LVL problem \(\problem\), with checking radius \(t > 0\), that is solvable over some graph family \(\FF\), and assume we have a \detlcl algorithm \(\AA\) that solves \(\problem\) over \(\FF\) and has running time \(T(n) \ge t > 0\), \(n\) being the size of the input graph. We remark that \(T\) might also depend on other parameters of the input graph, such as the maximum degree. 
However, we omit such dependencies for the sake of readability.
Let us fix the size \(n\) of the input graph.
Suppose there exists a graph \(G_n \notin \FF\) such that \(\problem\) is not solvable over \(G_n\), and let us run \(\AA\) for \(T(n)\) rounds over \(G_n\) (\(G_n\) does not have necessarily size \(n\)---we force outputs after time \(T(n)\) if the protocol did not produce any or if, at any time, at any node the computational procedure is not well-defined).
As \(\problem\) is not solvable over \(G_n\), we know that there is a node \(v \in V(G_n)\) such that \(G_n[\NN_t(v)]\) contains some non-admissible output for \(\problem\).
Let us denominate \(G_n[\NN_t(v)]\) the \emph{bad neighbor}.
Assume now that there exists a graph \(H_n \in \FF\) of size \(n\) which contains a subgraph \(\tilde{H}_n\) such that \(H_n[\NN_{T(n)}(\tilde{H}_n)]\) is isomorphic to \(G_n[\NN_{T(n)} (\NN_t(v))]\), and assume \(H_n[\NN_{T(n)}(\tilde{H}_n)]\) and \(G_n[\NN_{T(n)} (\NN_t(v))]\) are given exactly the same identifiers and input labels.
As \(\tilde{H}_n\) and \(G_n[\NN_t(v)]\) look identical, \(\AA\) must have produced the same non-admissible output over \(\tilde{H}_n\) in time \(T(n)\), which is a contradiction.

This argument can be extended to the \randlcl model with some care:
while in the \detlcl model the local failure is deterministic and takes necessarily place in all graphs that locally look like the bad neighbor, this is not the case for \randlcl. 
We now show how to deal with random outputs.

\paragraph{Indistinguishability argument: \randlcl model.}
We keep the same hypothesis (except that now \(\AA\) is a \randlcl algorithm) and, in addition, we ask that the graph \(G_n \notin \FF\), over which \(\problem\) is not solvable, admits a subgraph cover \(\{G_n^{(i)}\}_{I \in I}\) with the following properties:
\begin{enumerate}[(1)]
    \item For each \(v \in G_n\), there exists \(i_v \in I\) such that \(\NN_t (v) \in G_{n}^{(i_v)}\);
    \item For each \(i \in I\), there exists a graph \(H_n^{(i)} \in \FF\) of size \(n\) which contains a subgraph \(\tilde{H}_n^{(i)}\) such that \(\tilde{H}_n^{(i)}\) is isomorphic to \(G_n^{(i)}\), and \(H_n^{(i)} [ \NN_{T(n)} ( \tilde{H}_n^{(i)} )]\) is isomorphic to \(G_n[\NN_{T(n)}(G_n^{(i)})]\).
\end{enumerate}
Again, we run \(\AA\) on \(G_n\), and we know that \(G_n\) will contain at least one \(t\)-neighborhood that has a non-admissible output vector.
As \(\{G_n^{(i)}\}_{I \in I}\) is a subgraph cover, there exists \(i^\star \in I\) such that the probability of \(\AA\) failing over \(G_n^{(i^\star)}\) is at least \(1/\abs{I}\).
We denominate \(G_n^{(i^\star)}\) the \emph{bad subgraph}.
Hence, assuming  \(H_n^{(i)} [ \NN_{T(n)} ( \tilde{H}_n^{(i)} )]\) and \(G_n[\NN_{T(n)}(G_n^{(i)})]\) are given the same identifiers and input labels, \(\AA\) fails on \(H_n^{(i^\star)}\) with probability at least \(1/\abs{I}\).
We observe that property (2) is sufficient but not necessary to the technique: it is sufficient to ensure the existence of the graph \(H_n ^{(i)}\) for \(i = i^{\star}\).
Nevertheless, in many practical scenarios actually determining \(i^\star\) is hard, while it is easier to ensure (2) for many graph families.

This result is useful when \(\abs{I}\) is not too large: 
however, in many cases, it is not possible to find subgraph covers with few elements; 
furthermore, one may want a failure probability that is higher than a constant value.
Luckily, other properties of the \randlcl model come to our aid and sometimes allow to boost the failing probability.
Let \(N \in \natPos\).
The idea is to replicate \(N\) times the \(T(n)\)-neighborhood of the bad subgraph (making sure we still obtain a graph that belongs to the graph family under consideration) and exploit the independence of the outcome generated by \(\AA\) over subsets of the nodes that are ``far enough''.
More formally, assume that \(\NN_{T(n)}(G_n^{(i)})\) has size at most \(\floor{n/N}\) for each \(i \in I\). 
Let us replace property (2) as follows: 
\begin{itemize}
    \item[(2)] For each choice of indices \(\xx_N = (x_1, \dots, x_N) \in [\abs{I}]^N\), there exists a graph \(H_{\xx_N} \in \FF\) of size \(n\) which contains a subgraph \(\tilde{H}_{\xx_N}\) such that \(\tilde{H}_{\xx_N}\) is isomorphic to the disjoint union \(\bigsqcup_{j = 1}^N G_n^{(x_j)}\), and \(H_{\xx_N}[\NN_{T(n)}(\tilde{H}_{\xx_N})]\) is isomorphic to the disjoint union \(\bigsqcup_{j = 1}^N G_n[\NN_{T(n)}(G_n^{(x_j)})]\).
\end{itemize}
By independence, the probability that \(\AA\) solves \(\problem\) over \(H_{\xx_N}\) is now 
\(
    (1 - {1}/{\abs{I}})^N
\) (assuming again the same identifiers and input labels are given to \(H_{\xx_N}[\NN_{T(n)}(\tilde{H}_{\xx_N})]\) and \(\bigsqcup_{j = 1}^N G_n[\NN_{T(n)}(G_n^{(x_j)})]\)).

In both arguments for the \detlcl and the \randlcl model, we denominate the graph \(G_n\) as the \emph{cheating graph}, because it allows us to ``trick'' the distributed algorithm, since nodes cannot distinguish between different inputs if they have the same local view.

\subsubsection{\texorpdfstring{Indistinguishability argument in the \nslcl model}{Indistinguishability argument in the NS-LOCAL model}}

The argument outlined in \cref{sec:lb-technique:lcl} cannot be directly applied to the \nslcl model for two reasons:
\begin{enumerate}
    \item Independence is not guaranteed between far away subsets of nodes (e.g., there could be some shared resources).
    \item We cannot consider an outcome over two graphs \(G, H\) of different sizes and require it to have the same output distribution over two subgraphs \(G' \subseteq G, H' \subseteq H\) that have the same local neighborhood due to the no-cloning principle \cite{wootters1982,dariano2017,masanes2006} (in fact, the properties of a non-signaling outcome hold only for graphs of the same size; see \cref{def:ns-outcome}).
\end{enumerate}

However, we overcome these issues and show that:
\begin{enumerate}
    \item The dependencies actually go ``in the right direction'', i.e., the bound on the failing probability does not decrease w.r.t.\ the bound we showed in the \randlcl model;
    \item We can restrict ourselves to graphs of same sizes in many cases, as we show in the applications of our lower bound technique (\cref{sec:coloring:lb,sec:coloring:grids,sec:coloring:trees}).
\end{enumerate}

Some technicalities are required to obtain 1.
This section is devoted to the formal proof of our lower bound technique in the \nslcl model.
The technique applies to LVL problems restricted to graph families meeting some specific properties.

\begin{definition}[Cheating graph]\label{def:cheating-graph}
    Let \(\problem\) be any LVL problem with checking radius \(\lvlCR\in \nat\) that is solvable over some graph family \(\FF\).
    Suppose that, for some integer \(n \in \nat\), there exists a triple \((k,N,T) \in \nat^3\) (that can depend on \(n\) and, possibly, other parameters defining the graph family \(\FF\)), with \(T \ge \lvlCR\), and a graph \(G_n\) of size at most \(\floor{n / N}\), such that the following properties are met:
    \begin{enumerate}[(i)]
        \item \(\problem\) is not solvable on \(G_n\);
        \item \(G_n\) has a subgraph cover \(\{G_n^{(1)}, \dots G_n^{(k)}\}\) such that
        \begin{enumerate}[(a)]
            \item for each \(v \in V(G_n)\), there exists \(j \in [k]\) such that \(\NN_{\lvlCR}(v) \subseteq V(G_n^{(j)})\);
            \item for each choice of indices \(\xx_N = (x_1, \dots, x_N) \in [k]^N\), there exists a graph \(H_{\xx_N} \in \FF\) of size \(n\) which contains a subgraph \(\tilde{H}_{\xx_N}\) such that \(\tilde{H}_{\xx_N}\) is isomorphic to the disjoint union \(\bigsqcup_{j = 1}^N G_n^{(x_j)}\), and \(H_{\xx_N}[\NN_T(\tilde{H}_{\xx_N})]\) is isomorphic to the disjoint union \(\bigsqcup_{j = 1}^N G_n[\NN_T(G_n^{(x_j)})]\).
        \end{enumerate}
    \end{enumerate}
    Then we say that \(G_n\) is an \((n,k,N,T)\)-\emph{cheating graph} for the pair \((\problem, \FF)\) or, more generally, that \(\FF\) admits an \((n,k,N,T)\)-\emph{cheating graph} for \(\problem\).
\end{definition}

\begin{remark}\label{rem:solvability-cheatgraph-family}
    Being \(\problem\) an LVL problem, \cref{def:cheating-graph}.(ii).(b) implies that \(\problem\) is solvable on \(G_n[\NN_T(G_n^{(i)})]\) for \(i = 1, \dots, k\).
\end{remark}

We now present our general lower bound theorem.

\begin{theorem}\label{thm:lb-technique}
    Let \(\problem\) be an LVL problem with checking radius \(\lvlCR\), and \(\FF\) be a graph family that admits an \((n,k,N,T)\)-{cheating graph} for \(\problem\).
    Suppose \(\outcome\) is an outcome  over \(\FF \) in \nslcl with locality \(T \ge \lvlCR\).
    Then, there exists a graph \(H \in \FF\) on \(n\) vertices such that the probability of \(\outcome\) solving \(\problem\) on \(H\) is at most \((1 - 1/k)^N\).
    Furthermore, \(H\) can be chosen among the graphs in the family \(\{H_{\xx_N}\}_{\xx_N \in [k]^N}\) given by \cref{def:cheating-graph}.(ii).(b). 
\end{theorem}
\begin{proof}
    Let \(G_n\) be a \((n,k,N,T)\)-{cheating graph} for \((\problem, \FF)\).
    We know that \(G_n\) has size at most \( \floor{n/N}\) and satisfies the properties listed in \cref{def:cheating-graph}.
    Now, consider a new graph that consists of \(N\) disjoint copies \(G_{n,1}, \dots, G_{n, N}\) of \(G_{n}\), and \(n -  \abs{V(G_n)}\) isolated nodes.

    For each \(i = 1, \dots, N\), consider the subgraph cover \(\{G_{n, i}^{(j)}\}_{j \in [k]}\) for \(G_{n, i}\) given by \cref{def:cheating-graph}.
    Let \(\outcome\) be any outcome having locality \(T\) and solving problem \(\problem\) over \(\FF\).

    As \(\problem\) is not solvable on \(G_{n, i}\), then the failing probability of \(\outcome\) over \(G_{n, i}\) is 1, for each \(i = 1, \dots, N\).
    Consider one of the \(G_{n, i}\) and notice that, if \(\outcome\) produces a permissible vector output on \(G_{n, i}[\NN_T(G_{n, i}^{(j)})]\) for each \(j \in [k]\) at the same time, then, by definition of LVL problem, we have a global permissible vector output on \(G_{n, i}\) (which does not exist).
    Hence, by \cref{def:cheating-graph}.(ii).(a), there must exist \(j \in [k]\) and \(v \in V(G_{n, i}^{(j)})\) such that \(\NN_{\lvlCR}(v) \subseteq V(G_{n, i}^{(j)})\) and the output vector on \(\NN_{\lvlCR}(v)\) is not permissible: in such a case, we say that \(G_{n, i}^{(j)}\) contains a \emph{bad node}.

    We now prove that there exists a sequence of indices \(\xx_N = (x_1, \dots, x_N) \in [k]^N\) such that \(\outcome\) produces a bad node in \(\bigsqcup_{j = 1}^N G_{n,j}^{(x_j)}\) with probability at least \(1 - (1 - 1/k)^N\).
    If we had independence between ``far away'' parts of the graphs (as in the \randlcl model), this thesis would be trivial (see \cref{sec:lb-technique:lcl}). 
    However, in the \nslcl model non-trivial dependencies are possible (e.g., pre-shared quantum state).
    
    We here present a shorter proof by induction on \(N\), as we already gave a somewhat ``constructive'' intuition in  \cref{sec:intro:lb}.
    Assume \(N=1\): as \(\outcome\) fails on \(G_{n,1}\) with probability 1, and the latter is covered by \(\{G_{n, 1}^{(j)}\}_{j = 1}^k\), then there exists an index \(x_1 \in [k]\) such that the probability that \(G_{n, 1}^{(x_1)}\) contains a bad node is at least \(1/k\).
    Now, assume \(N > 1\) and the claim to be true for \(N-1\).
    Let \(\EE_{i}^{(j)}\) be the event that \(\outcome\) produces a bad node in \(G_{n,i}^{(j)}\) (we remark that \(\problem\) is solvable on \(G_{n,i}^{(j)}\) by \cref{rem:solvability-cheatgraph-family}):
    the inductive hypothesis can be rewritten as \(\pr{\cup_{i = 1}^N \EE_{i}^{(x_i)}} = 1 - (1 - 1/k)^{N-1} + y\) for some \(y \ge 0\). 
    Assume \(\pr{\cup_{i = 1}^{N-1} \EE_{i}^{(x_i)}} < 1\) otherwise the thesis is trivial.

    Let us denote the complement of any event \(A\) by \(\bar{A}\).
    As \(\outcome\) fails on \(G_{n,N}\) with probability 1, we know that
    \begin{align*}
        1 & = \pr{\cup_{j = 1}^k \EE_{N}^{(j)}} \\
        & = \pr{\cup_{j = 1}^k \EE_{N}^{(j)} \bigcup (\cup_{i = 1}^{N-1} \EE_{i}^{(x_i)})}.
    \end{align*}
    By the law of total probability, we get that
    \begin{align*}
       1 & = \pr{(\cup_{j = 1}^k \EE_{N}^{(j)}) \bigcup (\cup_{i = 1}^{N-1} \EE_{i}^{(x_i)})} \\
       & = \pr{\cup_{j = 1}^k \EE_{N}^{(j)} \ \st \ \cap_{i = 1}^{N-1} \bar{\EE}_{i}^{(x_i)}} \pr{\cap_{i = 1}^{N-1} \bar{\EE}_{i}^{(x_i)}} + \pr{\cup_{i = 1}^{N-1} {\EE}_{i}^{(x_i)}}.
    \end{align*}
    Hence, \(\pr{\cup_{j = 1}^k \EE_{N}^{(j)} \ \st \ \cap_{i = 1}^{N-1} \bar{\EE}_{i}^{(x_i)}} = 1\); 
    by the union bound, it follows there exists \(x_N \in [k]\) with \(\pr{\EE_{N}^{(x_N)} \ \st \ \cap_{i = 1}^{N-1} \bar{\EE}_{i}^{(x_i)}} \ge 1/k\).

    Then, by the inclusion-exclusion principle and the law of total probability,
    \begin{align*}
        \pr{\cup_{i = 1}^N \EE_{i}^{(x_i)}} & = \pr{\EE_N^{(x_N)}} + \pr{\cup_{i = 1}^{N-1} \EE_{i}^{(x_i)}} - \pr{\EE_N^{(x_N)} \bigcap (\cup_{i = 1}^{N-1} \EE_{i}^{(x_i)})} \\
        & = \pr{\EE_N^{(x_N)} \bigcap (\cap_{i = 1}^{N-1} \bar{\EE}_{i}^{(x_i)}) } + \pr{\cup_{i = 1}^{N-1} \EE_{i}^{(x_i)}} \\
        & = \pr{\cup_{j = 1}^k \EE_{N}^{(j)} \ \st \ \cap_{i = 1}^{N-1} \bar{\EE}_{i}^{(x_i)}} \pr{\cap_{i = 1}^{N-1} \bar{\EE}_{i}^{(x_i)}} + \pr{\cup_{i = 1}^{N-1} \EE_{i}^{(x_i)}} \\
        & \ge \frac{1}{k}\cdot \left[\left(1 - \frac{1}{k} \right)^{N-1} - y \right] + 1 - \left(1 - \frac{1}{k} \right)^{N-1} + y \\
        & = 1 - \left(1 - \frac{1}{k} \right)^{N} + y \left(1 - \frac{1}{k} \right) \\
        & \ge 1 - \left(1 - \frac{1}{k} \right)^{N},
    \end{align*}
    proving the claim.

    By \cref{def:cheating-graph}.(ii).(b), there exists a graph \(H_{\xx_N} \in \FF\) over the same set of \(n\) nodes that contains a subgraph \(\tilde{H}_{\xx_N}\) with \(H_{\xx_N}[\NN_T(\tilde{H}_{\xx_N})]\) being the same graph as \(\bigsqcup_{i=1}^N G_{n,i}[\NN_T(G_{n,i}^{(x_i)})]\).
    Consider the same identifiers and input labels over \(H_{\xx_N}[\NN_T(\tilde{H}_{\xx_N})]\) and \(\bigsqcup_{i=1}^N G_{n,i}[\NN_T(G_{n,i}^{(x_i)})]\):
    by the definition of \nslcl, the probability that \(\outcome\) fails on \(\tilde{H}_{\xx_N} \subseteq H_{\xx_N}\) is the same as that on  \(\bigsqcup_{j = 1}^N G_{n,j}^{(x_j)}\), yielding the thesis.
\end{proof}

As long as one can find a cheating graph for a pair \((\problem,\FF)\), where \(\problem\) is an LVL problem and \(\FF\) an input graph family, the lower bound technique can be applied.
In \cref{sec:coloring:lb,sec:coloring:grids,sec:coloring:trees}, all the analysis that is carried out serves to show that there exists such a cheating graph for, respectively, \(c\)-coloring \(\chi\)-chromatic graphs, \(3\)-coloring grids, and \(c\)-coloring trees. \subsection{Lower bound for \texorpdfstring{$c$}{c}-coloring \texorpdfstring{$\chi$}{𝜒}-chromatic graphs} \label{sec:coloring:lb}

The goal of this section is to prove \cref{thm:lb-coloring} by showing that the family of \(\chi\)-chromatic graphs admits cheating graphs for the \(c\)-coloring problem (\cref{def:cheating-graph}). 
We restate the theorem for the sake of readability.

\thmMainResultLB*

We base our analysis on \cite[Theorem 1.2]{bogdanov2013}, a result that has gone relatively unnoticed and lies at the intersection between graph theory, combinatorics, and topology, which ensures the existence of a graph with high chromatic number which, locally, is \(\chi\)-chromatic.
The first half of the sections aims at constructing such graph, and the second half is devoted to the proof that this graph admits a (small enough) subgraph cover that satisfies the properties of \cref{def:cheating-graph}.

\paragraph{Preliminaries.}
We first define some graph operations.
For any two graphs \(G, H\), we define the intersection graph \(G \cap H\) as a graph whose vertex set is the set \(V(G) \cap V(H)\), and whose edge set is the set \(E(G) \cap E(H)\).
Similarly, we define the union graph \(G \cup H\) as a graph whose vertex set is the set \(V(G) \cup V(H)\), and whose edge set is the set \(E(G) \cup E(H)\).
We define the difference graph \(G \setminus H\) as the subgraph of \(G\) induced by \(V(G) \setminus V(H)\).
The \emph{\(T\)-local chromatic number} of a graph \(G\), denoted by \(\lclCHR{T}(G)\), is the minimum \(c \in \nat\) such that the graph induced by the \(T\)-neighborhood of any node is \(c\)-colorable.
More formally
\[
    \lclCHR{T}(G) = \min \left\{c \in \nat \ \st \  \forall u \in V, \, G [{\NN_t(u)}] \text{ is } c \text{-colorable} \right\}.   
\]
Given two graphs \(G\) and \(H\), a function \(f:V(G) \to V(H)\) is a homomorphism from \(G\) to \(H\) if, for any \(\{u,v\} \in E(G)\), \(\{f(u),f(v)\} \in E(H)\). 
A homomorphism from \(G\) to \(K_c\), the \(c\)-clique, is equivalent to saying that \(G\) is \(c\)-colorable.
Notice that the composition of homomorphisms is a homomorphism: hence, if \(G\) is homomorphic to \(H\), then \(\XX(H) \ge \XX(G)\).
Furthermore, we define the \emph{tensor product} of graphs \(G\) and \(H\) as a graph \(G \times H\) whose vertex set is \(V(G) \times V(H)\), and whose edge set is determined by the following: for any \((g,h), (g',h') \in V(G \times H)\), \(\{(g,h), (g',h')\} \in E(G \times H)\) iff \(gg' \in E(G)\) and \(hh' \in E(H)\) (see \cref{fig:tensor-product} for an example).

We hereby state \cite[Theorem 1.2]{bogdanov2013}.

\begin{theorem}[\cite{bogdanov2013}]\label{thm:graph-theory:lb}
    Let \(\chi \ge 2\), \( r \ge 2\), and \(k \ge 1\) be integers. 
    There exists a graph \(G_k = (V, E)\) 
    such that \(\lclCHR{r}(G_k) = \chi\) and \(\XX(G_k) \ge k(\chi-1) + 1\) with 
    \[
        \abs{V} = \frac{(2r\chi + 1)^k - 1}{2r}.
    \] 
\end{theorem}

\begin{remark}\label{remark:bipartite-case}
    \cite[Theorem 1.2]{bogdanov2013} has been stated for \(\chi \ge 3\) since the result for \(\chi = 2\) was already known from a different construction provided by \cite{stiebitz1985} (a proof translated in English was reproduced by \cite{gyarfas2004}).
    Nevertheless, the proof of \cite[Theorem 1.2]{bogdanov2013} also holds for the case \(\chi = 2\).
\end{remark}

\begin{remark}\label{remark:construction-radius}
    \cref{thm:graph-theory:lb} holds even for \(r \ge 1\), as explicitly mentioned at the end of \cite[Section 4]{bogdanov2013}, slightly changing the proof.
\end{remark}

We will discuss the tightness and the related works of this result already in \cref{ssec:graph-theory-related-works}. 
We proceed proving \cref{thm:lb-coloring}.
The idea of the whole proof is to show that the graph from \cref{thm:graph-theory:lb} provides a cheating graph for the \(c\)-coloring graphs problem and the family of \(\chi\)-chromatic graphs.

We now construct the graph from \cref{thm:graph-theory:lb}.

\paragraph{The \(r\)-join of graphs.}

\begin{figure}
    \centering
    \includegraphics{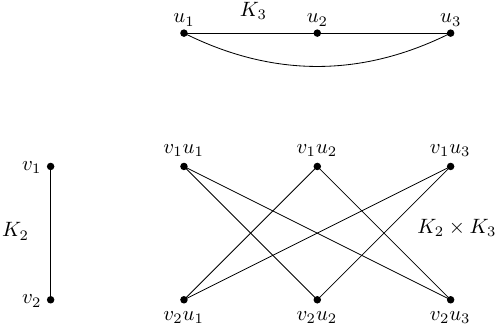}
    \caption{Tensor product \(K_2 \times K_3\).}
    \label{fig:tensor-product}
\end{figure}

\begin{figure}
  \centering
  \begin{subcaptionblock}{\textwidth}
    \centering
    \includegraphics[scale=0.9]{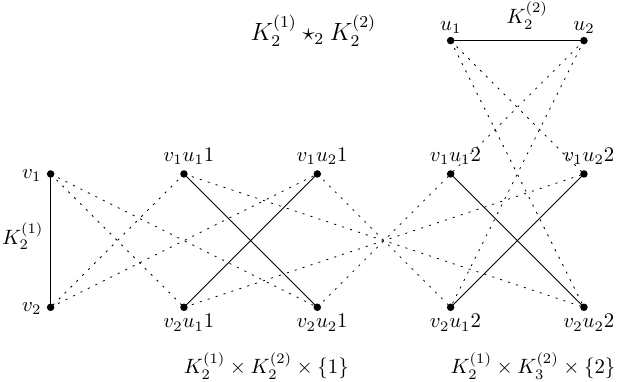}
    \caption{ The \(2\)-join of \(K_2^{(1)}\) and \(K_3^{(2)}\), with all the
    connections.
    Full lines represent edges within the tensor product graphs plus the
    starting and ending graph. Dotted lines represent edges among these graphs.
    \(K_2^{(1)} \times K_2^{(2)} \times \{1\}\) and \(K_2 ^{(1)}\times K_2^{(2)}
    \times \{2\}\) are two copies of the tensor product \(K_2^{(1)} \times
    K_2^{(2)}\).}
  \end{subcaptionblock}
  \\[5mm]
  \begin{subcaptionblock}{\textwidth}
    \centering
    \includegraphics[scale=0.9]{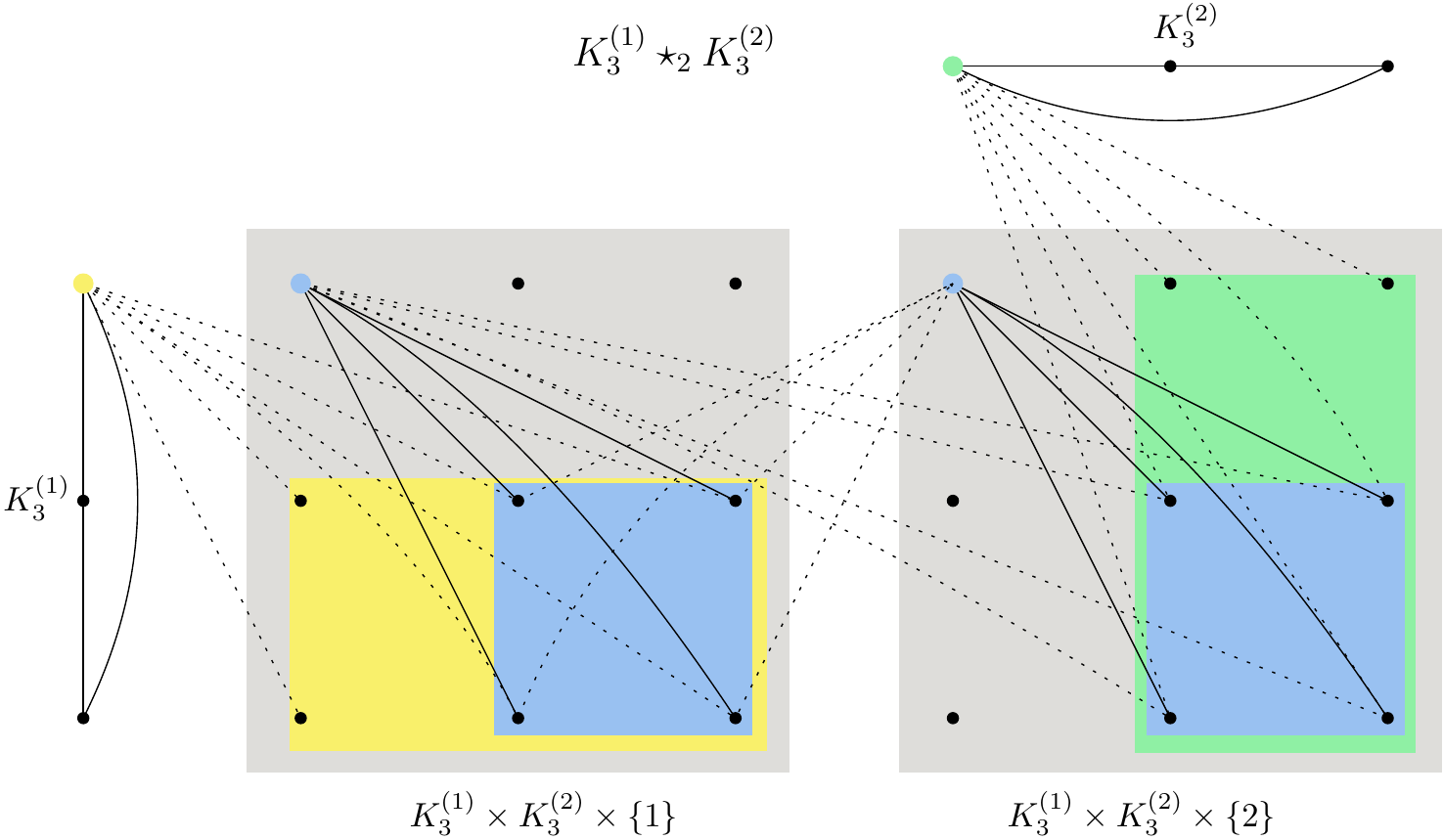}
    \caption{Representation of the \(2\)-join of two copies \(K_3^{(1)}\) and
    \(K_3^{(2)}\) of \(K_3\). Here, for the sake of visibility, only some edges
    are represented. The yellow node in \(K_3^{(1)}\) is connected to all nodes in
    the yellow area of \(K_3^{(1)} \times K_3^{(2)} \times \{1\}\); similarly, the
    green node in \(K_3^{(2)}\) is connected to all nodes in the green area of
    \(K_3^{(1)} \times K_3^{(1)} \times \{2\}\). The blue nodes of \(K_3^{(1)}
    \times K_3^{(2)} \times \{1,2\}\) are connected to all nodes in the blue areas
    of \(K_3^{(1)} \times K_3^{(2)} \times \{1,2\}\). Other connections can be
    deduced by symmetry.}
  \end{subcaptionblock}
\caption{Examples for the $2$-join of graphs.}
  \label{fig:graph-join}
\end{figure}

Given two graphs \(G\) and \(H\), we aim to define the \(r\)-join operation \(G \rjoin{r} H\), for any \(r \ge 0\).
The vertex set \(V(G \rjoin{r} H)\) is defined by \(V(G \rjoin{r} H) = V(G) \cup V(G)  \times V(H) \times \{1, \dots, r\} \cup V(H)\). 
Let 
\begin{align*}
    E_{1,r}(G \rjoin{r} H) =  \left\{\{(g,h,i),(g',h',j)\} \ \st \ g,g' \in V(G), h,h' \in V(H), gg' \in E(G), hh' \in E(H), \abs{i - j} \le 1\right\}.
\end{align*}
Furthermore, let
\[
    E_0(G \rjoin{r} H) = E(G) \cup \left\{ \{g, (g',h',1)\} \ \st \ g,g' \in V(G), h' \in V(H), gg' \in E(G) \right\}
\]
and 
\[
    E_{r+1}(G \rjoin{r} H) = E(H) \cup \left\{ \{(g,h,t), h'\} \ \st \ g \in V(G), h,h' \in V(H), hh' \in E(H) \right\}.
\]
Then, the edge set \(E(G \rjoin{r} H)\) is defined by \(E(G \rjoin{r} H) = E_0(G \rjoin{r} H) \cup E_{1,r}(G \rjoin{r} H) \cup E_{r+1}(G \rjoin{r} H)\).

An intuitive visualization of this graph follows: take a sequence of \(r+2\) disjoint copies \((G \times H)_0, (G \times H)_1, \dots, (G \times H)_{r+1}\) of the tensor product \(G \times H\) (an example of a tensor product graph is given in \cref{fig:tensor-product}).
Clearly, for any \(0 \le i \le r\), there is an isomorphism \(f_i:(G \times H)_i \to (G \times H)_{i+1} \).
Then, any two nodes \((g,h) \in (G \times H)_i\) and \((g',h') \in (G \times H)_{i+1}\) are connected if and only if \(\{f_i((g,h)), (g',h')\} \in E((G \times H)_{i+1})\), or, equivalently, \(\{(g,h), f_i^{-1}((g',h'))\} \in E((G \times H)_{i})\).
Finally, ``collapse'' \((G \times H)_0\) into \(G\) and \((G \times H)_{r+1}\) into \(H\) by merging nodes (merging multiple edges and deleting self-loops).

This join operation in graphs is some kind of ``discrete'' variant of the join operation between two topological spaces (see \cite{bogdanov2013} or \cite{matousek2010}).

We define two projection operators for the join of graphs.
\begin{definition}\label{def:projection-operators}
    Let \(\proj{G} : G \rjoin{r} H \setminus H \to G\) and \(\proj{H} : G \rjoin{r} H \setminus G \to H\) be defined as follows: \(\proj{G}((g,h,i)) = g\) and \(\proj{H}((g,h,i)) = h\) for \((g,h,i) \in V(G)  \times V(H) \times \{1, \dots, r\}\), while \(\proj{G} \restriction_{G}\) and \(\proj{H} \restriction_{H}\) are the identity maps on \(G\) and \(H\), respectively.
\end{definition}

\begin{remark}\label{remark:projection-homomorphisms}
    The two projections are homomorphisms, implying that the chromatic number of \(G \rjoin{r} H \setminus H\) is the same as that of \(G\), and the chromatic number of \(G \rjoin{r} H \setminus G\) is the same as that of \(H\).
\end{remark}

The join of two connected graphs results in a connected graph.

\begin{lemma}\label{lem:gadget:no-isolated-nodes}
    Let \(r \ge 1\), and let \(G\) and \(H\) be two connected graphs with at least two nodes each.
    Then, \(G \rjoin{r} H\) is connected.
\end{lemma}
\begin{proof}
    Consider any node \(u\) in \(G\rjoin{r} H\) which does not belong to \(G \cup H\).
    Then \(u = (g,h,i)\) for \(g \in V(G)\), \(h \in V(H)\), and \(i \in [r]\).
    There exist \(g' \in V(G)\) and \(h' \in V(H)\) such that \(gg' \in E(G)\) and \(hh' \in E(H)\). 
    We now construct two paths \(v_0v_1\dots v_i\) and \(w_{r+1}w_r\dots w_i\) that connect \(G\) and \(H\) to \(v_i = w_i = u\), respectively.
    Suppose \(i\) is odd.
    Then, set \(v_0 = g'\), \(v_j =(g,h,j)\) for any odd \(j\), and \(v_j = (g',h',j)\) for any even \(j\).
    We have that \(v_j\) is connected to \(v_{j+1}\) for any \(0 \le j \le i-1\), and \(v_i = u\).
    Similarly, set \(w_{r+1} = h'\), \(w_j =(g,h,j)\) for any odd \(j\), and \(w_j = (g',h',j)\) for any even \(j\).
    We have that \(w_j\) is connected to \(w_{j-1}\) for any \(i+1 \le j \le r+1\), and \(w_i = u\).
    Suppose \(i\) is even.
    Then, set \(v_0 = g\), \(v_j =(g',h',j)\) for any odd \(j\), and \(v_j = (g,h,j)\) for any even \(j\).
    We have that \(v_j\) is connected to \(v_{j+1}\) for any \(0 \le j \le i-1\), and \(v_i = u\).
    Similarly, set \(w_{r+1} = h\), \(w_j =(g',h',j)\) for any odd \(j\), and \(w_j = (g,h,j)\) for any even \(j\).
    We have that \(w_j\) is connected to \(w_{j-1}\) for any \(i+1 \le j \le r+1\), and \(w_i = u\).
\end{proof}

We are ready to define the graph from \cref{thm:graph-theory:lb} which we will prove to be a cheating graph for the family of \(\chi\)-chromatic graphs.

\begin{definition}[The construction]\label{def:graph:gadget}
Let \(\chi \ge 2\), \(r \ge 2\), and \(k \ge 1\). 
Consider a sequence \(K_\chi^{(1)}, \dots, K_a^{(k)}\) of \(k\) disjoint copies of \(K_\chi\).
We construct the graph recursively.
Let \(G_1 = K_\chi^{(1)}\) be the clique with \(\chi\) nodes.
Then, for \(k \ge 2\), \(G_k = G_{k-1} \rjoin{2r} K_\chi^{(k)}\).
\cite{bogdanov2013} proved that \(\lclCHR{r}(G_k) = \chi\), \(\XX(G_k) \ge k(\chi-1)+1\), and 
\[
    \abs{V(G_k)} = \frac{(2r\chi + 1)^k - 1}{2r} .
\]
See \cref{fig:graph-join} for some examples (notice that \(r = 1\) in the examples: as observed in \cref{remark:construction-radius}, the result still holds in this case).
\end{definition}    
Further discussion on \(G_k\) is deferred to \cref{sec:remarks-graph-theoretic-constr}.
In order to continue, we state the following lemma for an induced subgraph, whose proof is trivial.
\begin{lemma}\label{lem:graph:connected-induced-subgraph}
    Let \(T \in \nat\).
    Let \(G\) be a connected graph, and \(H \subseteq G\) be a connected subgraph of \(G\).
    Then, \(G [{\NN_T(H)}]\) is connected.
\end{lemma}
For a graph \(G\) and any node \(v \in V(G)\), we define by \(\dist_G(v,H) = \min_{u \in V(H)}\left\{\dist_G(u,v)\right\}\) the distance between \(v\) and any subgraph \(H \subseteq G\). 
We write \(\dist(v,H)\) when the underlying graph \(G\) is clear from the context.
We now prove some key-properties of \(G_k\) which allow us to show that \(G_k\) is a cheating graph. 

\begin{lemma}\label{lem:gadget:partition}
    Let \(\chi \ge 2\), \(r \ge 3\), and \(k \ge 2\) be integers. 
    Define \(T = \floor{\frac{2r}{3}}\).
    Let \(G_k\) be the graph defined in \cref{def:graph:gadget} built with copies of \(K_\chi\).
    There exists a subgraph cover \(\{G_{k}^{(i)}\}_{i \in [k]}\) of \(G_{k}\) such that the following statements hold:
    \begin{enumerate}[(i)]
        \item The chromatic number of \(G_{k}[{\NN_T(G_{k}^{(i)})}]\) is \(\chi\) for all \(i \in [k]\);
        \item \(G_{k}^{(i)}\) is connected for all \(i \in [k]\);
        \item for each \(v \in V(G_k)\), there exists \(i \in [k]\) such that \(\NN_1(v) \subseteq V(G_{k}^{(i)})\);
        \item \(G_k[\NN_T(G_k^{(i)})]\) contains at least one node at distance \(T\) from \(G_k^{(i)}\) for all \(i \in [k]\).
    \end{enumerate} 
\end{lemma}

\begin{proof}
    We prove the thesis by induction on \(k\).
    Remember that \(G_2\) is obtained by the \(2r\)-join of two disjoint copies \(K_\chi^{(1)}\), \(K_\chi^{(2)}\) of \(K_\chi\).
    Assume \(K_\chi^{(1)}\) is connected to the first copy of \(K_\chi^{(1)} \times K_\chi^{(2)}\), and \(K_\chi^{(2)}\) to the last.
    Let 
    \begin{align*}
        V_2^{(1)} & = V(K_\chi^{(1)}) \cup V(K_\chi^{(1)}) \times  V(K_\chi^{(2)}) \times \left\{1,\dots, T+2\right\}; \\
        V_2^{(2)} & = V(K_\chi^{(2)}) \cup V(K_\chi^{(1)}) \times  V(K_\chi^{(2)}) \times \left\{T+1, \dots, 2r \right\}.
    \end{align*}
    Consider the graphs \(G_2^{(1)} = G_2 [{V_2^{(1)}}]\) and \(G_2^{(2)}= G_2 [{V_2^{(2)}}]\).
    The cover property and properties (i)-(iii) are straightforward.
    Clearly, \(G_2 = \cup_{i \in [2]} \ G_2^{(i)}\).
    Furthermore, \(\XX(G_{2} [{\NN_T(G_{2}^{(1)})}]) = \XX(G_{2} [{\NN_T(G_{2}^{(2)})}]) = \chi\) 
    as the projections \(\proj{K_\chi^{(1)}} [{\NN_T(G_{2}^{(1)})}] ,\proj{K_\chi^{(2)}} [{\NN_T(G_{2}^{(2)})}] \) are homomorphisms.
    Moreover, it is easily verifiable that \(G_2^{(1)}\) and \(G_2^{(2)}\) are connected graphs:
    Observe that \(G_{2}^{(1)} =( K_\chi^{(1)} \rjoin{2r} K_\chi^{(2)}) [{\NN_{T+2}(K_\chi^{(1)})}] \).
    As both \(K_\chi^{(1)}\) (by the inductive hypothesis) and \(K_\chi^{(2)}\) are connected, \cref{lem:gadget:no-isolated-nodes} implies their join is connected.
    Then, \cref{lem:graph:connected-induced-subgraph} implies \(G_{2}^{(1)}\) is connected.
    The same applies for  \(G_{2}^{(2)}\) by observing that \(G_{k}^{(2)} = (K_\chi^{(1)} \rjoin{2r} K_\chi^{(2)}) [{\NN_{2r - T}(K_\chi^{(2)})}] \).
    As for property (iv), consider any two nodes \(u \in V(K_\chi^{(1)}) \times  V(K_\chi^{(2)}) \times \left\{2T+2\right\}\) and \(v \in  V(K_\chi^{(1)}) \times  V(K_\chi^{(2)}) \times \left\{1\right\}\).
    Clearly, \(u \in G_2[\NN_T(G_2^{(1)})]\) and has distance \(T\) from \(G_2^{(1)}\), while \(v \in G_2[\NN_T(G_2^{(2)})]\) and has distance \(T\) from \(G_2^{(2)}\).

    Let \(k \ge 3\) and assume the thesis is true for \(G_{k-1}\).
    We now construct the subgraph cover of \(G_k = G_{k-1} \rjoin{2r} K_\chi\).
    For \(i = 1, \dots, k-1\), we define \(G_k^{(i)}\) by 
    \[
        G_k^{(i)} = G_{k-1}^{(i)} \cup \left(G_k [{\NN_{T+2}(G_{k-1}^{(i)})}] \ \cap \ G_k [{V(G_{k-1}^{(i)})\times V(K_\chi^{(k)}) \times \{1,\dots, T+2\}}]\right)
    \]
    Then, we define
    \[
        G_k^{(k)} = G_k [{V(G_{k-1})\times V(K_\chi^{(k)}) \times \{T+1,\dots,2r\}}] \cup K_\chi^{(k)}.
    \]
    We now prove that the family \(\{G_k^{(i)}\}_{i \in [k]}\) respects properties (i)-(iv).

    \paragraph{Subgraph covering.} 
    Any node \(v \in(V(G_{k-1}) \cup V(K_\chi^{(k)}))\) belongs either to \((\cup_{i \in [k-1]} G_k^{(i)}) \) (by observing that the latter contains \(V(G_{k-1})\) and by using the inductive the hypothesis on the subgraph covering) or to \( G_k^{(k)}\) (which contains \(V(K_\chi^{(k)})\) by construction). 
    Consider any node \(v \in V(G_k) \setminus (V(G_{k-1}) \cup V(K_\chi^{(k)}))\). 
    Then \(v = (v_1, v_2, j)\) for \(v_1 \in V(G_{k-1})\), \(v_2 \in V(K_\chi^{(k)})\), and \(j \in [2r]\).
    By the inductive hypothesis, there exists \(i \in [k-1]\) such that \(v_1 \in G_{k-1}^{(i)}\).
    Then, \((v_1, v_2,j) \in V(G_{k-1}^{(i)}) \times V(K_\chi^{(k)}) \times \{j\}\).
    As \(G_{k-1}^{(i)}\) and \(K_\chi^{(k)}\) contain no isolated nodes (inductive hypothesis (ii)), there exist \(u_1 \in V(G_{k-1}^{(i)})\) adjacent to \(v_1\) and \(u_2 \in V(K_\chi^{(k)})\) adjacent to \(v_2\).
    The path \(w_0w_1\dots w_j\) defined by \(w_0 = u_1\), \(w_k = (v_1, v_2, k)\) for \(1 \le k \le j\) odd, \(w_k = (u_1, u_2, k)\) for \(2 \le k \le j\) even connects \(u_1\) to \(v\) if \(j\) is odd.
    The path \(w_0w_1\dots w_j\) defined by \(w_0 = v_1\), \(w_k = (u_1, u_2, k)\) for \(1 \le k \le j\) odd, \(w_k = (v_1, v_2, k)\) for \(2 \le k \le j\) even connects \(v_1\) to \(v\) if \(j\) is even.
    Hence, \(v \in \NN_j(G_{k-1}^{(i)})\).
    If \(1 \le j \le T+2\), then \(u \in V(G_{k}^{(i)})\).
    If, instead, \(T+2 \le j \le 2r\), then \(u \in V(G_k^{(k)})\). 

    Now consider any two nodes \(u,v\) which are connected in \(G_k\). 
    If \(u,v \in V(G_{k-1})\) or \(u,v \in V(K_\chi^{(k)})\) we have that \(\{u,v\} \in E(\cup_{i \in [k-1]} G_k^{(i)})\) or \(\{u,v\} \in E(G_k^{(k)})\), respectively.
    Suppose \(u \in V(G_{k-1})\) but \(v \notin V(G_{k-1}) \). 
    Then, \(v = (v_1, v_2, 1)\) for some \(v_1 \in V(G_{k-1})\) and some \(v_2 \in V(K_\chi^{(k)})\).
    As \(u\) and \(v\) are connected, it means that \(uv_1\) is an edge in \(G_{k-1}\). 
    By the inductive hypothesis on the subgraph covering, there exists \(i \in [k-1]\) such that \(uv_1 \in G_{k-1}^{(i)}\). 
    Hence, \(uv \in E(G_k^{(i)})\).
    If \(u \in V(K_\chi^{(k)})\) but \(v \notin V(K_\chi^{(k)})\), \(uv \in E(G_k^{(k)})\).
    Suppose now that \(u,v \notin V(G_{k-1}) \cup V(K_\chi^{(k)})\). 
    Then, there exist \(u_1,v_1 \in V(G_{k-1})\), \(u_2,v_2 \in V(K_\chi^{(k)})\), and \(j_u,j_j \in [2r]\) with \(\abs{j_u - j_v} \le 1\) such that \(u = (u_1, u_2, j_u)\) and \(v = (v_1, v_2, j_v)\).
    Furthermore, as \(uv\) is an edge of \(G_k\), it holds that \(u_1u_2\) is an edge in \(G_{k-1}\), and \(v_1v_2\) is an edge in \(K_\chi^{(k)}\).
    From the inductive hypothesis on the subgraph covering, there exists \(i \in [k-1]\) such that \(G_{k-1}^{(i)}\) contains \(u_1u_2\).
    We have that \(u,v \in V(G_k [{V(G_{k-1}^{(i)}) \times V(K_\chi^{(k)}) \times \{j_u, j_v\}}])\).
    Let \(j = \max(j_u, j_v) \)
    Then, \(uv \in E(G_k [{\NN_{j}(G_{k-1}^{(i)})}])\).
    Hence, if \(j \le T+2\), \(uv \in E(G_k^{(i)})\). 
    If, instead, \(T+1\le j \le 2r\), \(uv \in E(G_k^{(k)})\).

    \paragraph{Property (i).}
    Consider \(G_k^{(i)}\) for \(i < k\). 
    The function \(f_i = \proj{G_{k-1}} \restriction_{\NN_T(G_k^{(i)})} \) is a homomorphism from \(G_k [{\NN_T(G_k^{(i)})}]\) to \(G_{k-1} [{\NN_T(G_{k-1}^{(i)})}]\). 
    As \(G_{k-1} [{\NN_T(G_{k-1}^{(i)})}]\) is \(\chi\)-colorable by the inductive hypothesis (i), so it is \(G_k [{\NN_T(G_k^{(i)})}]\).
    Since \(G_k [{\NN_T(G_k^{(i)})}]\) contains \(K_\chi^{(i)}\) as a subgraph, \(\chi\) colors are also necessary.

    Similarly, consider \(G_k [{\NN_T(G_k^{(k)})}]\). 
    Then, \(f_{k} = \proj{K_\chi^{(k)}} \restriction_{[\NN_T(G_k^{(k)})]} \) is a homomorphism from \(G_k [{\NN_T(G_k^{(k)})}]\) to \(K_\chi^{(k)}\).
    Hence, \(G_k [{\NN_T(G_k^{(k)})}]\)  is \(\chi\)-colorable. 
    As \(G_k [{\NN_T(G_k^{(k)})}]\)  contains \(K_\chi^{(k)}\), its chromatic number is \(\chi\).

    \paragraph{Property (ii).}
    Fix \(i \in [k-1]\).
    Observe that \(G_{k}^{(i)} = G_{k-1}^{(i)} \rjoin{2r} K_\chi^{(k)} [{\NN_{T+1}(G_{k-1}^{(i)})}] \).
    As both \(G_{k-1}^{(i)}\) (by the inductive hypothesis (ii)) and \(K_\chi^{(k)}\) are connected, their join is connected.
    Then, \cref{lem:graph:connected-induced-subgraph} implies that \(G_{k}^{(i)}\) is connected.
    The same applies for  \(G_{k}^{(k)}\) by observing that \(G_{k}^{(i)} = G_{k-1}^{(i)} \rjoin{2r} K_\chi^{(k)} [{\NN_{2r - T}(K_\chi^{(k)})}] \).
   
    \paragraph{Property (iii).}
    Consider any node \(v \in V(G_{k-1}) \cup V(G_{k-1}) \times V(K_\chi^{(k)}) \times \{1, \dots, T+1\} \).
    Let \(v ' = \proj{G_{k-1}}({v})\).
    By the inductive hypothesis (iii), the set of nodes \(\NN_1(v') \cap V(G_{k-1})\) is contained in some \(G_{k-1}^{(i)}\) for \(i \in [k-1]\).
    By definition of \(G_{k}^{(i)}\), it follows that \(\NN_1(v) \subseteq V(G_{k}^{(i)})\).
    Now, consider any node \(v \in V(K_\chi^{(k)}) \cup V(G_{k-1}) \times V(K_\chi^{(k)}) \times \{T+2, \dotsm 2r\} \).
    By definition of \(G_k^{(k)}\), we have that \(\NN_1(v) \subseteq G_k^{(k)}\).

    \paragraph{Property (iv).} 
    Fix \(1 \le i \neq j \le k-1\). 
    By the inductive hypothesis (iv), there exists a node \(u \in V(G_{k-1}[ \NN_T (G_{k-1}^{(i)})])\) from \(G_{k-1}^{(i)}\).    
    By definition of \(G_k^{(i)}\), we have that \(u\) has distance \(T\) from \(G_k^{(i)}\).
    Furthermore, observe that any node in \(V(G_{k-1}) \times V(K_\chi^{(k)}) \times {1}\) has distance \(T\) from \(G_k^{k}\), concluding the proof.
\end{proof}

\begin{remark}
    We highlight that the subgraph cover from \cref{lem:gadget:partition} is not unique: other families of graphs could be used obtaining a shorter and simpler proof at the expense of a higher number of graphs in the family (e.g., a number of subgraphs that is exponential in \(k\)).
    The latter would result in a worse bound on \(T\) in the next corollary.
\end{remark}

We are now ready to show that the family of \(\chi\)-chromatic graphs admits a cheating graph for the \(c\)-coloring graphs problem.

\begin{corollary}\label{cor:approx-graph-col:cheatgraph-family}
    Let \(\chi \ge 2\), \(c \ge \chi\) be integers, and \(k = \floor{\frac{c-1}{\chi - 1}}\).
    Consider \(\FF\) to be the family of all connected \(\chi\)-chromatic graphs, and \(\PP\) to be the problem of \(c\)-coloring graphs.
    For every \(N \ge 1\) and \(n \ge ((6\chi + 1)^{k+1} - 1)N/6\), there exists a value \(T\) with
    \[
        T = \myTheta{ \frac{1}{\chi^{1 + \frac{1}{k}}} \left( \frac{n}{N} \right) ^{ \frac{1}{k}} }  
    \]
    such that \(\FF\) admits an \((n,k+1, N, T)\)-cheating graph for \(\PP\).
\end{corollary}
\begin{proof}
    There exists a unique integer \(r \ge 3\) such that
    \[
        \frac{(2r\chi + 1)^{k+1} - 1}{2r} \le \frac{n}{N} < \frac{(2r\chi + 2\chi + 1)^{k+1} - 1}{2r + 2}.
    \]

    We claim that the graph \(G_{k+1}\) defined in \cref{def:graph:gadget} by iterating \(2r\)-join operations is an \((n,{k+1},N,T)\)-cheating graph for \((\PP,\FF)\).
    Clearly, \(\PP\) is not solvable on \(G_{k+1}\), while property \cref{def:cheating-graph}.(ii).(a) follows by \cref{lem:gadget:partition}.(iii) (since coloring is an LVL problem with checking radius \(t = 1\)).

    We now prove that \cref{def:cheating-graph}.(ii).(b) holds as well.
    Consider \(N\) copies \(G_{{k+1}, 1}. \dots. G_{{k+1}, N}\) of the graph \(G_{k+1}\).
    For each \(j \in [N]\), \cref{lem:gadget:partition} gives us a subgraph cover \(\{G_{{k+1},j}^{(i)}\}_{i \in [{k+1}]}\) of \(G_{k,j}\) with properties \cref{lem:gadget:partition}.(i) and (iv) verified for \(T = \floor{\frac{2r}{3}}\).
    Notice that 
    \[
        T = \myTheta{ \frac{1}{\chi^{1 + \frac{1}{k}}} \left( \frac{n}{N} \right) ^{ \frac{1}{k}} }  
        .
    \]

    For any choice of indices \(\xx_N = (x_1, \dots, x_N)\in [{k+1}]^N\), we now show that there exists a connected graph \(H_{\xx_N} \in \FF\) on \(n\) nodes that admits a subgraph \(\tilde{H}_{\xx_N}\) such that
    \begin{enumerate}[(1)]
        \item \(H_{\xx_N} [{\NN_T(\tilde{H}_{\xx_N})}]\) is isomorphic to \( \bigsqcup_{j \in [N]} G_{{k+1},j} [{\NN_T(G_{{k+1},j}^{(x_j)})}] \);
        \item \(\XX(H_{\xx_N}) = \chi\).
    \end{enumerate} 
    The vertex set \(V(H_{\xx_N})\) is \(V(\sqcup_{j \in [ N]} G_{{k+1},j})\) together with \(n - N \cdot \frac{(2r\chi + 1)^{k+1} - 1}{2r} \le \frac{n}{N}\) extra nodes.
    We take \(H_{\xx_N}\) to be the disjoint union of \(G_{{k+1},j} [{\NN_T(G_{{k+1},j}^{(x_j)})}]\) for all \(j \in [N]\) where, for each \(j \in [N-1]\) we add an edge between a node \( v_j \in G_{{k+1},j} [{\NN_T(G_{{k+1},j}^{(x_j)})}]\) and a node \(  v_{j+1} \in G_{{k+1},j+1} [{\NN_T(G_{{k+1},j+1}^{(x_{j+1})})}]\) such that \(\dist(v_j, G_{{k+1},j}^{(x_j)}) = T\) and \(\dist(v_{j+1}, G_{{k+1},{j+1}}^{(x_{j+1})}) = T\) (such nodes exist because of \cref{lem:gadget:partition}.(iv)).
    All remaining nodes form a path of which one endpoint is connected to any node in \(G_{{k+1},N} [\NN_T({G_{{k+1},N}^{(x_N)}})]\) at distance \(T\) from  \({G_{{k+1},N}^{(x_N)}}\).
    Property (1) follows by construction.
    As for property (2), we observe that the chromatic number of \(H_{\xx_N}\) is still \(\chi\) as each component \(\NN_T(G_{{k+1},j}^{(x_j)})\) is \(\chi\)-chromatic by \cref{lem:gadget:partition}.(i), and \(\chi \ge 2\).
    Furthermore, \(H_{\xx_N}\) is connected by \cref{lem:gadget:partition}.(ii) combined with \cref{lem:graph:connected-induced-subgraph} and observing that connected disjoint connected components through paths.
\end{proof}

The proof of our main lower bound now follows easily.

\begin{proof}[Proof of \cref{thm:lb-coloring}]
    Let \(k = \alpha\) and
    \[
        N = \ceil*{\frac{\log \frac{1}{\varepsilon}}{\log \left(1 + \frac{1}{k}\right)}}. 
    \]
    By \cref{cor:approx-graph-col:cheatgraph-family}, \(\FF\) admits an \((n,{k+1},N,T)\)-cheating graph for the \(c\)-coloring graphs problem for any \(n \ge ((6\chi+1)^{k+1}-1)N/6\) and for some 
    \[
        T = \myTheta{\frac{1}{\chi^{1 + \frac{1}{k}}} \cdot \left(\frac{n \log ( 1 + \frac{1}{k})}{\log \frac{1}{\varepsilon}}\right)^{\frac{1}{k}}}.    
    \]
    \cref{thm:lb-technique} implies that there is a connected graph \(H \in \FF\) on \(n\) nodes such that the probability that any outcome \(\outcome\) with locality \(T\) is \(c\)-coloring \(H\) is at most \((1-1/(k+1))^N\). 
    We remind the reader that the graph can be chosen as in \cref{def:cheating-graph}.(ii).(b), and \cref{cor:approx-graph-col:cheatgraph-family} implies that such graph is connected.
    By definition of \(k\) and \(N\), this probability is at most \(\varepsilon\). 
    By observing that \(\log^{\frac 1k} ( 1+ \frac 1k) = \Theta(1)\), we get the thesis.      
\end{proof}

\subsubsection{\texorpdfstring{Notes on the graph \(G_k\)}{Remarks about the graph Gk}}\label{sec:remarks-graph-theoretic-constr}
This section is dedicated to the reader that is interested in the topological elements underlying the construction of the graph \(G_k\) in \cite{bogdanov2013}.
Familiarity with the notion of topological space, join of topological spaces, homotopy equivalence, abstract simplicial complex, and geometric realization of an abstract simplicial complex is required. 
Every graph \(G\) can be associated with an abstract simplicial complex \(\nc{G}\) called the \emph{neighborhood complex}, defined as follows:
\[
    \nc{G} \coloneqq \left\{ A \subseteq V(G) \ \st \ \exists v \in V(G) \text{ such that } \forall u \in A, v \in \NN(u)  \right\},
\]
that is, \(\nc{G}\) consists in all subsets \(A \subseteq V\) of nodes that have a common neighbor.
Let us denote the geometric realization of \(\nc{G}\) by \(\norm{\nc{G}}\) (the geometric realization is unique up to homeomorphisms).
We say that a non-empty topological space \(X\) is \(m\)-connected if each continuous map \(\pi : S^{i-1} \to X\) extends to a continuous map \(\bar{\pi} : D^i \to X \) for each \(i = 1, \dots, m\), where \(S^{i-1} = \left\{x \in \real^i \ \st \ \norm{x}_2 = 1  \right\}\) is the \((i-1)\)-dimensional sphere, and \(D^i = \left\{x \in \real^i \ \st \ \norm{x}_2 \le 1  \right\}\) is the \(i\)-dimensional disk.

\textcite{lovasz1978} proved the following theorem.
\begin{theorem}[\cite{lovasz1978}]\label{thm:lovasz}
    Let \(G\) be any graph such that \(\nc{G}\) is non-empty.
    If \(\norm{\nc{G}}\) is \(m\)-connected, then \(\XX(G) \ge m+3\).
\end{theorem}

\citeauthor{lovasz1978}'s result provides a clear and effective tool to bound from below the chromatic number of a graph.
\textcite{bogdanov2013} linked the \(r\)-join operation between graphs to the join operation of topological spaces, proving the following lemma.

\begin{lemma}[\cite{bogdanov2013}]\label{lem:bogdanov}
    Let \(G,H\) be any two graphs, and \(r \in \natPos\).
    Then, \(\norm{\nc{G \rjoin{r} H}} \simeq \norm{\nc{G}} \star \norm{\nc{H}}\), where \(\simeq\) means homotopy equivalent and the latter \(\star\) operator represents the join operation between topological spaces.
\end{lemma}
Note that, for any three topological spaces \(A\), \(A'\), and \(B\), if \(A \simeq A'\), then \(A \star B \simeq A' \star B\) \cite{matousek2010}.
\cref{thm:lovasz,lem:bogdanov} can be combined to construct graphs of high chromatic number.
E.g., if \(G = K^{(1)}_\chi\) and \(H = K^{(2)}_\chi\), that is, two disjoint copies of the complete graph of \(\chi\) nodes, we know that \(\norm{\nc{G}} \simeq \norm{\nc{H}} \simeq S^{\chi - 2} \), hence \(\norm{\nc{G \rjoin{r} H}} \simeq S^{\chi -2} \star S^{\chi -2} \simeq S^{2\chi - 3}\).
As \(S^{2\chi - 3}\) is \((2\chi - 4)\)-connected, then \(\XX(G \rjoin{r} H) \ge 2\chi - 1\).
By applying recursively the above procedure, one can understand why the graph \(G_k\) from \cref{def:graph:gadget} has chromatic number bounded from below by \(k(\chi - 1) + 1\).

As for the properties of the local chromatic number, \textcite{bogdanov2013} also showed the following simple result, which can be proved by using the projections \(\proj{G}\) and \(\proj{H}\), that immediately gives the desired result.

\begin{lemma}[\cite{bogdanov2013}]\label{lem:bogdanov:lclCHR}
    Let \(G,H\) be two graphs.
    For every \(r \in \natPos\) it holds that 
    \[
        \lclCHR{r}{(G \rjoin{2r} H)} = \min \{\lclCHR{r}(G), \lclCHR{r}(H)\}.  
    \]
\end{lemma}

\subsubsection{On highly chromatic graphs with small local chromatic number}\label{ssec:graph-theory-related-works}
In this section we present previous works studying the relation between the chromatic number and the local chromatic number of a graph.
We remind to the reader that the local chromatic number \(\lclCHR{r}(G)\) of radius \(r\) of a graph \(G\)  is the maximum number of colors required to color any \(r\)-ball in \(G\).
We remark that, in this paper, we depart from the standard definition of local chromatic number in the literature that started with \cite{erdos1986}, which is not useful for our purposes.
For positive integers \(c, \chi, r\), define \(f_\chi(c,r)\) to be the maximal integer \(n\) such that every graph \(G\) of \(n\) nodes with \(\lclCHR{r}(G) \le \chi \) is \(c\)-colorable.
The first bounds on \(f_\chi(c,r)\) held for specific values of \(\chi\), especially \(\chi \in \{2,3\}\) \cite{erdos1959,stiebitz1985,berlov2012} (for a summary of the progress of such results we defer the reader to \cite{bogdanov2014,alon2016}).
We only describe the results that consider a generic \(\chi\), as they are more directly related to our work. 
\textcite{kierstead1984} proved that, for each \(k \in \nat\), 
\begin{align}
    f_\chi(k(\chi-1)+1, r) \ge \floor{r / (2k)}^k , \label{eq:kierstead1984}
\end{align}
i.e., any graph \(G\) having \(\lclCHR{r}(G) \le \chi\) and \(\abs{V} \le \floor{r / (2k)}^k\) has chromatic number at most \(k(\chi - 1) + 1\).
Thirty years later, \textcite{bogdanov2013} proved that this result is basically tight when \(k,\chi\) are fixed: it showed that, for each \(k \ge 2\), 
\begin{align}
    f_\chi(k(\chi-1),r) \le \frac{(2r\chi + 1)^k - 1}{2r}, \label{eq:bogdanov2013}
\end{align}
i.e., there exists a graph  \(G_k\) with \(\lclCHR{r}(G_k) = \chi\), \(\XX(G_k) \ge k(\chi - 1) + 1\), and \(\abs{V(G_k)} = \frac{(2r\chi)^{k} - 1}{2r}\):
such result is built upon a lemma  by \textcite{lovasz1978}.
The two above results show an interesting phenomenon.
For constant positive integers \(\chi\) and \(c\), the minimum number of vertices of a graph \(G\) with \(\lclCHR{r}(G) \le \chi\) and \(\chi(G) = c\) is roughly \( r ^{\floor{\frac{c-1}{\chi - 1}}}\), that is, it jumps to the powers of \(r\) where the exponents are the values \(c\) congruent to \(1\) modulo \(\chi - 1\).

The estimate by \textcite{kierstead1984} breaks when \(k \gtrsim r\).
\textcite{bogdanov2014} and \textcite{alon2016} investigated and provided better lower bounds to \(f_\chi(c,r)\) when \(c \gtrsim (\chi - 1)r\);
\textcite{alon2016} estimated also upper bounds to \(f_\chi(c,r)\) which are roughly tight for fixed \(r\).
We do not discuss such results in details as they are not useful for our purposes:
we are interested in the asymptotic relation between the locality radius \(r\) and the number of nodes \(n\) given the local and global chromatic number.
Such case is covered by \cref{eq:bogdanov2013}: in particular, we used the example graph that was built in \cite{bogdanov2013} as our baseline for the lower bound proof.

We remark that our result for the specific case \(\chi = 2\) could also be achieved through other graphs that were studied in the literature, e.g., the \emph{generalized Mycielski graph} (studied by \textcite{stiebitz1985}, see \cite{gyarfas2004} for an English version of the proof).
 \subsection{\texorpdfstring{No quantum advantage for \(3\)-coloring grids}{No quantum advantage for 3-coloring grids}}\label{sec:coloring:grids}

In this section we prove that \(3\)-coloring \(\gridW\times \gridH\) grids of requires time \(\myOmega{\min(\gridW, \gridH)}\) rounds in the \nslcl model, by using the same graph-theoretical lower bound argument of \cref{sec:lb-technique}. 

For any two integers \(a \le b\), we denote the set \(\{a, a+1, \dots, b\}\) by \([a:b]\).

\begin{definition}[KB-gadget]\label{def:kb-quadrangulation}
    Consider a graph \(H_{\gridW, \gridH} = (V(H_{\gridW, \gridH}),E(H_{\gridW, \gridH}))\) of \((\gridW+1)(\gridH+1)\) nodes, where we label the nodes by using coordinates from the set \([0:\gridW]\times [0:\gridH]\);
    two nodes \((i,j)\) and \((i',j')\) are connected by an edge if \(\abs{i - i'} + \abs{j - j'} = 1\), i.e., their Manhattan distance is 1.
    Now, define the following equivalence relation for nodes:
    \((i,0) \sim_V (\gridW-i,\gridH)\) for \(i = 0, \dots, \gridW\), and \((0,j) \sim_V (\gridW, j)\) for \(j = 0, \dots, \gridH\).
    Now define a new graph \(G_{\gridW, \gridH} = (V(G_{\gridW, \gridH}), E(G_{\gridW, \gridH}))\) where \(V(G_{\gridW, \gridH}) = V(H_{\gridW, \gridH}) \slash {\sim_V}\) and \(E(G_{\gridW, \gridH})\) is characterized as follows: 
    Take any two nodes \(u,v\in V(G_{\gridW, \gridH})\).
    Such nodes are equivalence classes for \(\sim_V\). 
    If there exists \(u' \in u \subseteq V(H_{\gridW, \gridH}), v' \in u \subseteq V(H_{\gridW, \gridH})\) such that \(u'v' \in E(H_{\gridW, \gridH})\), then \(uv \in E(G_{\gridW, \gridH})\) (see \cref{fig:kb-quadrangulation}).
    We name \(G_{\gridW, \gridH}\) an \(\gridW \times \gridH\) \emph{Klein bottle gadget} (KB-gadget).
\end{definition}

Now, clearly, \(G_{\gridW, \gridH}\) is everywhere locally grid-like; however, we claim that \(\XX(G) \ge 4\) if \(\gridW\) is odd.

\begin{figure}
    \centering
    \includegraphics[scale=0.5]{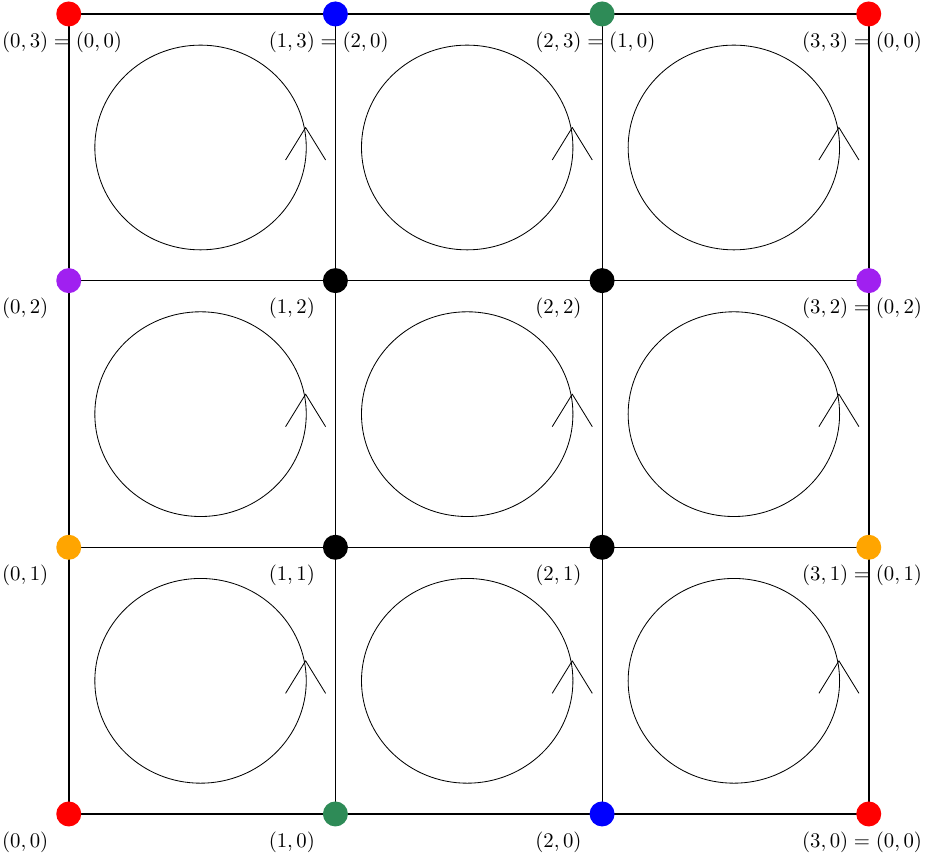}
    \caption{Representations of graphs \(G_{\gridW, \gridH}\) and \(H_{\gridW, \gridH}\) with \(\gridW = \gridH = 3\).
    Nodes on the borders are identified in \(G_{\gridW, \gridH}\) as indicated by the same colors.
    Each face of the graph is oriented the same as the oriented circle inside it.}
    \label{fig:kb-quadrangulation}
\end{figure}

\paragraph{Quadrangulation of the Klein bottle.}

To establish the truth of our claim, we make use of results on the chromatic number of quadrangulations of surfaces \cite{archdeacon2001,mohar2002,mohar2013}.
Following the preliminaries of \cite{mohar2013}, by \emph{surface} we mean a compact connected 2 manifold without boundary. 
A \emph{quadrangulation} of a surface is a graph without self-loops on that surface with all faces being quadrilaterals. 
For our purpose, we can restrict to the class of simple graphs.
Let \(G\) be a quadrangulation of a surface \(S\).
An \emph{orientation} of a face of \(G\) is a closed walk along its boundaries.
Given two oriented faces sharing at least one boundary edge \(e\), we say that the consistency of the orientation is \emph{not broken} at \(e\) if we traverse edge in opposite directions when considering the two orientations of the faces.
Given an arbitrary orientation of all faces we can count the edges that break consistency of the orientation of the surface.
Note that reversing the orientation of a face changes the status of its four edges, and thus the parity of the number of edges breaking consistency does not change. 
We say that \(G\) is \emph{even} or \emph{odd} depending on this parity. 
By the previous remark, the parity of the quadrangulation is determined by any orientation of all faces and is invariant, i.e., it depends only on the graph \(G\).

The following theorem was independently proved by \cite{archdeacon2001,mohar2002}, but we report its formulation by \cite{mohar2013}.

\begin{theorem}[\cite{mohar2013}]\label{thm:odd-quadrangulation}
    Let \(G\) be an odd quadrangulation of some surface \(S\). 
    Then, \(\XX(G) \ge 4\).
\end{theorem}

The Klein bottle can be defined as the quotient space of the square \([0,1] \times [0,1]\) over the equivalence relation \(\sim\) that identifies sides as follows:
\((x,0) \sim (x,1)\) for \(x \in [0,1]\), and \((0,y) \sim (1,1-y)\) for \(y \in [0,1]\).

\begin{lemma}\label{lem:kb-gadget-chr-num}
    Let \(\gridH \ge 2\). For any odd integer \(\gridW \ge 3\), \(\XX(G_{\gridW, \gridH}) \ge 4\) and \(\XX(G_{\gridH, \gridW}) \ge 4\).
\end{lemma}
\begin{proof}
    Notice that \(G_{\gridH, \gridW}\) is isomorphic to \(G_{\gridW, \gridH}\), hence we focus on the latter.
    It is easy to see that \(G_{\gridW, \gridH}\) is a quadrangulation of the Klein-bottle. 
    Indeed, \(H_{\gridW, \gridH}\) is a quadrangulation of the square \([0,1]\times [0,1]\): it suffices to map the nodes \((i,j)\) into \((i/\gridW,j/\gridH)\), for all \(i,j\). 
    As \(G_{\gridW, \gridH}\) is obtained by node identification through the relation \(\sim_V\), which acts on the border of the square, we have a quadrangulation of the Klein-bottle.
    Let's orient the faces of \(G_{\gridW, \gridH}\) as in \cref{fig:kb-quadrangulation}. 
    Even though in \(H_{\gridW, \gridH}\) the orientation is even as consistency is never broken, through simple observations, one can verify that \(G_{\gridW, \gridH}\) is an odd quadrangulation if and only if \(n\) is odd, as the orientation consistency is broken only at edges belonging to the set \(\left\{ \{(i,0),(i+1,0)\} \ \st \ 0 \le i \le \gridW-1 \right\}\).
    Then \cref{thm:odd-quadrangulation} gives the thesis.
\end{proof}

It suffices to construct a suitable subgraph cover for \(G_{\gridW, \gridH}\) which shows that the family of grid graphs admits cheating graphs for the 3-coloring problem. 
First, let us denote by \(Q_{\gridW, \gridH}\) any subgraph of the infinite two-dimensional lattice that is isomorphic to an \((\gridW+5)/2 \times (\gridH+5)/2\) grid, and by \(Q_{\gridW, \gridH}^{T}\) the graph induced by its \(T\)-neighborhood in the lattice.

\begin{lemma}\label{lem:kb-gadget-partition}
    Let \(5 \le \gridW, \gridH\) be odd integers, and let \(G_{\gridW, \gridH}\) be the graph defined in \cref{def:kb-quadrangulation}.
    Let \(T = \floor{\frac{\min(\gridW, \gridH)-5}{4}}\).
    There exists a subgraph cover \(\{G_{\gridW, \gridH}^{(i)}\}_{i\in [4]}\) of \(G_{\gridW, \gridH}\) such that the following statements hold:
    \begin{enumerate}[(i)]
        \item  The chromatic number of \(G_{\gridW, \gridH}[\NN_T(G_{\gridW, \gridH}^{(i)})]\) is 2 for \(i \in [4]\);
        \item \(G_{\gridW, \gridH}[\NN_T(G_{\gridW, \gridH}^{(i)})]\) is isomorphic to \(Q_{\gridW, \gridH}^{T}\);
        \item For each \(v \in V(G_{\gridW, \gridH})\), there exists \(i \in [4]\) such that \(\NN_1(v) \subseteq V(G_{\gridW, \gridH}^{(i)})\).
    \end{enumerate}
\end{lemma}
\begin{proof}
    Observe that claim (ii) implies claim (i):
    hence, we prove claim (ii).
    Consider the graph \(H_{\gridW, \gridH}\) used to construct the KB-gadget in \cref{def:kb-quadrangulation}.
    Let \(V_1\) be the set of nodes \((i,j)\) of \(H_{\gridW, \gridH}\) respecting the following property
    \begin{align*}
        V_1: \ i \in \left[0:\frac{\gridW+1}{2}\right] \cup [\gridW-1:\gridW], j \in \left[0: \frac{\gridH+1}{2}\right] \text{ or }
        i \in [0:1]\cup\left[\frac{\gridW-1}{2}:\gridW\right], j \in [\gridH-1:\gridH].
    \end{align*}
    Similarly, we define \(V_2,V_3,V_4\) by the following properties
    \begin{align*}
        V_2 : & \ i \in [0:1] \cup \left[\frac{\gridW-1}{2}:\gridW\right], j\in \left[0:\frac{\gridH-1}{2}\right] \text{ or } i \in \left[0:\frac{\gridW+1}{2}\right] \cup [\gridW-1:\gridW], j \in [\gridH-1:\gridH]; \\
        V_3 : & \ i \in \left[0:\frac{\gridW+1}{2}\right] \cup [\gridW-1:\gridW], j \in [0:1] \text{ or } i \in [0:1]\cup\left[\frac{\gridW-1}{2}:\gridW\right], j \in \left[\frac{\gridH-1}{2}:\gridH\right]; \\
        V_4 : &  \ i \in [0:1]\cup \left[\frac{\gridW-1}{2}:\gridW\right], j \in [0:1] \text{ or } i \in \left[0: \frac{\gridW+1}{2}\right]\cup[\gridW-1:\gridW], j \in \left[\frac{\gridH-1}{2}:\gridH\right].
    \end{align*}
    For an example of \(V_1,V_2\) see \cref{fig:KB-cover}.
    Then, define \(S_i\) to be subset of \(V(G_{\gridW, \gridH})\) induced by \(V_i\) after applying the equivalence relation \(\sim_V\) in \(V(H_{\gridW, \gridH})\), and \(G_{\gridW, \gridH}^{(i)} = G_{\gridW, \gridH}[S_i]\).
    Notice that \(G_{\gridW, \gridH}^{(i)}\) is isomorphic to \(Q_{\gridW, \gridH}\) for each \(i \in [4]\).
    Furthermore, the \(T\)-view of \(G_{\gridW, \gridH}^{(i)}\) identifies in \(H_{\gridW, \gridH}\) four graphs that do not intersect, implying claim (ii).
    Claim (iii) is trivial.
\end{proof}
\begin{remark}
    The strange shape of the subgraph cover is needed for property (iii). 
    However, the analysis carries on by considering ``simpler'' subgraph covers which don't meet property (iii) as a valid coloring can be checked by looking at single edges, not necessarily entire neighborhoods.
    Nevertheless, we choose to use the general result for LVLs (\cref{thm:lb-technique}). 
\end{remark}

\begin{figure}
  \centering
  \begin{subcaptionblock}{\textwidth}
    \centering
    \includegraphics[scale=0.7]{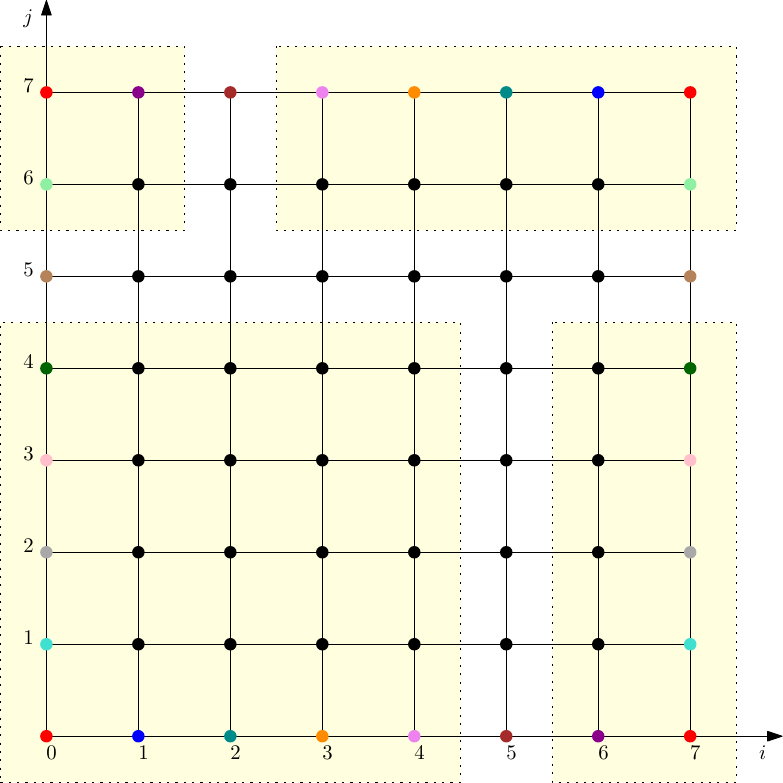}
    \caption{\(V_1\) (marked in yellow).}
  \end{subcaptionblock}
  \\[5mm]
  \begin{subcaptionblock}{\textwidth}
    \centering
    \includegraphics[scale=0.7]{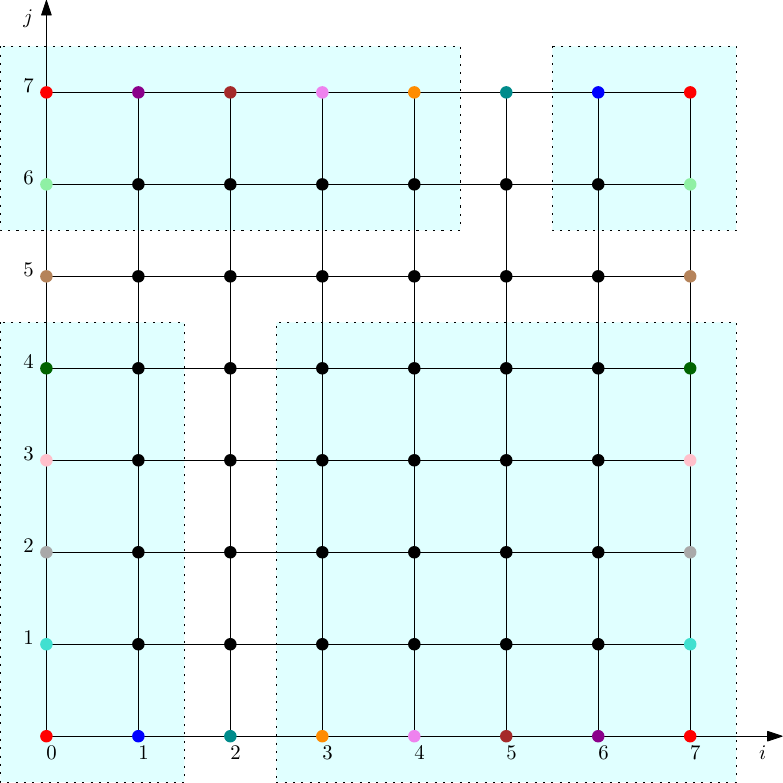}
    \caption{\(V_2\) (marked in blue).}
  \end{subcaptionblock}
\caption{Two elements of the subgraph cover of the KB-gadget for \(\gridW = \gridH = 7\).
  Nodes with identical colors are identified.
  \(V_3\) and \(V_4\) can be obtained by symmetries.}
  \label{fig:KB-cover}
\end{figure}

\cref{lem:kb-gadget-partition} implies that the family of  grid graphs admits cheating graphs for the \(3\)-coloring graphs problem.

\begin{corollary}\label{cor:grid:cheatgraph-family}
    Let \(N \ge 1\), \(\gridW, \gridH \ge 5N\), and 
     \(\FF\) be the family of all grids of size \(\gridW \times \gridH\).
    There exist a \(T = T(\gridW, \gridH)\), with 
    \[
        T = \myTheta{ \frac 1N \cdot \min(\gridW, \gridH)},  
    \]
    such that \(\FF\) admits an \((\gridW\cdot \gridH, 4,N,T)\)-cheating graph for the \(3\)-coloring graph problem.
\end{corollary}
\begin{proof}
    Let \(M_1 = \floor{\frac \gridW N}\) if \( \floor{\frac \gridW N}\) is odd, otherwise \(M_1 =  \floor{\frac \gridW N} - 1\).
    Similarly, let \(M_2 =  \floor{\frac \gridH N}\) if \(\floor{\frac \gridH N}\) is odd, otherwise \(M_2 = \floor{\frac \gridH N} - 1\).
    Consider the graph \(G_{M_1,M_2}\) and the subgraph cover given by \cref{lem:kb-gadget-partition}.
    Property (i) of \cref{def:cheating-graph} is clearly satisfied as \(\XX(G_{M_1,M_2}) \ge 4\).
    Property (ii).(a) is implied by \cref{lem:kb-gadget-partition} as the checking radius of the coloring problem is \(t = 1\).
    Consider now a vertex \(v \in V(G_{M_1,M_2})\).
    Let us now prove property (ii).(b).
    Consider the subgraph family \(\{G_{M_1, M_2}^{(i)}\}_{i \in [4]}\) given by \cref{lem:kb-gadget-partition}:
    then, \(T = \floor{\frac 14 \cdot (\min (M_1, M_2) - 5)}\).
    Consider any choice of indices \(\xx_N \in [4]^N\). 
    Notice that \(G_{M_1, M_2} [ \NN_T ( G _ {M_1, M_2} ^{(i)})]\) is 
    isomorphic to \(Q^T_{M_1, M_2}\) which is, in turn, isomorphic to a proper 
    subgraph of an \(M_1 \times M_2\)-grid.
    Hence, it is always possible to construct an \(\gridW \times \gridH\) grid which respects property (ii).(b) of \cref{def:cheating-graph}.
    The thesis follows by observing that \(T \ge { \frac{1}{4N} \cdot \min(\gridW, \gridH) - 3}\).
\end{proof}

\thmCslLclLbGrids*
\begin{proof}
    By \cref{cor:grid:cheatgraph-family}, there exists a \(T = \myTheta{ \frac 1N \cdot \min(\gridW, \gridH)}\) such that the family \(\FF\) of \(n \times m\) grids admits an \((\gridW \cdot \gridH, 4, N, T)\)-cheating graph for the problem of 3-coloring graphs.
    Suppose there is an outcome \(\outcome\) with locality \(T\) that \(3\)-colors \(\gridW \times \gridH\) grids.
    By \cref{thm:lb-technique} and the choice of \(N\),the probability of \(\outcome\) solving the problem on an \(\gridW \times \gridH\)-grid is at most \(\varepsilon\).
    Hence, \(T = \myOmega{\frac{\min(\gridW, \gridH)}{\log \frac 1\varepsilon}}\).
\end{proof}
 
\subsection{\texorpdfstring{No quantum advantage for \(c\)-coloring trees}{No quantum advantage for c-coloring trees}}\label{sec:coloring:trees}

\textcite{linial1992} showed that \(c\)-coloring trees requires time \(\myOmega{\log_c n}\) in the classical setting.
In this section we prove that the problem has roughly the same complexity in the \nslcl model, by combining our lower bound technique (\cref{thm:lb-technique}) and the same graph-theoretical argument used in \cite{linial1992}. 

For any graph \(G\), we define the \emph{girth} of \(G\) to be the length of the
shortest cycle contained in \(G\); if \(G\) contains no cycle, then its girth is
infinite.
We denote the girth of a graph \(G\) by \(\girth{G}\).
By \(\diam{G}\) we denote the diameter of the graph \(G\), i.e., the quantity \(\max_{u,v \in V(G)} \dist_G(u,v)\).
The argument followed by \citeauthor{linial1992} uses the following graph-theoretical result.

\begin{lemma}[Ramanujan graphs, \cite{lubotszky1988}]\label{lem:trees:high-girth}
    Let \(p,q\) be two primes that are congruent to 1 modulo 4. 
    If \(p \pmod{q}\) is a perfect square, there exists a \(d\)-regular graph \(G\) on \(n = q(q^2 -1 )\) nodes with \(d = p+1\), which satisfies the following properties:
    \begin{enumerate}
        \item \( \girth{G} > 2 \log_p q \);
        \item \(\diam{G} \le 2 (\log_p n + \log_p 2) + 1\);
        \item \(\XX(G) \ge \frac{1}{2}\sqrt{d}\).
    \end{enumerate}
\end{lemma}

We make use of a result in number theory.

\begin{lemma}[Dirichlet's theorem on arithmetic progressions]\label{lem:dirichlet-progression}
    Let \(a,d\) two coprime positive integers.
    There are infinitely many primes \(p\) of the form \(p = dm+a\) for some integer \(m \ge 0\).
\end{lemma}

We use the above result to prove the following.
\begin{corollary}\label{cor:trees:infinitely-many}
    The following statements hold:
    \begin{enumerate}
        \item There exists a constant \(\lambda\) such that, for each \(n \ge 1\), there is a prime \(p \in [n+1:(1+\lambda)n]\) such that \(p \equiv 1 \pmod{4}\);
        \item If \(p\) is a prime such that \(p \equiv 1 \pmod{4}\), there exist infinitely many primes \(q\) such that \(q \equiv 1 \pmod{4}\) and \(p \pmod{q}\) is a non-zero perfect square;
    \end{enumerate}
\end{corollary}
\begin{proof}
    For the first result, we follow the same approach of \cite{tonellicueto2023}.
    Let \(a \in \nat, b \in \nat_+\).
    Consider the arithmetic progression defined by \(s_n = a+bn\), for \(n \ge 1\).
    Let  
    \[
        \pi_{a,b}(x) = \abs{\left\{p \ \st \ p \le x,\, p \text{ is prime},\, p \equiv a \! \pmod{b} \right\}}.
    \]
    By the prime number theorem for arithmetic progression, we know that 
    \[
          \lim_{n \to +\infty} \frac{\pi_{a,b} (n)}{\frac{n}{\eulerT{b} \log n}} = 1,
    \]
    where \(\eulerT{x}\) is the Euler's totient function.
    Let now
    \[
        \rho_{a, b}(n) = \abs{\left\{ k \le n \ \st \ s_k \text{ is prime} \right\}}.
    \]
    It holds that \(\rho_{a,b}(n) = \pi_{a,b}(a + bn) - \pi_{a,b}(a)\).
    Then,
    \[
        \lim_{n \to +\infty} \frac{\rho_{a,b}(n)}{\frac{a + bn}{\eulerT{b}\log 
        (a + bn)}  -\pi_{a,b}(a)} = \lim_{n \to +\infty} \frac{\rho_{a,b}(n)}{\frac{bn}{\eulerT{b}\log 
        (bn)} }= 1.
    \]
    Let \(\lambda > 0\) be any constant: then,
    \[
        \lim_{x \to +\infty} {\rho_{a,b}((1 + \lambda)x)} - {\rho_{a,b}(x)} = +\infty.
    \]
    Define \(n_\lambda\) to be the smallest natural number such that \({\rho_{a,b}((1 + \lambda)n)} - {\rho_{a,b}(n)} > 0\) for all \(n \ge n_\lambda\).
    Notice that the function \(\lambda \mapsto n_\lambda\) is monotone non-increasing and tends to (actually, reaches) 1  as \(\lambda\) tends to \(+\infty\):
    by choosing \(\lambda\) large enough, one obtains \(n_\lambda = 1\).
    Hence, there is a value \(\lambda = \lambda(a,b)\) such that, for each \(n \ge 1\), there is a prime \(p \in [n+1:(1+\lambda)n]\) that is congruent to \(a\) modulo \(b\).
    Setting \(a = 1\), \(b = 4\) yields the thesis.   

    Let us focus on the second claim: we use the same approach as in \cite{magidin2023}.
    By the law of quadratic reciprocity \cite[Theorem 98]{hardy2008}, \(p \pmod{q}\) is a perfect square if and only if \(q \pmod{p}\) is a perfect square.
    Consider the arithmetic progression \(1 + m(4p)\) for \(m \ge 1\).
    By \cref{lem:dirichlet-progression},there exist infinitely many \(m\) such that \(1 + 4pm\) is a prime.
    Set \(q = 1+4pm\).
    Now, \(q \equiv 1 \pmod{4}\) and \(q \equiv 1 \pmod{p}\), which implies that \(p \equiv 1 \pmod{q}\).
\end{proof}

Through \cref{lem:trees:high-girth,cor:trees:infinitely-many}, we can show that the family of trees admits a cheating graph for the \(c\)-coloring problem.

\begin{lemma}\label{lem:trees:cheatgraph-family}
    Let \(c \ge 2\) be an integer.
    For every \(N \ge 1\) and infinitely many \(n \in \nat\),
    there exists a value
    \[
        T = \myTheta{\log_c \frac n N}
    \]
    such that the family \(\FF\) of all trees of size \(n\) of height at most \(2T+3\) admits an \((n,k = n/N,N,T)\)-cheating graph for the \(c\)-coloring problem.
\end{lemma}
\begin{proof}
    Let \(p\) be the smallest prime congruent to 1 modulo 4 such that \(\frac{1}{2}\sqrt{p+1} > c\).
    \cref{lem:trees:high-girth,cor:trees:infinitely-many} imply that there are infinitely many \(q\) such that there is a \((p+1)\)-regular graph \(G_{p,q}\) on \(q(q^2 - 1)/2\) nodes that respect properties \cref{lem:trees:high-girth}.(1) to (3).  
    Choose any \(q\) that is at least \(p^2\) (notice that this choice is a proof's artifact---smaller values of \(q\) work too).
    Clearly, \(G_{p,q}\) is not \(c\)-colorable as \(\XX(G_{p,q}) \ge \frac 12 \sqrt{p+1} > c\).
    Let \(d = p+1\), \(k = q(q^2 - 1)/2\).
    The subgraph cover we consider is \(\{G_{p,q}[\NN_1(v)]\}_{v \in V(G_{p,q})}\): we enumerate such subgraphs by \(G_1 = G_{p,q}[\NN_1(v_1)], \dots, G_k = G_{p,q}[\NN_1(v_k)]\).
    Let \(T = \floor{\log_p q}-1\): notice that \(T \in [\frac 13 \log _d (k) - 1 : \log_d (k-1)]\).
    As \(2T+2 < \girth{G_{p,q}}\), then \(G_{p,q}[\NN_T(G_i)]\) is a \(d\)-regular tree of height \(T+1\) and, hence, is \(2\)-colorable.
    Furthermore, if \(G_{p,q}[\NN_T(G_i)]\) is rooted in \(v_i\), it has \(d^{T+1} \ge k^ \frac{1}{3}\) leaves at distance \(T+1\) from \(v_i\).
    Let \(\xx_N \in [k]^N\).
    Then, one can always take a tree \(H_{\xx_N}\) on \(n = kN\) nodes of height at most \(2T+3\) such that it contains a subgraph whose \(T\)-neighborhood is isomorphic to the disjoint union \( \sqcup_{i \in [N]}G_{p,q}[\NN_{T}(G_{\xx_i})]\).
    Notice that \(k = n/N\), hence \(T = \myTheta{\log_d n/N}\).
    The thesis follows by observing that \cref{cor:trees:infinitely-many} implies \(p \le (1+\lambda)c^2 \) for some constant \(\lambda\), hence \(c^2 \le d = (1+\lambda){c^2}\).
\end{proof}

Now we are ready to state our final result.

\thmTreesLb*
\begin{proof}
    By contradiction, assume 
    \[
        T = \myTheta{\log _c k} 
    \]
    as given by \cref{lem:trees:cheatgraph-family}, with \(n = Nk\).
    Let \(N = k \log \frac{1}{\varepsilon}\) (hence, \(k \sim \sqrt{n / \log \varepsilon^{-1}} \)).
    By \cref{lem:trees:cheatgraph-family,thm:lb-technique}, there exists a tree \(H\) on \(n\) nodes and height at most \(2T+3\) such that the success probability of any outcome \(\outcome\) with locality \(T\) on \(H\) is at most 
    \[
        \left(1 - \frac{1}{k}\right)^N \le e^{- \frac{N}{k}} \le \varepsilon. 
        \qedhere
    \]
\end{proof}
 \section*{Acknowledgments}

We would like to thank Dennis Olivetti for pointing out the work by \textcite{bogdanov2013} on highly chromatic graphs with small local chromatic number.
We are grateful to Sebastian Brandt and Michele Pernice for investigating and contributing to the understanding of topological properties of some structures used in preliminary versions of our results.
Furthermore, we thank Sameep Dahal for helping dealing with dependencies arisen in the lower bound technique for the non-signaling model, and we thank Darya Melnyk, Shreyas Pai, Chetan Gupta, Alkida Balliu, Amirreza Akbari, Yannic Maus, and all participants of the Distributed Graph Algorithms Workshop (February 2023, Freiburg, Germany) for the helpful discussions during the long development of this work.
We are also thankful to Armin Biere, Austin Buchanan, Bill Cook, Stefano Gualandi, and Bernardo Subercaseaux for discussions and advice related to our computational investigations of local and global chromatic numbers.
We also thank Elie Wolfe, Harry Buhrman and Paolo Perinotti for helpful discussions on the relation between the \nslclFull model and physical principles. 
We made use of e.g.\ Plingeling \cite{biere2017}, Treengeling \cite{biere2017}, Gimsatul \cite{fleury2022}, MapleSAT \cite{liang2016}, and PySAT \cite{ignatiev2018} in our computational experiments, and we also wish to acknowledge CSC – IT Center for Science, Finland, for computational resources.

This work was supported in part by the Research Council of Finland, Grant 333837 and by the German Research Foundation (DFG), Grant  491819048.
François Le Gall was supported by JSPS KAKENHI grants Nos.~JP20H05966,
JP20H04139 and MEXT Q-LEAP grant No.~JPMXS0120319794.
Augusto Modanese is supported by the Helsinki Institute of Information
Technology (HIIT).
Marc-Olivier Renou was supported by INRIA in the Action Exploratoire project DEPARTURE.
Xavier Coiteux-Roy was supported by a Postdoc.Mobility fellowship from the Swiss
National Science Foundation (SNSF).

\printbibliography

\end{document}